\documentclass[reqno,12pt,letterpaper,openany]{amsbook}
\usepackage{amsmath,amsfonts,amssymb,amsthm,graphicx,mathrsfs,url}
\usepackage[usenames,dvipsnames]{color}
\usepackage[colorlinks=true,linkcolor=Blue,citecolor=Blue]{hyperref}
\usepackage{slashed}
\usepackage{color}
\usepackage{autonum}
\usepackage[top=1in,left=.8in,right=.8in,bottom=1in]{geometry}
\usepackage{environ}
\usepackage{floatrow}
\usepackage{tikz}
\usepackage{bm}
\usepackage[english]{babel}

\usepackage{scalerel,stackengine}
\stackMath
\newcommand\widecheck[1]{%
\savestack{\tmpbox}{\stretchto{%
  \scaleto{%
    \scalerel*[\widthof{\ensuremath{#1}}]{\kern-.6pt\bigwedge\kern-.6pt}%
    {\rule[-\textheight/2]{1ex}{\textheight}}
  }{\textheight}%
}{0.5ex}}%
\stackon[1pt]{#1}{\scalebox{-1}{\tmpbox}}%
}

\newcommand{\R}{\mathbb{R}}
\newcommand{\C}{\mathbb{C}}
\newcommand{\Z}{\mathbb{Z}}
\newcommand{\N}{\mathbb{N}}
\newcommand{\Ww}{\mathbb{W}}

\newcommand{\w}{{\omega}}
\newcommand{\az}{{\alpha}}
\newcommand{\te}{{\theta}}
\newcommand{\vp}{{\varphi}}
\newcommand{\epsi}{{\varepsilon}}

\newcommand{\SSS}{\mathcal{S}}

\newcommand{\PP}{\mathcal{P}}

\newcommand{\EE}{\mathcal{E}}
\newcommand{\AAA}{\mathcal{A}}
\newcommand{\OO}{\mathcal{O}}

\newcommand{\LL}{\mathcal{L}}
\newcommand{\DD}{\mathcal{D}}

\newcommand{\UU}{\mathcal{U}}
\newcommand{\FF}{\mathcal{F}}
\newcommand{\TT}{\mathcal{T}}

\newcommand{\fE}{\frak{E}_{n,\epsilon}}

\newcommand{\tH}{{\widetilde{H}}}

\newcommand{\tm}{{\widetilde{m}}}

\newcommand{\ty}{{\tilde{y}}}
\newcommand{\tvarphi}{{\widetilde{\varphi}}}
\newcommand{\tS}{{\widetilde{S}}}

\newcommand{\tlambda}{{\widetilde{\lambda}}}
\newcommand{\tR}{{\widetilde{R}}}
\newcommand{\tA}{{\widetilde{A}}}

\newcommand{\tu}{{\widetilde{u}}}

\newcommand{\tv}{{\widetilde{v}}}
\newcommand{\ta}{{\widetilde{a}}}
\newcommand{\tU}{{\widetilde{U}}}
\newcommand{\tdi}{{\widetilde{\di}}}
\newcommand{\tw}{{\widetilde{w}}}
\newcommand{\tchi}{{\widetilde{\chi}}}

\newcommand{\tp}{{\widetilde{p}}}
\newcommand{\tz}{{\tilde{z}}}
\newcommand{\tF}{{\widetilde{F}}}

\newcommand{\tr}{{\widetilde{r}}}
\newcommand{\tpsi}{{\widetilde{\psi}}}
\newcommand{\tG}{{\widetilde{G}}}

\newcommand{\tPhi}{{\widetilde{\Phi}}}

\newcommand{\tphi}{{\widetilde{\phi}}}

\newcommand{\tsigma}{{\widetilde{\sigma}}}

\newcommand{\tkappa}{\widetilde{\kappa}}

\newcommand{\hf}{\widehat{f}}
\newcommand{\ha}{{\widehat{a}}}
\newcommand{\hq}{{\widehat{q}}}

\newcommand{\hg}{{\widehat{g}}}

\newcommand{\p}{\partial}
\newcommand{\Di}{\slashed{D}}
\newcommand{\Dii}{\slashed{\mathfrak{D}}}
\newcommand{\di}{{\slashed{\partial}}{}}
\newcommand{\dii}{{\slashed{\mathfrak{d}}}{}}
\newcommand{\de}{ \ \mathrel{\stackrel{\makebox[0pt]{\mbox{\normalfont\tiny def}}}{=}} \ }
\newcommand{\dd}[2]{\dfrac{\p #1}{\p #2}}
\newcommand{\lr}[1]{\left \langle #1 \right\rangle}
\newcommand{\ove}[1]{{\overline{#1}}}
\newcommand{\matrice}[1]{\left[ \begin{matrix} #1 \end{matrix}\right]}
\newcommand{\systeme}[1]{\left\{ \begin{matrix} #1 \end{matrix}\right.}

\newcommand{\Tr}{{\operatorname{Tr \ }}}

\newcommand{\rk}{{\operatorname{rk}}}
\newcommand{\sgn}{{\operatorname{sgn}}}
\newcommand{\dist}{{\operatorname{d}}}
\newcommand{\Id}{{\operatorname{Id}}}

\newcommand{\supp}{{\operatorname{supp}}}
\newcommand{\End}{\operatorname{End}}

\newcommand{\vol}{\operatorname{Vol}}

\newcommand{\aaa}{{\mathfrak{a}}}

\NewEnviron{equations}{%
\begin{equation}\begin{gathered}   \BODY   \end{gathered}\end{equation} 
}

\newtheorem{lemma}{Lemma}[chapter]
\newtheorem{theorem}{Theorem}[chapter]
\newtheorem{remark}{Remark}[chapter]
\newtheorem{definition}{Definition}[chapter]
\newtheorem{corollary}{Corollary}[chapter]

\numberwithin{equation}{chapter}
\numberwithin{section}{chapter}
\numberwithin{figure}{chapter}

\title{Topological insulators in semiclassical regime}

\makeatletter
\renewcommand{\@maketitle}{%
  \thispagestyle{plain}%
  \begingroup 
  \centering
  \begingroup
    \LARGE\bfseries
    \@title\par\vspace{24pt}%
    \def\and{\par\medskip}\centering
    \mdseries\authors\par\bigskip
  \endgroup
}
\makeatother

\makeatletter
\newenvironment{chaptersummary}{%
  \normalfont
  \list{}{%
    \labelwidth\z@
    \leftmargin3pc
    \rightmargin\leftmargin
    \listparindent\normalparindent
    \itemindent\z@
    \parsep\z@ \@plus\p@
    
  }%
  \item[\hskip\labelsep\scshape\chaptersummaryname.]%
}{%
  \endlist
}
\makeatother

\providecommand{\chaptersummaryname}{}
\addto\captionsenglish{\renewcommand{\chaptersummaryname}{Abstract}}

\begin{document}

\author{Alexis Drouot\footnote{University of Washington, USA}}
\maketitle

\vspace{1cm}

\begin{chaptersummary}
We study solutions of $2 \times 2$ systems $(h D_t + \Di) \Psi_t = 0$ on $\R^2$ in the semiclassical regime $h \rightarrow 0$.  
Our Dirac operator $\Di$ is a standard model for interfaces between topological insulators: it represents a semimetal at null energy, and distinct topological phases at positive / negative energies. Its semiclassical symbol $\di$  has two opposite eigenvalues, intersecting conically along the curve $\Gamma_\di = \di^{-1}(0)$ of positions / momenta where the material behaves like a semimetal. 

We prove that wavepackets solving $(h D_t + \Di) \Psi_t = 0$, initially  concentrated in phase space on $\Gamma_\di$, split in two parts. The first part travels coherently along $\Gamma_\di$ with predetermined direction and speed. It is the dynamical manifestation of the famous ``edge state". The second part immediately collapses. 
This result provides a quantitative description of the asymmetric transport at the core of topological insulators -- a principle known as the bulk-edge correspondence.

Our approach consists of: (i) a Fourier integral operator reduction; (ii) a careful WKB analysis of a canonical model; (iii) a reconstruction procedure. It yields concrete formulas for the speed and profile of the traveling modes. As applications, we analytically describe dynamical edge states for models of magnetic, curved, and strained topological insulators.

\end{chaptersummary}

\tableofcontents

\renewcommand{\thechapter}{A}


\chapter{Introduction}

\section{Overview}

In this work, we study solutions to the semiclassical system
\begin{equation}\label{eq-10w}
(h D_t + \Di) \Psi_t = 0, \qquad \text{where:}
\end{equation}
\begin{itemize}
\item $h > 0$ is the semiclassical parameter;
\item $(t,x) \in \R \times \R^2$ and $D_t = \frac{1}{i}\p_t$;
\item $\Di$ is the Weyl quantization (see \cite[\S4]{Z12}) of a $2 \times 2$ Hermitian traceless symbol 
\begin{equation}\label{eq-10s}
\di(x,\xi) = \sum_{j=1}^3 p_j(x,\xi) \sigma_j, \qquad \sigma_1 = \matrice{0 & 1 \\ 1 & 0}, \qquad \sigma_2 = \matrice{0 & -i \\ i & 0}, \qquad \sigma_3 = \matrice{1 & 0 \\ 0 & -1}.
\end{equation}
\end{itemize}
We refer to \S\ref{sec:1} for a canonical example.

Semiclassically near points $(x,\xi)$ such that $\di(x,\xi)$ has distinct eigenvalues, the dynamics of \eqref{eq-10w} is well understood. Wavepackets concentrated near these points split in two pieces. A symbolic diagonalization procedure as in \cite{T75} and the classical-to-quantum correspondence \cite{DH72} imply that these parts follow the Hamiltonian trajectories of the two \textit{opposite} eigenvalues $\pm\sqrt{p_1^2+p_2^2+p_3^2}$ of \eqref{eq-10s}. In particular, they propagate in \textit{opposite} directions. 

In this work, we focus instead on the dynamics of solutions initially concentrated near crossings of  eigenvalues of \eqref{eq-10s}: wavepackets localized at points
\begin{equation}
(x,\xi) \in \Gamma_\di \de \bigcap_{j=1}^3 p_j^{-1}(0) = \di^{-1}(0).
\end{equation}
Our motivation stems from the field of topological insulators. These are physical systems insulating in their interior but whose boundaries or interfaces support unidirectional waves. The equation \eqref{eq-10w} is a standard mathematical model for the evolution of a particle along an interface between topologically distinct insulators; see e.g. \cite{RH08, HK10, S12, FC13}. According to the bulk-edge correspondence, the interface is conductive, with overall conductivity equal to the difference of  bulk Chern numbers computed across the interface \cite{H93,EG02,B19}. In particular, it is non-zero: topological insulators rhyme with asymmetric transport. This phenomenon manifests itself virtually across  all scales and physical fields: quantum science \cite{TK+82,S83,H93}, photonics \cite{RH08,LJS14}, electromagnetism \cite{MD20}, plasma physics \cite{P21}, mechanical and accoustical systems \cite{WLB15,PC+}, avian flocks \cite{SBM17}, geophysics \cite{DMV17,PDV19,DTV19}, and so on. 

According to the above discussion, the transport is symmetric away from $\Gamma_\di$. Hence, a quantitative understanding of the bulk-edge correspondence necessarily passes through a microlocal analysis of \eqref{eq-10w} near the crossing set $\Gamma_\di$. While one would naively expect that coherent transport takes place along $\Gamma_\di$ in both directions (the same way it does away from $\Gamma_\di$), we show instead that it is coherent in one direction but strongly dispersive in the other. Specifically, wavepackets initially concentrated (in phase space) on $\Gamma_\di$ split in two parts:
\begin{itemize}
\item One that travels coherently along $\Gamma_\di$, at speed independent of the initial profile;
\item One that immediately loses coherence and disperses along $\Gamma_\di$.
\end{itemize}
This validates, at the particle level, the experimental observation at the core of topological insulators: interfaces between topologically distinct medias act as one-way channels for the propagation of energy. As byproducts of our result, we obtain a direct method to produce edge states profiles in concrete systems, and explicit formula for the speed of propagation and the rate of dispersion.

\section{An example}\label{sec:1}

We present here our results when $\Di$ is the concrete operator
\begin{equation}\label{eq-13w}
\Di = \matrice{m(x) & hD_1 - i D_2 \\ hD_1+iD_2 & -m(x)}. 
\end{equation}
Such Hamiltonians emerge in the two-scale analysis of honeycomb systems, such as graphene  \cite{RH08,FLW16,D19,LWZ19,DW20}. We assume that $m \in C^\infty_b(\R^2,\R)$ (i.e. it is bounded together with all its derivatives) vanishes transversely across its zero set: 
\begin{equation}
m(x) = 0 \quad \Rightarrow \quad \nabla m(x) \neq 0. 
\end{equation}
In the physics literature, $m$ is called 
a domain wall. The set $m^{-1}(0)$ separates regions of positive and negative topology (corresponding to the sign of $m$); see \cite{B19,B19b} for precise meaning of the topological invariants involved here. 

The bulk-edge correspondence -- see e.g. \cite{KRS02,EG02,B19b} -- suggests that $m^{-1}(0)$ supports one-way currents. In other words, solutions to 
\begin{equation}\label{eq-13x}
\left(h D_t + \matrice{m(x) & hD_1 - i D_2 \\ hD_1+iD_2 & -m(x)}\right) \Psi_t = 0
\end{equation}
that are appropriately prepared at $t=0$ travel along $m^{-1}(0)$ in a prescribed direction. In \cite{BB+21}, we explicitly constructed some solutions of \eqref{eq-13x} with such properties. We asked whether the dynamics along $m^{-1}(0)$ is the sum of such traveling waves with a dispersive part.

Here, we introduce a microlocal framework for the analysis of \eqref{eq-13x} which allows us to respond to this question and treat much more general models. 
The operator \eqref{eq-13w} is a semiclassical operator, with Weyl symbol
\begin{equation}\label{eq-13t}
\di(x,\xi) \de \matrice{m(x) & \xi_1 - i \xi_2 \\ \xi_1 +i\xi_2 & -m(x)} = \xi_1 \sigma_1 + \xi_2 \sigma_2 + m(x) \sigma_3. 
\end{equation}
These $2 \times 2$ matrices have eigenvalues $\pm \sqrt{\xi^2 + m(x)^2}$, distinct away from the set
\begin{equation}\label{eq-13u}
\Gamma_\di \de \di^{-1}(0) = m^{-1}(0) \times \{0\}^2.
\end{equation}
We think of $\Gamma_\di = \di^{-1}(0)
$ as the natural semiclassical lift of the interface $m^{-1}(0)$. From energy conservation, we expect that functions initially concentrated away from $\Gamma_\di = \di^{-1}(0)$ remain away from $\Gamma_\di$. And indeed, the wavefront set away from $\Gamma_\di$ moves along the Hamiltonian trajectories of $\pm \sqrt{\xi^2+m(x)^2}$, i.e. the curves $(x_t,\xi_t)$ and $(x_{-t},\xi_{-t})$ with
\begin{equation}
\dfrac{dx_t}{dt} = \dfrac{\xi_t}{\sqrt{\xi_t^2 + m(x_t)^2}}, \qquad \dfrac{dx_t}{dt} =  \dfrac{\nabla m(x_t)}{\sqrt{\xi_t^2 + m(x_t)^2}},
\end{equation}
see e.g. \cite{DH72,T75,HS90}. These two curves never cross $\Gamma_\di$ and travel in opposite direction: the transport is symmetric away from $\Gamma_\di$. Hence, assymmetric transport (as predicted by the bulk-edge correspondence) may only arise along $\Gamma_\di$. On the other hand, because $\di$ has two equal eigenvalues along $\Gamma_\di$, the standard results about propagation of singularity do not apply. 

Our main result is a quantitative description of the wavepacket dynamic along $\Gamma_\di$. The Poisson bracket $M_\di = \frac{1}{2i} \{\di,\di\}$ (which is generally non-zero for matrix-valued symbols) plays an essential role in our statement. Based on the relation $\sigma_j \sigma_k = i\epsi_{jk\ell} \sigma_\ell$ for $(j,k,\ell)$ pairwise disjoint and $\epsi_{jk\ell}$ the signature of the permutation $(j,k,\ell)$, a computation gives:
\begin{equation}
M_\di \de \dfrac{1}{2i} \sum_{j=1}^2 \dd{\di}{\xi_j} \dd{\di}{x_j} - \dd{\di}{x_j} \dd{\di}{\xi_j} 
=\matrice{0 & \p_2 m + i \p_1 m \\ \p_2 m - i \p_1 m & 0}.
\end{equation}
This is a $2 \times 2$ Hermitian traceless matrix, which we can explicitly diagonalize:  
\begin{equations}
U^{-1}_\di M_\di U_\di  = \lambda^2_\di \matrice{-1 & 0 \\ 0 & 1},
\end{equations}
where the quantities $\lambda_\di$ and $U_\di$ are given by:
\begin{equations}
\lambda_\di(x,\xi) \de \big\| \nabla m(x) \big\|^{1/2}, \quad
U_\di(x,\xi) \de \dfrac{1}{2} \matrice{   -e^{-i\te(x)} & 1 \\ 1 & e^{i\te(x)} }, \quad e^{i\te(x)} = \dfrac{\p_2 m(x) - i \p_1 m(x)}{\| \nabla m(x) \|}.
\end{equations}
While the conjugation matrix $U_\di$ is not unique, the lines $\LL_{x,\xi}^-$ and $\LL_{x,\xi}^+$ respectively spanned by the first and second column of $U_\di$ are uniquely characterized for $\lambda_\di(x,\xi) \neq 0$ (in particular along $\Gamma_\di$). Specifically, they are the negative and positive energy eigenspace of $M_\di$:
\begin{equations}
\LL^\pm_{x,\xi} \de \ker\big( M_\di(x,\xi) \mp \lambda_\di(x,\xi) \big) = \C \matrice{\mp 1 \\ e^{i\te(x)}}.
\end{equations}

\begin{figure}[!t]
\floatbox[{\capbeside\thisfloatsetup{capbesideposition ={right,center}}}]{figure}[10cm]
{\caption{The interface $m^{-1}(0)$ separates the positive topology region ($m>0$) and negative-topology region ($m<0$). The quantity $x_t$ moves along $m^{-1}(0)$ at unit speed, with the positive topology region on its left.}\label{fig:4}}
{\begin{tikzpicture}
\node at (0,0) {\includegraphics[scale=1.05]{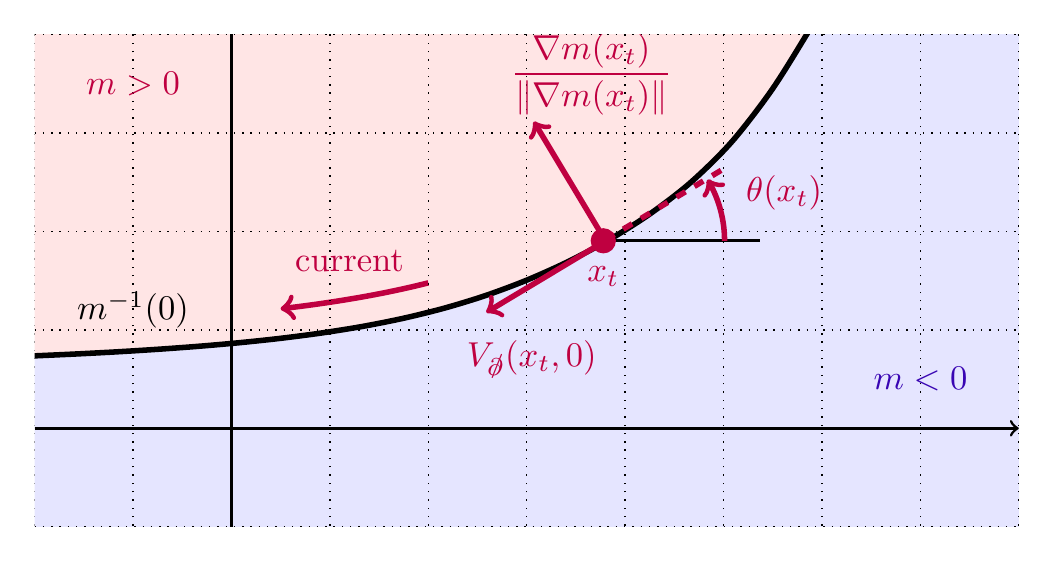}};
\end{tikzpicture}\hspace{1cm}}
\end{figure}

Another important object in our analysis is a vector field $V_\di$ constructed from $M_\di$ and the (matrix-valued) Hamiltonian vector field $H_\di$ of $\di$:
\begin{equation}
V_\di \de - \dfrac{1}{2 \lambda^2_\di} \Tr \left( M_\di   \cdot H_\di \right), \qquad H_\di \de \sum_{j=1}^2 \dd{\di}{\xi_j} \dd{}{x_j} - \dd{\di}{x_j} \dd{}{\xi_j}.
\end{equation}
For $\di$ given by \eqref{eq-13t}, a computation gives:
\begin{equation}
V_\di = \dfrac{1}{\| \nabla m(x) \|} \left( -\dd{m(x)}{x_2} \dd{}{x_1} + \dd{m(x)}{x_1}  \dd{}{x_2} \right) = \dfrac{\nabla m(x)^\perp}{\| \nabla m(x) \|},
\end{equation}
where $\nabla m(x)^\perp$ denotes the $\pi/2$ counter-clockwise rotation of a vector $\nabla m(x) \in \R^2$. 
The vector field $V_\di$ is well-defined when $\lambda_\di \neq 0$, in particular along $\Gamma_\di$; and at such points it is tangent to $\Gamma_\di$, see Figure \ref{fig:4}. Hence, integral curves of $V_\di$ starting at a point $(x_0,\xi_0) \in \Gamma_\di$ (that is, $m(x_0) = 0$ and $\xi_0 = 0$) are unit speed parametrizations of $\Gamma_\di$: they take the form $(x_t,0)$ with
\begin{equation}\label{eq-13z}
\dot{x_t} = \dfrac{\nabla m(x_t)^\perp}{\| \nabla m(x_t) \|}.
\end{equation}

\begin{definition}\label{def:1} A wavepacket concentrated at $(x_\star,\xi_\star) \in \R^2 \times \R^2$ is a $h$-dependent function $\Psi : \R^2 \rightarrow \C^2$ of the form
\begin{equation}
\Psi(x) = \dfrac{e^{is_h  + \frac{i}{h}\xi_\star(x-x_\star)}}{\sqrt{h}} A \left( \dfrac{x-x_\star}{\sqrt{h}} \right), 
\end{equation}
for some $A \in \SSS(\R^2,\C^2)$ (the Schwartz class of functions from $\R^2$ to $\C^2$) independent of $h$; and $s_h$ independent of $x$. If there is a vector $u \in \C^2$ such that for all $y \in \R^2$, $A(y) = a(y) u$ with $a \in \SSS(\R^2,\C)$, we will moreover say that $\Psi$ has orientation $u$ (or points in the direction of $\C u$). 
\end{definition}

Our result describes solutions to \eqref{eq-10w} that originally are wavepackets concentrated along $\Gamma_\di$:

\begin{theorem}\label{thm:1o} Let $\Di$ given by \eqref{eq-13w}; $\Psi_t$ a solution to \eqref{eq-10w}, with $\Psi_0$ a wavepacket concentrated at $(x_0,0) \in \Gamma_\di$; and $x_t$ be the solution to the ODE \eqref{eq-13z}. Then there exists $T > 0$ such that for all $t \in (0,T)$, uniformly as $h \rightarrow 0$:
\begin{equation}\label{eq-13y}
\Psi_t = \Phi_t + O_{L^\infty}(h^{-1/4}) + O_{L^2}(h^{1/4}),
\end{equation}
where $\Phi_t$ is a wavepacket concentrated at $(x_t,0)$, pointing along $\LL^-_{x_t,0}$. Moreover:

(i) If $\Psi_0$ points along $\LL^-_{x_0,0}$ then \eqref{eq-13y} holds for all $t > 0$ without $O_{L^\infty}(h^{-1/4})$-remainder;

(ii) If $\Psi_0$ points along $\LL^+_{x_0,0}$ then $\Phi_t = 0$; 

(iii) If $\| \nabla m(x) \|$ is constant along $m^{-1}(0)$ and $m^{-1}(0)$ is homeomorphic to $\R$ then \eqref{eq-13y} holds for all $t > 0$.
\end{theorem}

We interpret Theorem \ref{thm:1o} as follows. First, the remainder $O_{L^2}(h^{1/4})$ in \eqref{eq-13y} is negligible compared to the initial mass and can be ignored. From (i), if $\Psi_0$ points along $\LL^-_{x_0,0}$, then it propagates coherently along $\Gamma_\di = m^{-1}(0) \times \{0\}^2$, in direction and speed prescribed by \eqref{eq-13z}.  This is roughly the result of \cite{BB+21}. From (ii) and the $L^2$-conservation of energy, we see that the $O_{L^\infty}(h^{-1/4})$-remainder tends to be non-trivial.  Generally, when $\Psi_0$ points along an arbitrary direction, then $\Psi_t$ splits in two parts:
\begin{itemize}
\item[(a)] A traveling mode as in (i), corresponding to the time-dynamics of the projection of $\Psi_0$ onto the line $\LL^-_{x_0,0}$;
\item[(b)] A piece that immediately collapses: its $L^\infty$-norm is  $O(h^{-1/4})$, negligible compared to that of the traveling part $\Phi_t$, which is of order $h^{-1/2}$. 
\end{itemize}
This in spirit responds to \cite[Conjecture 1]{BB+21}: when $\| \nabla m(x) \|$ is constant along $m^{-1}(0)$ and $m^{-1}(0)$ is homeorphic to $\R$ (i.e. it represents an infinite edge), there is only one way to realize coherent propagation along $m^{-1}(0)$. Forcing propagation in the opposite direction results in an immediate loss of coherence. When $\| \nabla m(x) \|$ is not constant along $m^{-1}(0)$, \cite[Conjecture 2]{BB+21} remains open.

\section{General result}

There are several reasons to consider generalizations of \eqref{eq-13x}. It is genuinely interesting to study materials beyond graphene, in particular once symmetries are broken; these give rise to anisotropic versions of \eqref{eq-13x}, see e.g. \cite{TG+12,OL16}. Also of interest are the effects of magnetic fields \cite{BBD22} or a non-Euclidean backgrounds \cite{GAB20}. We refer to \S\ref{chap:4} for concrete Dirac operators corresponding to these sitations.

From a broader point of view, equations such as \eqref{eq-13x} emerge in the effective analysis of (either continuous or discrete) two-scale systems with microscopically varying topological structure. While topological variations generically guarantee that the symbol of the effective Hamiltonian admits conical crossings of eigenvalues \cite{D21b}, nothing can be said about the aperture of the associated Dirac cones. In the absence of additional symmetries, it generally varies with position. This leads to the general form \eqref{eq-10s} for the effective Hamiltonian. Because of the bulk-edge correspondence (which is independent of precise analytic aspects), we still expect interfaces to support unidirectional waves. Under generic hypothesis on $\Di$, our main result confirms this expectation and provides quantitative information about the structure of these waves.

\subsection{Assumptions and definitions}\label{sec-2.1} Let $\di$ be a symbol of the form \eqref{eq-10s}:
\begin{equation}\label{eq-0i}
\di(x,\xi) = \matrice{p_3(x,\xi) & p_1(x,\xi) - i p_2(x,\xi) \\ p_1(x,\xi) + i p_2(x,\xi) & - p_3(x,\xi)} =  \sum_{j=1}^3 p_j(x,\xi) \sigma_j, \qquad (x,\xi) \in \R^2 \times \R^2,
\end{equation}
where the symbols $p_j$ are real-valued, smooth, and grow at most quadratically: $\nabla^2 p_j \in C^\infty_b(\R^4)$. 
The eigenvalues of $\di$ are given by $\pm \sqrt{p_1^2 + p_2^2 + p_3^2}$. In particular, they are repeated precisely when 
\begin{equation}
(x,\xi) \in \Gamma_\di \de \cap_{j=1}^3 p_j^{-1}(0).
\end{equation}

We make the following two assumptions on $\Gamma_\di$:
\begin{itemize}
\item[$\bm{(\AAA_1})$] For all $(x,\xi) \in \Gamma_\di, \ dp_1(x,\xi), \ dp_2(x,\xi), \ dp_3(x,\xi)$   are linearly independent.
\item [$\bm{(\AAA_2)}$] The projection of $\Gamma_\di$ to $\R^2 \times \{0\}^2 \equiv \R^2$ is an embedded one-dimensional manifold of $\R^2$.
\end{itemize}
From a perspective inspired by topological insulators, $\bm{(\AAA_1})-\bm{(\AAA_2)}$ imply that $\pi(\Gamma_\di)$ splits $\R^2$ in insulating phases of alternating topology, see Figure \ref{fig:6}.

Again, the main object is the Poisson bracket $M_\di = \frac{1}{2i} \{\di,\di\}$. A computation gives 
\begin{equation}
M_\di = \dfrac{1}{2} \sum_{j,k,\ell=1}^3 \epsi_{jkl} \{p_j,p_k\} \sigma_\ell = \matrice{\{p_1,p_2\} & \{p_2,p_3\} - i \{p_3,p_1\} \\ \{p_2,p_3\} + i \{p_3,p_1\} & -\{p_1,p_2\}}.
\end{equation}
where $\epsi_{jk\ell}$ is the totally antisymmetric tensor, i.e. the signature of the permutation $(j,k,\ell)$. The symbol $M_\di$ is traceless and Hermitian; in particular it has two opposite eigenvalues:
\begin{equation}\label{eq-0o}
\pm \lambda_\di(x,\xi)^2, \qquad \lambda_\di \de \left(\dfrac{1}{2} \cdot \Tr\big( M^2_\di\big)\right)^{1/2} = \left( \sum_{j < k} \big\{p_j ,p_k\big\}^2 \right)^{1/4}.
\end{equation}

Under $\bm{(\AAA_1})$, at least one of the Poisson brackets $\{p_j,p_k\}$ does not vanish at each point of $\Gamma_\di$: otherwise, the kernels of the linear forms $dp_k$ at that point would be equal (to the span of $H_{p_1}, H_{p_2}, H_{p_3}$), contradicting $\bm{(\AAA_1})$. Hence $\lambda_\di$ does not vanish on $\Gamma_\di$. In particular, the eigenvalues of $M_\di$ are distinct along $\Gamma_\di$. We can then define line bundles $\LL^\pm \rightarrow \Gamma_\di$, whose fibers at points $(x,\xi) \in \Gamma_\di$ are
\begin{equation}
\LL^\pm_{x,\xi} = \ker\big( M_\di(x,\xi) \mp \lambda_\di(x,\xi)^2 \big). 
\end{equation}
The relation $\lambda_\di \neq 0$ on $\Gamma_\di$ also allows us to define the vector field 
\begin{equations}\label{eq-10v}
V_\di \de -\dfrac{\Tr \left(M_\di   \cdot H_{\di} \right)} {2 \lambda^2_\di } = \ -\dfrac{\sum_{j, k, \ell=1}^3 \epsi_{jk\ell} \big\{ p_j, p_k\big\} H_{p_\ell}}{\left(2\sum_{j, k=1}^3 \big\{ p_j, p_k\big\}^2\right)^{1/2}},
\end{equations}
on a neighborhood of $\Gamma_\di$.  We note that $V_\di$ is tangent to $\Gamma_\di$ because $V_\di p_j = 0$; moreover $\| V_\di \| \leq C \| H_\di \| \leq C \lr{x,\xi}$ (where $\lr{z} = \sqrt{1+|z|^2}$ denotes the Japanese bracket of $z$). In particular integral curves $(x_t,\xi_t)$ of $V_\di$ that start along $\Gamma_\di$, i.e. solutions to
\begin{equation}
\dfrac{d(x_t,\xi_t)}{dt} = V_\di(x_t,\xi_t), \qquad (x_0,\xi_0) \in \Gamma_\di
\end{equation} 
are defined for all times and contained within $\Gamma_\di$.

\subsection{Main result} We are ready to state our main result:

\begin{theorem}\label{thm:1b} Let $\Di$ with Weyl symbol satisfying $\bm{(\AAA_1)}-\bm{(\AAA_2)}$; $\Psi_t$ a solution to $(hD_t + \Di) \Psi_t = 0$, with $\Psi_0$ a wavepacket concentrated at $(x_0,\xi_0) \in \Gamma_\di$ (see Definition \ref{def:1}); and $(x_t,\xi_t)$ the solution to \eqref{eq-13z}. There exists $T > 0$ such that for all $t \in (0,T)$, uniformly as $h \rightarrow 0$:
\begin{equation}\label{eq-13j}
\Psi_t = \Phi_t + O_{L^\infty}(h^{-1/4}) + O_{L^2}(h^{1/4}),
\end{equation}
where $\Phi_t$ is a wavepacket concentrated at $(x_t,\xi_t)$, pointing along $\LL^-_{x_t,\xi_t}$. Moreover,

(i) If $\Psi_0$ points along $\LL^-_{x_0,\xi_0}$ then \eqref{eq-13j} holds for all $t > 0$ without $O_{L^\infty}(h^{-1/4})$-remainder;

(ii) If $\Psi_0$ points along $\LL^+_{x_0,\xi_0}$ then $\Phi_t = 0$.

 (iii) If $\lambda_\di$ is constant along $\Gamma_\di$ and $\Gamma_\di$ and its projection to $\R^2 \times \{0\}^2$ are homeomorphic to $\R$, then \eqref{eq-13j} holds for all $t > 0$.
\end{theorem}

Hence, the picture described below Theorem \ref{thm:1o} remains valid. It sits in sharp contrast with the standard propagation of singularity theorem: away from $\Gamma_\di$, wavepackets split in two pieces propagating in opposite directions. But along $\Gamma_\di$, coherent propagation is possible only at the speed and direction prescribed by $V_\di$, and for data with the right orientation. Te one with orthogonal orientation just disperse along $\Gamma_\di$. This is a quantitative version of the bulk-edge correspondence: it gives the speed and orientation of dynamical edge states; and elucidates the fate of their \textit{a priori} counter-propagating analogues. 

We add a few comments about \eqref{eq-13j}:
\begin{itemize}
\item Theorem \ref{thm:1b} concretely generates edge states for \eqref{eq-10w}: they are wavepackets concentrated at points $(x_0,\xi_0) \in \Gamma_\di$, pointing along the negative-energy eigenspace of $M_\di(x_0,\xi_0)$. Because edge states serve as robust channels for the propagation of energy \cite{HK10}, this has potential for technological applications. 
\item Technically speaking, our approach produces explicit oscillatory integral formulas for the two pieces in \eqref{eq-13j}. These can serve to derive additional estimates, such as concentration along $\Gamma_\di$ or finite speed of propagation for the $L^\infty$-remainder in \eqref{eq-13j}.
\item The limitation $t \in (0,T)$ in Theorem \ref{thm:1b} stems from the reversibility and spreading properties of a classical Hamiltonian flow emerging in a WKB analysis. It is not clear to us whether it can be systematically removed. 
\end{itemize}

\section{Examples}\label{sec-1.4} In \S\ref{sec:D}, we apply Theorem \ref{thm:1b} to the following scenarios:
\begin{itemize}
\item[1.] Domain wall Dirac operators;
\item[2.] Magnetic domain wall Dirac operators;
\item[3.] Domain wall Dirac operators in curved geometry;
\item[4.] Anistropic Haldane models.
\end{itemize}

We summarized the conclusions of the first example in \S\ref{sec:1}. For a Dirac operator with a domain wall $m(x)$, we recover the result of \cite{BB+21}, and show that the time $T$ of Theorem \ref{thm:1b} can be taken arbitrarily large when $\| \nabla m(x) \|$ is constant.  

The second example study the same equation as \cite{BBD22}: it corresponds to adding an external magnetic field $B(x)$ to \eqref{eq-13w}. There we constructed by hand traveling modes for magnetic Dirac operators. We showed that these propagate at speed
\begin{equation}
\dfrac{\nabla  m(x)^\perp}{\sqrt{B(x)^2 + \| \nabla m(x) \|^2}}
\end{equation}
along $m^{-1}(0)$. In \S\ref{sec:Da}, we directly verify this prediction using our microlocal framework. Theorem \ref{thm:1b} completes the result of \cite{BBD22} by showing that no other modes may propagate coherently.

We then apply our result to Dirac operators on curved surfaces with broken time-inversal invariance. Perhaps unsurprisingly, we find that propagation arises at unit speed (measured with respect to the Riemannian metric); in particular, it is unaffected by the curvature.

Finally, we consider distortions of a standard system from solid state physics: Haldane's model. We introduce an equation that describes the effective dynamics of Schr\"odinger waves in a weighted honeycomb lattice with broken time-reversal invariance. It has features of both scenarios above: a intrinsic, generally anisotropic Riemannian metric; and a (pseudo) magnetic field. Our result predicts that the distortion can only slow down the traveling mode. Hence, while defects do not affect the conductive nature of the interface, they actually reduce the speed of propagating signals.

\section{Organization of the paper.}

The proof of Theorem \ref{thm:1b} goes along the following outline:
\begin{itemize}
\item[\textbf{1.}] \textbf{Semiclassical reduction.} In \S\ref{chap:2}, we show that symbols of the form \eqref{eq-0i} are along $\Gamma_\di$ gauge-symplectic equivalent (that is, equal up to conjugation with a $SU(2)$-valued symbol $U$ and composition with a symplectomorphism $\kappa$) to symbols of the form
\begin{equation}\label{eq-0j}
\lambda(x_1) \matrice{\xi_1 & x_2 - i\xi_2 \\ x_2 + i \xi_2 & -\xi_1}, \qquad \lambda \in C^\infty(\R,\R^+).
\end{equation}
Our approach comes with explicit formula for $\lambda$ and $U$ in terms of the symplectomorphism $\kappa$. A quantization of \eqref{eq-0j} implies that $\Di$ and 
\begin{equation}\label{eq-0k} 
\Dii  \de \lambda(x_1) \matrice{h D_1 & x_2 -h \p_2 \\ x_2 + h \p_2 & -h D_1}
\end{equation}
are essentially conjugated up to a remainder that is semiclassically small along $\Gamma_\di$. For technical reasons (which we do not discuss at this point), we will need a refinement of $\Dii$. \\
\item[\textbf{2.}] \textbf{Analysis of the model equation.} In \S\ref{chap:3}, we analyze the behavior of solutions to $(h D_t + \Dii)\Psi_t = 0$, where $\Dii$ is (the refinement of) the Dirac operator \eqref{eq-0k}. This system (on the plane $\R^2$) decomposes into decoupled systems on the line, indexed by an integer $n \in \N$. We perform a WKB analysis on each of these systems, for a suitably rescaled semiclassical parameter. The limitation on $t$ in Theorem \ref{thm:1b} arises then from the reversibility and spreading properties of the emerging Hamiltonian dynamics. We then go on to prove \eqref{eq-13j} for the model equation.\\
\item[\textbf{3.}] \textbf{Reconstruction of the solution.} In \S\ref{chap:4}, we deduce from Steps 1 and 2 an integral formula for solutions to $(h D_t + \Di)\Psi_t = 0$ that are initially wavepackets concentrated at points of $\Gamma_\di$. We then use harmonic analysis techniques to prove that the estimate \eqref{eq-13j} for the model equation transfers back to the actual equation.
\end{itemize} 

We conclude the paper with a detailed analysis of the examples mentioned in \S\ref{sec-1.4}.

\section{Connections with existing results}

From a microlocal perspective, this work describes the propagation of singularities (in phase-space) for semiclassical systems. When eigenvalues of the matrix-valued symbol are simple (and more generally of constant multiplicity) singularities travel along Hamiltonian trajectories of the eigenvalues; see e.g. \cite{T75,B02,R07}. At leading order, vector-valued wavepackets essentially evolve according to a direct sum of scalar systems.  

The situation becomes more interesting when two eigenvalues cross along a subset of phase space. Such problems have a long history, starting with the Laudau--Zener effect \cite{L32,Z32}. In this setup, two eigenvalue cross along a curve, with energy varying along the curve. Once a wavepacket reaches this curve, it splits in two parts, propagating along each Hamiltonian trajectory of the eigenvalues. The Landau--Zener transition rule computes the splitting probabilities. The first mathematical analysis goes back to Haggedorn \cite{H91,H94} and continues with  \cite{HJ98, CV04, WW18, FLR21, FGH21}. 

The Landau--Zener situation is rather different from the setup studied here because the crossing energy for symbols of the form \eqref{eq-10s} is constant (equal to $0$). The set $\Gamma_\di = \di^{-1}(0)$ actually confines particles: a particle localized on $\Gamma_\di$ at some time must remain localized on $\Gamma_\di$ at all times. This scenario is typical of many topological insulator models \cite{HK10, BH13, FC13, AOP16}. The crossing set acts as a one-way channel for coherent propagation of energy. Theorem \ref{thm:1b} predicts moreover that forcing propagation in the opposite direction results in a strong and immediate loss of coherence. 

This picture is consistent with the bulk-edge correspondence, which predicts that transport between topologically distinct structures is asymmetric. This is a universal principle that has been proved for various discrete and continuous models, with and without disorder; see \cite{H93,KRS02,EG02,EGS05,GP12,ASV13,BKR17,
K17,B19,CST21,D21a,QB21}. 
For a precise statement in the context of translation-invariant Dirac operators, we refer to \cite{B19,B19b,B21,QB21}. Locally defined invariants give rise to asymptotic versions require  (as $h \rightarrow 0$); see e.g. \cite{MP+21} and \cite[\S1]{BB+21}.  Traveling modes were produced for some specific models in \cite{BB+21,BBD22,HXZ22}. The present work properly exploits the semiclassical structure of the problem and extends these constructions to general Dirac models, adressing along the way the fate of counter-propagating states. 

The symbol \eqref{eq-0i} has complex coefficients, and hence is not symmetric. The analysis of symmetric systems with crossings is also an interesting subject, with applications to elasticity \cite{H92} and electromagnetism \cite[\S6]{BD93}. We refer to \cite{BD93,FG02,CV03,F04,NU07} for their mathematical analysis.\\

\noindent \textbf{Acknowledgments.} I am thankful for fruitful discussions with Guillaume Bal, Simon Becker, Hadrian Quan and Maciej Zworski. I also gratefully acknowledge support from NSF grants DMS-2118608 and DMS-2054589.

\renewcommand{\thechapter}{B}

\chapter{Semiclassical reduction}\label{chap:2}

\section{Setting and main results} The first part of this paper is an analysis at the symbolic level. We focus on symbols
\begin{equation}\label{eq-10d}
\di(x,\xi) = \sum_{j=1}^3 p_j(x,\xi) \sigma_j, \qquad p_j \in C^\infty(\R^4,\R)
\end{equation}
whose characteristic set $\Gamma_\di = p_1^{-1}(0) \cap p_2^{-1}(0) \cap p_3^{-1}(0)$ satisfy the conditions $\bm{(\AAA_1})$, $\bm{(\AAA_2)}$ from \S\ref{sec-2.1}. 

As explained in \S\ref{sec-2.1}, a key quantity is the Hermitian traceless matrix 
\begin{equation}
M_\di = \dfrac{1}{2i} \big\{ \di,\di \big\}.
\end{equation} 
It has opposite eigenvalues $\pm \lambda^2_\di$ given by \eqref{eq-0o}. These are non-zero on a neighborhood of $\Gamma_\di$. Therefore, we can smoothly diagonalize $M_\di$ on simply connected subsets of $\{ \lambda_\di \neq 0\}$:
\begin{equation}
U_\di M_\di U_\di^* = - \lambda^2_\di \sigma_3, \qquad U_\di \in SU(2).
\end{equation}
Technically speaking, $U_\di$ is not unique: its columns are determined only modulo a phase.

A model operator is
\begin{equation}\label{eq-10e}
\dii(x,\xi) = \big(\lambda(x_1) +\mu(x_1)\xi_1\big) \matrice{ \xi_1 & x_2 -i\xi_2 \\ x_2+i\xi_2 & -\xi_1 },
\end{equation}
where $\lambda \in C^\infty(\R^+,\R)$ and $\mu \in C^\infty(\R,\R)$ are two given functions. These satisfy $\bm{(\AAA_1})$, $\bm{(\AAA_2)}$ on the set $\R \times \{0\}^3$: this set projects diffeomorphically to $\R \times \{0\}$; and $\dii$ vanishes transversely along $\R \times \{0\}^3$. For such symbols, we have
\begin{equation}\label{eq-0p}
M_\dii(x,\xi) = - \big( \lambda(x_1) + \mu(x_1) \xi_1 \big)^2 \sigma_3, 
\end{equation}
in particular $\lambda_\dii(x,\xi) = \lambda(x_1) + \mu(x_1) \xi_1$ and $U_\dii = \Id$.

To state our main result, we will use the notation $f = g + \OO(n)$ to denote two smooth functions $f, g \in C^\infty(\R^4)$ such that $f-g$ vanishes at order $n$ along $\R \times \{0\}^3$, i.e.
\begin{equation}
f(x,\xi) - g(x,\xi) =    \sum_{j+k+\ell = n} x_2^j \xi_1^k \xi_2^\ell \cdot  r_{jk\ell}(x,\xi), \qquad  r_{jk\ell} \in C^\infty(\R^4). 
\end{equation}
We provide $\R^4$ with its standard symplectic two-form, $\w = d\xi_1 \wedge dx_1 + d\xi_2 \wedge dx_2$.

\begin{theorem}\label{thm:1m} Let $\di$ satisfying $\bm{(\AAA_1})$, $\bm{(\AAA_2)}$ and fix $\rho \in \Gamma_\di$. For $(x,\xi)$ in a neighborhood of $0$:
\begin{equations}\label{eq-2h}
U(x,\xi) \cdot \di\circ \kappa(x,\xi) \cdot U(x,\xi)^* \  = \ \dii(x,\xi) + \OO(3), \qquad \text{where:}
\end{equations}
\begin{itemize}
\item $\kappa$ is a symplectomorphism of $\R^4$ with $\kappa(0) = \rho$, and such that $x \mapsto \pi \circ \kappa(x,0)$ is a smooth diffeomorphism from a neighborhood of $0$ in $\R^2$ to a neighborhood of $\rho$ in $\R^2$.
\item $\dii(x,\xi)$ is a symbol of the form \eqref{eq-10e};
\item $U$ is a smooth map from $\R^4$ to $SU(2)$.
\end{itemize}
\end{theorem}

We can actually compute the quantities $U$ and $\lambda$ of Theorem \ref{thm:1m} along $\Gamma_\di$, using transformation rules of $M_\di$ under conjugation and composition by symplectomorphism. According to \eqref{eq-0p} and \eqref{eq-2h}, on $\R \times \{0\}^3$:
\begin{align}
-\lambda^2 \sigma_3 = M_{\dii} & = \dfrac{1}{2i} \left\{ U \cdot \di\circ \kappa \cdot U^* , U \cdot \di\circ \kappa \cdot U^* \right\} 
\\
&
= U  \cdot \dfrac{1}{2i} \left\{ \di, \di \right\}\circ \kappa \cdot U^* 
 = U \cdot M_\di \circ \kappa \cdot U^* 
 \\
 & = - \lambda^2_\di \circ \kappa \cdot U U^*_\di \circ \kappa \cdot \sigma_3 \cdot U_\di \circ \kappa U^*,
\end{align}
where we used that $\di \circ \kappa$ and $\dii$ are conjugated, in particular $\di \circ \kappa = \OO(1)$. We deduce that along $\Gamma_\di$, $\lambda = \lambda_\di \circ \kappa$ and $U$, $U_\di \circ \kappa$ are equal up to multiplication of columns by a phase (in particular their columns are parallel).  We comment that providing an expression for $\mu$ is considerably more difficult as this quantity is a subleading parameter.

We now formulate an operator-theoretic version of Theorem \ref{thm:1m}. To this end, we very briefly review Fourier integral operators. These are semiclassical quantizations of symplectomorphisms; see \cite[\S11]{Z12}. Given $\phi \in C^\infty(\R^4,\R)$ quadratic away from a compact set, and satisfying the transversality conditions $\p_{x\xi} \phi \neq 0$, the prime example of symplectomorphism is defined (implicitly) by the relation $\kappa(\p_\xi \phi, \xi) = (x,\p_x\phi)$. We refer to $\phi$ as a generating function of $\kappa$. An example of Fourier integral operator quantizing $\kappa$ is
\begin{equation}\label{eq-10c}
\FF u(x) \de \int_{\R^2} e^{\frac{i}{h}(\phi(x,\xi) - y\xi)} c(x,\xi) u(y) \dfrac{dy}{(2\pi h)^2}, \qquad u \in \SSS(\R^2),
\end{equation}
where $c$ is a smooth symbol. 

Given three smooth function $s, \lambda, \mu : \R \rightarrow \R$, we introduce 
\begin{align}\label{eq-10b}
\Dii  = & \left( h \lambda(x_1) D_1 + \dfrac{h \lambda'(x_1)}{2i} + h^2 D_1 \mu(x_1) D_1  \right) \matrice{1 & 0 \\ 0 & -1}.
\\
& + \left( \lambda(x_1) + h \mu(x_1) D_1 + \dfrac{h \mu'(x_1)}{2i}\right) \matrice{0 & x_2 -h \p_2 \\ x_2 + h \p_2 & 0} + h s(x_1).
\end{align}
A direct verification shows that this semiclassical operator has principal symbol $\dii$. The term involving $s$ is subleading.

\begin{theorem}\label{thm:1d} Let $\Di$ be a pseudodifferential operator whose symbol $\di$ satisfies $\bm{(\AAA_1})$, $\bm{(\AAA_2)}$ above, and fix $\rho \in \Gamma_\di$. Then
\begin{equations}\label{eq-2u}
\FF^{-1} \Di \FF \  = \Dii + R, \qquad \text{where:}
\end{equations}
\begin{itemize}
\item $\FF$ is an invertible Fourier integral operator of the form \eqref{eq-10c}, quantizing the symplectomorphism $\kappa$ produced by Theorem \ref{thm:1m};
\item $\Dii$ takes the form \eqref{eq-10b}, where $\lambda, \mu$ are provided by Theorem \ref{thm:1m} and $s : \R \rightarrow \R$ is a smooth scalar-valued function;
\item $R$ is a pseudodifferential operator with symbol $\OO(3) + \OO(1) O(h) + \OO(h^2)$ near $0$. 
\end{itemize}
\end{theorem}

Our approach to prove Theorem \ref{thm:1m} is inspired by two beautiful papers of Colin de Verdi\`ere \cite{CV03, CV04}. There, the author perform a gauge-symplectic reduction of $2 \times 2$ symbols with crossings that mix time and space derivatives and coefficients. Technically speaking, one could apply \cite[Theorem 4]{CV04} to write the operator $h D_t + \Di$ in normal form. However, the Fourier integral operator produced in \cite{CV04} mixes time and space variable in an unspecified way. Without further analysis, this would doom our goal of estimating solutions to \eqref{eq-10w}. 

We chose instead to apply the strategy of \cite{CV03,CV04} to the operator $\Di$ rather than $h D_t + \Di$, tracking down along the way the value of the symplectic invariant $\lambda_\di$. Our conclusion differs from \cite[Theorem 4]{CV04} on three major points:
\begin{itemize}
\item Our gauge transformation is a conjugation by an element of $SU(2)$, instead of left/right multiplication by elements $P, P^\top \in GL(2)$.
\item We construct a symplectomorphism $\kappa$ that lifts and projects to a diffeomorphism of $\R^2$.
\item We guarantee that a fundamental coefficient arising in \cite[Propostion 3]{CV04} does not depend on $x_2,\xi_2$ (and we provide a leading-order expression in terms of $\di$). 
\end{itemize}
These additions happen to be essential in our analysis of solutions to \eqref{eq-10w}.

\section{Linear theory}\label{sec-2a} The simplest possible scenario consists of linear symbols. We introduce
\begin{equation}\label{eq-9y}
\DD \de \big\{ \di \quad \text{given by \eqref{eq-10d}, with $p_1, p_2, p_3$ linear in $(x,\xi)$, satisfying $\bm{(\AAA_1})$, $\bm{(\AAA_2)}$}\big\}.
\end{equation}
We have:

\begin{theorem}\label{thm:1l} Assume that $\di \in \DD$. There exists a $4 \times 4$ symplectic matrix $S$ with invertible upper-left $2 \times 2$ block; and a $2 \times 2$ matrix $U \in SU(2)$ such that
\begin{equation}\label{eq-9g}
U \cdot \di \circ S(x,\xi) \cdot U^* = \lambda_\di \cdot \matrice{ \xi_1 & x_2 -i\xi_2 \\ x_2+i\xi_2 & -\xi_1 }, \qquad \lambda_\di \de \left\| \dfrac{\left\{ \di,\di\right\} }{2}\right\|^{1/2}.
\end{equation}
Moreover, $U \in SU(2)$ diagonalizes the matrix $\{\di, \di\}$.

Finally, if $\di$ depends smoothly on a parameter taking values in a simply connected set, then we can chose $S$ and $U$ so that they also depend smoothly on the parameter. 
\end{theorem}

Given Theorem \ref{thm:1l}, our approach to prove Theorem \ref{thm:1d} will go as follows. When $\di$ does not depends linearly on $(x,\xi)$, we linearize $\di$ at each point of $\Gamma_\di$ (seen as parameter) and obtain, by Theorem \ref{thm:1l}, parameter-dependent quantities $\lambda$, $S$ and $U$. We then construct a symplectomorphism with differential $S$ along $\Gamma_\di$. Refinements transversal to $\Gamma_\di$ produce $\kappa$, $\mu$ and $U$; see \S\ref{sec-3.2}-\ref{sec:2}. 

\subsection{Symplectic reduction of quadratic forms}\label{sec:1.1} 

To prove Theorem \ref{thm:1l}, we follow a strategy from Colin de Verdi\`ere \cite{CV03}. It relies first on a symplectic reduction of $q = -\det \di$, which sets aside the construction of $U$; and a gauge-theoretic argument to produce $U$. 

We mention that \eqref{eq-9g} also admits a substantially simpler, more direct proof, which relies on the diagonalization of $\{\di,\di\}$, and comes with explicit expressions for $\lambda_\di$ and $U$, see \S\ref{sec-1.3} below. Colin de Verdi\`ere's approach nonetheless transfers more easily to non-linear symbols, see \S\ref{sec-2}.

When $\di(x,\xi)$ depends linearly on $(x,\xi)$ and satisfies the conditions (i) and (ii) above, $q = -\det \di$ is a  quadratic form of signature $(0,+,+,+)$. We refer to \cite[Theorem 21.5.3]{H85} for a general result on symplectic equivalence of quadratic forms. We develop here a special case, whose details will be relevant later.

\begin{lemma}\label{lem:1l} Let $q$ be a quadratic form of signature $(0,+,+,+)$ with $\pi ( \ker q ) \neq \{0\}$.  There exists $\lambda > 0$ and a $4 \times 4$ symplectic matrix $S$ with invertible upper-left $2 \times 2$ bloc such that:
\begin{equation}\label{eq-1e}
q \circ S(x,\xi) = \lambda^2 \big( x_2^2 + \xi_1^2 + \xi_2^2 \big). 
\end{equation}
Moreover, if $q$ depends smoothly on a parameter moving within an interval, then we can require $\lambda$ and $S$ to also depend smoothly on the parameter. 
\end{lemma}

\begin{proof} 1. Below we use the variable $z=(x,\xi)$. Let $Q = \nabla^2 q$ and $L = J Q$ (in particular, $\lr{Lz, \p_z}$ is the Hamiltonian vector field of $q$, so $L$ is the linearization of $H_q$ along $q^{-1}(0)$). This is a Hamiltonian matrix: since $J^2 = -\Id$, $J^\top J = \Id$ and $Q$ is a symmetric matrix:
\begin{equation}\label{eq-1b}
J L + L^\top J = - Q + Q^\top = 0.
\end{equation}
Equivalently, $\w(z,Lz') = -\w(Lz,z')$.
Moreover, we have the identity
\begin{equation}
q(z) = \dfrac{1}{2} \lr{Qz,z} = -\dfrac{1}{2} \lr{JLz,z} = -\dfrac{1}{2} \w(Lz,z). 
\end{equation}

2. Since $L = JQ$ with $\rk Q = 3$, we deduce $\rk L = 3$. In particular $0$ is an eigenvalue of $L$. We prove by contradiction that $L$ also has a non-zero eigenvalue. Otherwise, $L$ is nilpotent. From the theory of nilpotent matrices with maximal rank, we have $\rk L^j = 4-j$ for $j \in [1,4]$. We deduce that the quadratic form $q \circ L$ has kernel of dimension at least $3$, because
\begin{equation}
q (Lz) = -\dfrac{1}{2} \w(L^2z,Lz) = \dfrac{1}{2} \lr{JL^3z,z}. 
\end{equation}
However, the equation $q(Lz) = 0$ is equivalent to $QLz = 0$ because $q$ is nonnegative; and $QL = -JL^2$ has rank $2$, so the kernel of $q \circ L$ has dimension $2$. This is a contradiction. 

Let now $\mu$ be a non-zero eigenvalue of $L$. We note that $\mu \notin \R$: otherwise there would be a real-valued eigenvector $v$; and we would have
\begin{equation}
q(v) = -\dfrac{\mu}{2} \w(v,v) = 0.
\end{equation}
This cannot happen: indeed $\ker q = \ker L$ because they both have dimension one and the leftward inclusion is clear.

From \eqref{eq-1b} and the fact that $L$ is real-valued, $-\mu, -\ove{\mu}$ and $\ove{\mu}$ are also eigenvalues of $L$. Since zero is also an eigenvalue of $L$, we deduce that there must be a repetition; hence $\mu \in i\R$. Given that $L$ has rank $3$, we conclude that $L$ has precisely two non-zero eigenvalues, which moreover are simple, purely imaginary and opposite. We write them $\pm i \lambda^2$ with $\lambda > 0$. The other eigenvalue of $L$ is $0$, and has geometric multiplicity one, algebraic multiplicity two.

3. The vector spaces $\ker(L^2)$ and $\ker(L^2+\lambda^4)$ are symplectic orthogonal. Indeed, if $L^2 v = 0$ and $L^2 v' = -\lambda^4 v'$, then 
\begin{equation}
\w(v,v') = \w(v,\lambda^{-4}L^2v') = \lambda^{-4}\w(L^2 v,v') = 0.
\end{equation}
Let $u = u_2 + i u_4$ be an eigenvector for $i \lambda^2$. Taking real and imaginary parts of $Lu = i \lambda^2 u$, we have
\begin{equation}
L u_2 = - \lambda^2 u_4, \quad L u_4 = \lambda^2 u_2.
\end{equation}
In particular, we have
\begin{equation}
0 \leq q(u_2) = - \dfrac{1}{2} \w(Lu_2,u_2) = \dfrac{\lambda^2}{2} \w(u_4,u_2).  
\end{equation}
Since $u_2$ is symplectic-orthogonal to $\ker(L^2)$ and to itself, we must have $\w(u_4,u_2) \neq 0$, and we deduce that $\w(u_4,u_2) > 0$. Upon multiplying $u_2, u_4$ by $\w(u_4,u_2)^{-1/2}$ we can assume that $\w(u_4,u_2) = 1$. 

Since $0$ is an eigenvalue of $L$ of algebraic multiplicity two but geometric multiplicity one, $\ker(L^2)$ is two-dimensional and $\ker(L)$ is a one-dimensional subspace. Let $0 \neq u_3$ be a vector in the orthogonal complement of $\ker(L)$ in $\ker(L^2)$ and set $u_1 = \lambda^{-2} Lu_3$. Then, 
\begin{equation}\label{eq-10f}
0 \leq q(u_3) = - \dfrac{1}{2} \w(Lu_3,u_3) = -\dfrac{\lambda^2}{2} \w(u_1,u_3).
\end{equation}
We note that $\w(u_3,u_1) \neq 0$ is non-zero because $u_1$ is symplectic orthogonal to $\ker(L^2+\lambda^4)$ and to itself; hence we deduce from \eqref{eq-10f} that $\w(u_3,u_1) > 0$. Upon multiplying $u_1, u_3$ by $\w(u_3,u_1)^{-1/2}$ $u_1$ and $u_3$, we can assume that $\w(u_3,u_1) = 1$.

By construction, $(u_1,u_2,u_3,u_4)$ is a symplectic basis, hence $S = [u_1,u_2,u_3,u_4]$ is a symplectic matrix. We have:
\begin{equation}
S^{-1} L S = \lambda^2\matrice{0 & 0 & 1 & 0 \\ 0 & 0 & 0 & 1 \\ 0 & 0 & 0 & 0\\ 0 & -1 & 0 & 0}, \ \ \ \ S^\top Q S = -S^\top JL S =- JS^{-1} LS = \lambda^2 \matrice{0 & 0 & 0 & 0 \\ 0 & 1 & 0 & 0 \\ 0 & 0 & 1 & 0 \\ 0 & 0 & 0 & 1}.
\end{equation}
We conclude that
\begin{equation}\label{eq-9v}
q \circ S(z) = \dfrac{1}{2} \lr{z, S^\top Q S z} = \lambda^2 \big( z_2^2 + z_3^2 + z_4^2\big).
\end{equation}
This completes the proof of \eqref{eq-1e}. 

4. We show that we can pick $S$ so that its upper-left $2 \times 2$ block is invertible. With $u_j = [v_j,w_j]^\top, v_j, w_j \in \R^2$, this boils down to proving that the matrix $[v_1,v_2]$ is invertible -- after potentially modifying $S = [u_1,u_2,u_3,u_4]$. Recall that $\ker(q) = \R u_1$; hence from $\pi (\ker q) \neq \{0\}$ we deduce $v_1 \neq 0$. We claim that either $(v_1,v_2)$ or $(v_1,v_4)$ forms a basis of $\R^2$.

Otherwise, both $v_2$ and $v_4$ are parallel to $v_1$, hence (since $S$ is invertible) $(v_1,v_3)$ is a basis of $\R^2$. Assume that $v_2 = 0$. Then
\begin{equation}
0 = \w(u_2,u_1) = \lr{w_2,v_1}, \quad 0 = \w(u_2,u_3) = \lr{w_2,v_3};
\end{equation}
hence $w_2 = 0$ and $u_2 = 0$; this is a contradiction: $v_2 \neq 0$. Set $\tu_4 = u_4 - s u_2$, where $s$ is such that $v_4-s v_2 = 0$. Then $(u_1,u_2,u_3,\tu_4)$ is a symplectic basis of $\R^4$ and $\tv_4 = 0$ (with $\tu_4=[\tv_4,\tw_4]^\top$). The relations
\begin{equation}
0 = \w(\tu_4,u_1) = \lr{\tw_4,v_1}, \quad 0 = \w(\tu_4,u_3) = \lr{\tw_4,v_3}
\end{equation}
yield $\tw_4 = 0$, hence $\tu_4 = 0$: this is a contradiction. We conclude that either $(v_1,v_2)$ or $(v_1,v_4)$ forms a basis of $\R^2$.

If $(v_1,v_2)$ is a basis of $\R^2$, then $[v_1,v_2]$ -- the top-left block of $S$ -- is invertible, and we are done. If it is not, then $(v_2,v_4)$ is. Moreover,
\begin{equation}\label{eq-9w}
\tS \de [u_1,u_4,u_3,-u_2] = S \cdot \matrice{1 & 0 & 0 & 0 \\ 0 & 0 & 0 & 1 \\ 0 & 0 & 1 & 0 \\ 0& -1 & 0 & 0}
\end{equation}
has an invertible top-left $2 \times 2$ bloc; is symplectic (as the multiplication of two symplectic matrices); and satisfies $q \circ \tS = q \circ S$ because of \eqref{eq-9v}, \eqref{eq-9w}. 
 
5. Assume finally that $q$ depends smoothly on a parameter in a simply connected set $\Omega$; then so does $\lambda$. Let $\EE$ be the orthogonal complement of $\ker(L)$ in $\ker(L^2)$, which we regard as a vector bundle over $\Omega$. Since $\Omega$ is simply connected, $\EE \rightarrow \Omega$ admits a non-zero smooth section $u_3$, i.e. a vector $u_3$ that depends smoothly on $\zeta \in I$ and satisfying the requirement of Step 3. The vector $u_1 = \lambda^{-2} L u_1$ depends also smoothly on $\zeta \in \Omega$. 

Likewise, we regard $\ker\big(L-i\lambda(q)\big)$ as a vector bundle over $\Omega$. It admits a non-zero smooth section $u$, that is an eigenvector vector $u$ for $i\lambda(q)$ that depends smoothly on the parameter. Taking real and imaginary parts as above produces $u_2, u_4$ with smooth dependence on the parameter. 

We finally show that a suitable modification of $S=[u_1,u_2,u_2,u_4]$ enforces an invertible upper-left $2 \times 2$ block. Write $u = u_2+iu_4  = [v,w]^\top$. Since either $(v_1,v_2)$ or $(v_1,v_4)$ is a basis of $\R^2$, $\det[v_1,v] = \det[v_1,v_2] + i \det[v_1,v_4]$ never vanishes on $\Gamma_\di$. Therefore, its argument $\te$ is well-defined and varies smoothly on $\Gamma_\di$. We now make the phase change $u \mapsto e^{-i\te} u = \tu_2 + i \tu_4$. Then $v_2$ switches to $\tv_2 = \Re(e^{-i\te} v)$, and we observe that
\begin{equation}\label{eq-9q}
\det[v_1,\tv_2] = \Re\big( e^{-i\te} \det[v_1,v] \big) > 0. 
\end{equation}
The matrix $\tS = [u_1,\tu_2,u_3,\tu_4]$ depends smoothly on $\zeta$ and satisfies
\begin{equation}
\tS = S \matrice{1 & 0 & 0 & 0 \\ 0 & \cos \te & 0 & \sin \te \\ 0 & 0 & 1 & 0 \\ 0& -\sin(\te) & 0 & \cos \te}.
\end{equation}
In particular, it is symplectic (as the product of two symplectic matrices), has invertible upper-left $2 \times 2$ bloc by \eqref{eq-9q}, and satisfies $q \circ \tS = q \circ S$ by \eqref{eq-9v}. This completes the proof. 
\end{proof}

\subsection{Construction of the $SU(2)$-gauge.}\label{sec:1.2} Assume given $\di \in \DD$, see \eqref{eq-9y}. Then, $q = -\det \di$ satisfies the conditions of Lemma \ref{lem:1l}: its kernel is the line $\Gamma_\di$, and by assumption it projects non-trivially to $\R^2 \times \{0\}$. This produces a symplectic matrix $S$ such that \eqref{eq-1e} holds. Replacing $\di$ by $\lambda^{-1} \di \circ S$, we can then assume that
\begin{equation}
-\det \di(x,\xi) = x_2^2 + \xi_1^2 + \xi_2^2.
\end{equation}

\begin{lemma}\label{lem:1b} Assume that $\di \in \DD$ is such that $-\det \di(x,\xi) = x_2^2 + \xi_1^2 + \xi_2^2$. 
Then there exists $U \in SU(2)$ and $\nu \in \{1,-1\}$ such that 
\begin{equation}\label{eq-5r}
U^* \di(x,\xi) U =  \matrice{\nu\xi_1 & x_2 - i \xi_2 \\ x_2+i\xi_2 & -\nu\xi_1}.
\end{equation} 
Moreover, if $\di$ depends smoothly on a parameter $\zeta$ then $U$ also depends smoothly on $\zeta$.
\end{lemma} 

\begin{proof} 1. Let $z = (x_2,\xi_1,\xi_2)$. From \eqref{eq-10d} and the fact that $\di(x,\xi)$ depends linearly on $(x,\xi)$,
\begin{equation}
\di(0,z) = \sum_{k=1}^3 p_k(0,z) \sigma_k = \sum_{j=1}^3 \lr{\p_z p_k(0,z), z} \sigma_k = \sum_{j,k=1}^3 B_{jk} z_j \sigma_k = \lr{Bz, \sigma},
\end{equation}
where $B$ is the $3 \times 3$ matrix $[\p_z p_1, \p_z p_2, \p_z p_3]$; and we used the vector of Pauli matrices $\sigma = [\sigma_1,\sigma_2,\sigma_3]^\top$. From our assumption on $-\det \di$, we deduce that $B$ is an orthogonal matrix:
\begin{equation}
|z|^2 = -\det \di(0,z) = \sum_{k=1}^3 p_k(0,z)^2 = \sum_{k=1}^3 \lr{\p_z p_k, z}^2 = |Bz|^2.
\end{equation}

2. Assume that $\det B = 1$, i.e. $B \in SO(3)$ and define $\tsigma = B^\top \sigma$.  Since $(\sigma_1,\sigma_2,\sigma_3)$ is a direct orthonormal basis of traceless, Hermitian $2 \times 2$ matrices and $B \in SO(3)$, $(\tsigma_1,\tsigma_2,\tsigma_3)$ is also an orthonormal basis of $SO(3)$.  The isomorphism between $SO(3)$ and $SU(2)/\{\pm 1\}$ implies that there exists $U \in SU(2)$ with $U^* \tsigma_j U = \sigma_j$. We deduce that
\begin{equation}\label{eq-9h}
U \di(0,z) U^* = U^*\lr{z, B^\top \tsigma} U =  U^*\lr{z, \tsigma} U= \lr{z, \sigma}.
\end{equation}
Moreover, we note that 
\begin{equation}
0 = -\det \di(x_1,0) = \sum_{k=1}^3 p_k(x_1,0)^2.
\end{equation}
In particular, $p_k(x_1,0) = 0$ for $k \in [1,3]$. Using linearity of $\di(x,\xi)$, we conclude from $z = (x_2,\xi_1,\xi_2)$ and \eqref{eq-9h} that 
\begin{equation}
\di(x,\xi) = x_2 \sigma_1 + \xi_1 \sigma_2 + \xi_2 \sigma_3. 
\end{equation}
After conjugating by the matrix $e^{-i\pi\sigma_1/4}$, which is in $SU(2)$ and satisfies $e^{i\pi\sigma_1/4} \sigma_2 e^{-i\pi\sigma_1/4} = -\sigma_3$, $e^{i\sigma_1/2} \sigma_3 e^{-i\sigma_1/2} = \sigma_2$, we conclude that \eqref{eq-5r} holds (with $\nu = -1$).

3. When $\det B = -1$ then we can work with $z = (x_2,-\xi_1,\xi_2)$ instead. 
When $\di$ depends smoothly on a parameter then so does $B$; in particular $\det B$ is constant (say equal to $1$). Since we obtain $U$ from a Lie group isomorphism applied to $B$, $U$ also depends smoothly on $\zeta$. This completes the proof.\end{proof}

\begin{proof}[Proof of \eqref{eq-9g}] Fix $\di \in \DD$. Then we can apply Lemma \ref{lem:1l} to the quadratic form $-\det \di$. Let $S, \lambda$ be the emerging quantities. Then $\lambda^{-1} \di \circ S$ is in $\DD$ and satisfies the assumptions of Lemma \ref{lem:1b}. Therefore, there exists $U \in SU(2)$ and $\nu \in \{1,-1\}$ such that 
\begin{equation}
\dfrac{1}{\lambda} U^* \di\circ S(x,\xi) U =  \matrice{\nu\xi_1 & x_2 - i \xi_2 \\ x_2+i\xi_2 & -\nu\xi_1}.
\end{equation} 
We note that if $\di$ depends on a parameter, so do $S$, $\lambda$, $U$ and $\nu$. In particular, $\nu$ has constant sign. If $\nu=1$, this completes the proof. If $\nu = -1$, we compose the symplectic transformation $S$ with $(x,\xi) \mapsto (-x_1,x_2,-\xi_1,\xi_2)$ (which is also symplectic). This completes the proof.  \end{proof}

\subsection{Formula for $\lambda_\di$ and $U$}\label{sec:1.3} We give here a formula for $\lambda_\di$ and a characterization of $U$ -- up to diagonal elements of $SU(2)$. The approach presented here can also serve as the basis for an alternate proof of Theorem \ref{thm:1l}, see Remark \ref{rem:2} below. Given $\di \in \DD$, define
\begin{equation}\label{eq-12m}
M_\di \de \dfrac{1}{2i}\{\di,\di\} = \dfrac{1}{2i} \sum_{j,k=1}^3 \{p_j \sigma_j, p_k \sigma_k\} = \matrice{\{p_1,p_2\} & \{p_2,p_3\} - i \{p_3,p_1\} \\ \{p_2,p_3\} + i \{p_3,p_1\} & -\{p_1,p_2\}}.
\end{equation}
This is a $2 \times 2$ Hermitian traceless matrix, which moreover does not depend on $(x,\xi)$, because $\di$ is linear in $(x,\xi)$. In particular, $M_\di$ is unitarily equivalent to a scalar multiple of $\sigma_3$. 

We observe moreover that 
\begin{equation}
M_{\di_0} = -\sigma_3, \qquad \di_0 = \matrice{\xi_1 & x_2 - i \xi_2 \\ x_2 + i \xi_2 & -\xi_1}.
\end{equation}
Let $\lambda_\di, U, S$ such that \eqref{eq-9g} holds. Using that $\lambda_\di, U, S, M_\di$ do not depend on $(x,\xi)$ (because $\di$ is linear), and that Poisson brackets are invariants under symplectomorphisms:
\begin{equation}\label{eq-12i}
-\lambda^2_\di \sigma_3 = M_{\lambda_\di \di_0} = M_{U^* \cdot \di \circ S \cdot U} = U^* \cdot M_{\di \circ S} \cdot U = U^* \cdot M_{\di}  \circ S \cdot U = U^* \cdot M_{\di} \cdot U.
\end{equation}

The identity \eqref{eq-12i} provides the following description of $\lambda_\di$ and $U$:
\begin{itemize}
\item $\lambda_\di$ is the unique positive number such that $\lambda^2_\di$ is an eigenvalue of $M_\di$, that is
\begin{equation}
\lambda_\di = \left( \dfrac{1}{2}\Tr M^2_\di \right)^{1/2} = \| M_\di \|^{1/2} = \left\| \dfrac{\{ \di,\di\}}{2} \right\|^{1/2};
\end{equation}
\item $U$ is a matrix in $SU(2)$ whose first and second columns are respectively a negative and positive energy eigenvectors of $M_\di$ (in particular $U$ diagonalizes $M_\di$). This characterizes $U$ up to diagonal elements of $SU(2)$, i.e. up to a phase.
\end{itemize}

\begin{remark}\label{rem:2} The above characterization can serve to provide a more direct proof of \eqref{eq-9g}, based on a diagonalization of $M_\di$ rather than on a reduction of $\det \di$. It does not come immediately with the invertibility of the top-left $2 \times 2$ block of $S$ and the dependence on the parameters left to the reader; we leave these to the reader. 

We start with a spectral decomposition of $M_\di$ (which is $2 \times 2$ Hermitian, traceless):
\begin{equation}
\tU^* \cdot M_\di \cdot \tU = -\tlambda^2 \sigma_3, \qquad \tU \in SU(2), \qquad \tlambda > 0.
\end{equation}
We now define
\begin{equation}\label{eq-12j}
\tdi \de \dfrac{\tU^* \cdot \di \cdot \tU}{\lambda_\di} \de \matrice{\tp_3 & \tp_1-i\tp_2 \\ \tp_1+i\tp_2 & -\tp_3}.
\end{equation}
Moreover, using on one hand the transformation rule \eqref{eq-12i} for gauge-symplectic changes of $\di$; and on the other hand \eqref{eq-12m}:
\begin{equation}
M_\tdi = \dfrac{1}{\lambda^2_\di} \tU \cdot M_\di \cdot \tU = -\sigma_3 = \matrice{\{\tp_1,\tp_2\} & \{\tp_2,\tp_3\} - i \{\tp_3,\tp_1\} \\ \{\tp_2,\tp_3\} + i \{\tp_3,\tp_1\} & -\{\tp_1,\tp_2\}},
\end{equation}
we deduce that $\tp_1, \tp_2, \tp_3$ are three linear symbols with $\{\tp_1,\tp_2\} = -1$ and $\{\tp_2,\tp_3\} = \{\tp_3,\tp_1\} = 0$. According to the linear Darboux theorem, there exists a symplectic matrix $S$ such that 
\begin{equation}\label{eq-12k}
\tp_1 \circ S = x_2, \qquad \tp_2 \circ S = \xi_2, \qquad \tp_3 \circ S = \xi_1.
\end{equation}
Composing \eqref{eq-12j} with $S$ and using the relations \eqref{eq-12k} yields \eqref{eq-9g}.\end{remark}

\section{Non-linear theory}\label{sec-2}

In this section, we write $r(x,\xi) = \OO(n)$ to denote a smooth function $r$ of the form
\begin{equation}
r(x,\xi) = \sum_{j+k+\ell  = n} x_2^j \xi_1^k \xi_2^\ell \ r_{jk\ell}(x,\xi),
\end{equation}
where $r_{jk\ell}$ are smooth functions on $\R^4$. This means that $r$ vanishes at order at least $n$ along $\R \times \{0\}^3$. Our goal here is to prove Theorem \ref{thm:1m}. As in \S\ref{sec-2a}, we will base our analysis first on the symplectic reduction of $q = -\det \di = p_1^2+p_2^2+p_3^2$:

\begin{theorem}\label{thm:1c} Fix $\rho \in \Gamma_\di$. There exists a smooth symplectomorphism $\kappa : \R^4 \rightarrow \R^4$ and two smooth functions $\lambda : \R \rightarrow \R^+$, $\mu : \R \rightarrow \R$ such that $\kappa(0) = \rho$ and, on a neighborhood of $0$:
\begin{equation}\label{eq-2g}
q \circ \kappa(x,\xi) = \big(\lambda(x_1)+\mu(x_1)\xi_1\big)^2 \big( x_2^2 + \xi_1^2 + \xi_2^2\big) + \OO(4). 
\end{equation}

Moreover $\p_x \pi \circ \kappa(0)$ is invertible; and for $x_1$ in a neighborhood of $0$, the function $\lambda$ in \eqref{eq-10e} is given in terms of $\di$ and $\kappa$ by: 
\begin{equation}\label{eq-9r}
\lambda(x_1) = \lambda_\di \circ \kappa(x_1,0).
\end{equation}
\end{theorem}

We refer to \cite[Proposition 3]{CV04} for a related result. The statement there involves a multiplicative term $e(x,\xi)$ with unspecified dependence on $x$ and $\xi$; we demonstrate here that this term depends on $x_1, \xi_1$ only.  This is crucial in \S\ref{chap:3}.

\subsection{Geometric preliminaries} Our approach to prove Theorem \ref{thm:1c} follows closely \cite[\S7.2]{CV03}. We start by looking at the dominant term in the Taylor expansion of $q$ near each point of $\Gamma_\di$, which is a quadratic form of signature $(0,+,+,+)$. We can then apply Lemma \ref{lem:1l}, and obtain a family of linear symplectomorphisms, parametrized by points of $\Gamma_\di$, whose elements symplectically reduce the quadratic approximation of $q$ at those points. An appropriate application of Lemma \ref{lem:1p} below will generate a symplectomorphism whose differential at points of $\Gamma_\di$ is given by these symplectic maps; it yields a symplectic reduction of $q$ modulo $\OO(3)$. After further refinement, we will reduce this remainder to $\OO(4)$.

\begin{lemma}\label{lem:1p} Let $M$ be a smooth manifold provided with two symplectic forms $\w_0$, $\w_1$. Assume that $\w_0$ and $\w_1$ coincide on a one-dimensional submanifold $N \subset M$, and fix $\rho_0 \in N$. After potentially shrinking $M, N$ to neighborhoods of $\rho_0$, there exists a symplectomorphism $\kappa_1 : (M,\w_0) \rightarrow (M,\w_1)$ such that for all $\rho \in N$:
\begin{equation}\label{eq-6k}
\kappa_1(\rho) = \rho, \quad d\kappa_1(\rho) = \Id.
\end{equation} 
\end{lemma}

Related statements appear in \cite[\S4]{W71} and were used in \cite{CV03}.  Nonetheless it seems to us that they only come with the conclusion $d\kappa_1 |_{TN} = \Id_{TN}$; this happens to be unsufficient in later parts of our work. This justifies to include a self-contained proof of Lemma \ref{lem:1p}. 

We start with an auxiliary result. Given a one-form $\beta$ on $M$ that vanishes at a point $\rho \in N$, we define the linearization of $\beta$ at $\rho \in N$ as the map
\begin{equation}\label{eq-1a}
v \in T_\rho M \mapsto (\LL_V \beta)(\rho) \in T_\rho^* M. 
\end{equation}
In \eqref{eq-1a}, $V$ is any vector field on $M$ that extends $v$, i.e. such that $V(\rho) = v$.\footnote{To verify that \eqref{eq-1a} is indeed independent of the choice of $V$ extending $v$, we note that if $W$ is a vector field that vanishes at $0$, then $W \lrcorner d\beta = 0$ and $d(W \lrcorner \beta) = 0$ (because $W  \lrcorner \beta = \beta(W)$ is a function that vanishes at order $2$ at $\rho$). Hence by Cartan's magical formula
\begin{equation}
\LL_{V+W} \beta (\rho) = \LL_V \beta(\rho) + d  W \lrcorner \beta (\rho) + W \lrcorner d\beta (\rho) = \LL_V \beta(\rho).
\end{equation}
This proves that \eqref{eq-1a} is well-defined.}

\begin{lemma}\label{lem:1a} Let $\alpha$ be a one-form on $M$ such that $d\alpha$ vanishes on a one-dimensional submanifold $N$ of $M$, and fix $\rho_0 \in N$. After potentially shrinking $M, N$ to neighborhoods of $\rho_0$, there exists a one-form $\beta$ such that:
\begin{itemize}
\item[(i)] $\beta$ vanishes on $N$;
\item[(ii)] The linearization of $\beta$ at any point of $N$ vanishes;
\item[(iii)] $d\beta = d\alpha$. 
\end{itemize}
\end{lemma}

\begin{proof} Without loss of generalities, we can assume that $M$ admits a global system of coordinates $(x_1, \dots, x_n)$ such that:
\begin{equation}
N = \big\{ (x_1, 0) : x_1 \in (-1,1)\big\}. 
\end{equation} 
We now construct a smooth function $f$ on $M$ such that $\beta = \alpha-df$ vanishes on $M$ and such that its linearization at any point of $N$ vanishes as well.

By Taylor-expanding the coefficients of $\az$ in the basis $(dx_1, \dots, dx_n)$ with respect to the variables $(x_2, \dots, x_n)$, we can write
\begin{equation}\label{eq-1f}
\az(x) = \sum_{j=1}^n \left( \az_j(x_1) +  \sum_{k=2}^n x_k  \az_{jk}(x_1) \right) dx_j+ \OO(2) .
\end{equation}
We compute the differential of \eqref{eq-1f} at a point $(x_1,0) \in N$:
\begin{equation}
d\alpha(x_1,0) = \sum_{j=1}^n \az'_j(x_1) dx_1 \wedge dx_j + \sum_{j=1}^n \sum_{k=2}^n  \az_{jk}(x_1) dx_k \wedge dx_j.
\end{equation}
Since $d\alpha(x_1,0) = 0$, we deduce the functions $\alpha_j, \alpha_{jk}$ satisfy the relations
\begin{equation}\label{eq-1h}
\alpha_{jk} = \alpha_{kj}, \quad j, k \in [2,n]; \qquad \az'_k = \az_{1k}, \quad k \in [1,n].
\end{equation}

We now look for $f$ as a partial Taylor development with respect to $(x_2, \dots, x_n)$:
\begin{equation}\label{eq-1g}
f(x) = f_1(x_1) + \sum_{j=2}^n x_j f_j(x_1) + \dfrac{1}{2} \sum_{j,k=2}^n x_j x_k f_{jk}(x_1),
\end{equation}
with $f_j, f_{jk}$ smooth functions on $(-1,1)$ and $f_{jk} = f_{kj}$. Thanks to \eqref{eq-1g}:
\begin{equation}\label{eq-1i}
df(x) = \left(f'_1(x_1) + \sum_{j=2}^n  x_j f'_j(x_1)  \right) dx_1 + \sum_{j=2}^n\left(   f_j(x_1) + \sum_{k=2}^n x_k f_{jk}(x_1) \right) dx_j + \OO(2). 
\end{equation}
To ensure that ensure that $\beta = \alpha - df$ vanishes along $N$, we pick $f_j = \alpha_j$ for $j \in [2,n]$ and $f_1 = \int_0^{x_1} \alpha_1$. We now set $f_{jk} = \alpha_{jk}$ for $j,k \in [2,n]$ -- in particular the constraint $f_{jk} = f_{kj}$ holds because of \eqref{eq-1h}. These relations define $f_j$ and $f_{jk}$, hence the function $f$ by \eqref{eq-1g}. We plug them in \eqref{eq-1i}:
\begin{equation}
df(x) = \sum_{j=1}^n  \alpha_j(x_1) + dx_j + \sum_{j=2}^n x_j \alpha'_j(x_1)  dx_1 + \sum_{j,k=2}^n x_k \alpha_{jk}(x_1) dx_j + \OO(2). 
\end{equation}
Coming back to \eqref{eq-1f}, we deduce that $\beta = \alpha - df$ satisfies
\begin{equation}
\beta(x) = \sum_{k=2}^n x_k \az_{1k}(x_1) dx_1 -  \sum_{j=2}^n  x_j\alpha'_j(x_1)  dx_1 + \OO(2) = \OO(2),
\end{equation}
where we used $\alpha_{1k} = \alpha'_k$ from \eqref{eq-1h}. This shows that $\beta$ vanishes at order $2$ on $N$. In particular $\LL_{\p_k} \beta(x_1,0) = 0$ for all $k \in [1,n]$: the linearization of $\beta$ at any point of $N$ vanishes. This completes the proof. \end{proof}

\begin{proof}[Proof of Lemma \ref{lem:1p}] We apply Moser's trick. We first observe that the two-form $\w_0 - \w_1$ is closed and vanishes along $N$. Therefore, locally near any point of $M$, we may write $\w_0 - \w_1 = d\alpha$ for a one-form $\alpha$ that vanishes along $N$.

Let $\beta$ be a one-form on $M$ that satisfy the conclusions of Lemma \ref{lem:1a}; in particular we have $\w_0-\w_1 = d\beta$. We define  $X_t$ as the unique (time-dependent) vector field on $M$ that satisfies
\begin{equation}\label{eq-1j}
X_t \lrcorner \w_t = -\beta, \quad \w_t \de \w_0 + t (w_1-\w_0). 
\end{equation}

Let $\kappa_t$ be the flow of $X_t$. Using that $\p_t \w_t = w_1-\w_0$, Cartan's magical formula, \eqref{eq-1j} and $w_1-\w_0 = d\beta$, we obtain:
\begin{equation}
\dfrac{d}{dt} \kappa_t^* \w_t = \kappa_t^* \LL_{X_t} \w_t  + \kappa_t^* w_1-\w_0 = \kappa_t^* \big( X_t \lrcorner d\w_t +  d(X_t \lrcorner \w_t) + w_1-\w_0  \big) = \kappa_t^* (w_1-\w_0-d\beta) = 0.
\end{equation}
In particular, $\kappa_1^* \w_1 = \w_0$. 

Since $\beta$ vanishes on $N$, so does $X_t$ by \eqref{eq-1j}. It follows that $\kappa_t(p) = p$ for all $p \in N$. If $p \in N$ and $v \in T_p M$, $d\kappa_t(p) \cdot v \in T_p M$ and we have
\begin{equation}\label{eq-6l}
\dfrac{d}{dt} d\kappa_t(p) \cdot v = d\kappa_t(p) \cdot \LL_{X_t} v(p).
\end{equation}
Since the linearization of $\beta$ at $p$ vanishes, we have
\begin{equation}
0 = -\LL_v \beta(p) = \LL_v (X_t \lrcorner \w_t) (p) = X_t(p) \lrcorner \LL_v  \w_t + (\LL_v X_t)(p) \lrcorner \w_t = -\LL_{X_t} v(p) \lrcorner \w_t,
\end{equation}
where we used the product rule for derivatives and interior product; and $X_t(p) = 0$. Since $\w_t = \w_0$ along $N$, we deduce that $\LL_{X_t} v(p) = 0$. It follows that \eqref{eq-6l} vanishes; hence $d\kappa_t(p) \cdot v = d\kappa_0(p) \cdot v = v$ for any $t \in [0,1]$. This completes the proof. 
\end{proof}

\subsection{Symplectomorphism at leading order} We prove here Theorem \ref{thm:1c}. We use \S\ref{sec-2a} and Lemma \ref{lem:1a} to first construct a symplectomorphism such that \eqref{eq-2g} holds with a $\OO(3)$-remainder. Towards the end we will use a lemma to refine the remainder to $\OO(4)$, which we will prove in \S\ref{sec-3.2}.

\begin{proof}[Proof of Theorem \ref{thm:1c}] 1. In this proof, we use the notation $z = (x,\xi)$. Fix a point $\zeta \in \Gamma_\di$ and define a quadratic form $q(\zeta ; \cdot)$ on $\R^4$ by
\begin{equation}
q(\zeta;z) = \dfrac{1}{2} \lr{ z, \nabla^2 q(\zeta) z } = | \nabla p_1(\zeta) \cdot z |^2 + | \nabla p_2(\zeta) \cdot z |^2 + | \nabla p_3(\zeta) \cdot z |^2. 
\end{equation}
Since $\nabla p_1(\zeta), \nabla p_2(\zeta)$ and $\nabla p_3(\zeta)$ are linearly independent, $q(\zeta;\cdot)$ is a quadratic form of signature $(0,+,+,+)$. Its kernel is
\begin{equation}\label{eq-1m}
\ker\big( q(\zeta; \cdot) \big) = \bigcap_{j=1}^3 \big(\R \nabla p_j(\zeta) \big)^\perp = \bigcap_{j=1}^3 T_\zeta p_j^{-1}(0) = T_\zeta \Gamma_\di,
\end{equation}
where in the last identity we used that $p_j(\zeta) = 0$ and that gradients are orthogonal to level sets.

From Theorem \ref{lem:1l}, there exists a symplectic matrix $S(\zeta)$ and a number $\lambda(\zeta) > 0$ such that
\begin{equation}
q(\zeta; S(\zeta) z) = \lambda(\zeta)^2 \big( z_2^2 + z_3^2 + z_4^2 \big). 
\end{equation}
For $j \in [1,4]$, we set $u_j(\zeta) = S(\zeta) e_j$. Let $I$ be a simply connected open subset of $\Gamma_\di$ containing $\rho$. By varying $\zeta$ along $I$, we can regard each $u_j$ as smooth sections of $T\R^4|_I$. We also observe that $u_1(\zeta) = S(\zeta)e_1$ is in the kernel of $q(\zeta;\cdot)$; from \eqref{eq-1m} we deduce that $u_1 \in T\Gamma_\di \big|_I$.

2. For $j \in [1,4]$, we extend the sections $u_j$ of $T\R^4|_I$ to vector fields on $\R^4$ (still denoted $u_j$). Define a diffeomorphism $\phi$ from a neighborhood of $0$ to a neighborhood of $\rho$ by
\begin{equation}
\phi(z) = e^{z_2 u_2 + z_3 u_3 + z_4 u_4} \circ e^{z_1 u_1}(\rho).  
\end{equation}
With $\zeta = e^{z_1 u_1}(\rho)$, we observe that
\begin{equations}\label{eq-1c}
d\phi(z_1,0) \cdot \p_1 = d e^{z_1 u_1}(\rho) \cdot \p_1 = u_1(\zeta) = u_1(\zeta).
\end{equations}
Since $u_1$ is tangent to $\Gamma_\di$ on $I$, we deduce that $\phi$ maps an open subset of $\{(z_1,0)\}$ to $I$. Likewise,
\begin{equations}\label{eq-1d}
d\phi(z_1,0) \cdot \p_j =  d e^{z_j u_j}(\zeta) \cdot \p_j = u_j(\zeta) = u_j(\zeta) \quad \text{for $j \geq 2$.}
\end{equations}
We deduce from \eqref{eq-1c}, \eqref{eq-1d} that $d\phi(z_1,0) = S(\zeta)$ is symplectic for $z_1 \in \R$. Hence, $\phi$ is a smooth diffeomorphism; and if $\w_1 = \phi^* \w$ (with $\w$ the standard symplectic form on $\R^4$), then $\w_1$ is a symplectic form on a neighborhood of $0$ which coincides with $\w$ when restricted to $\R \times \{0\}^3$. By Lemma \ref{lem:1p}, there exists a symplectomorphism $\kappa_1$ from / to neighbhorhoods of $0$ with
\begin{equation}
\kappa_1(z_1,0) = (z_1, 0), \qquad d\kappa_1(z_1,0) = \Id.
\end{equation}

4. Define $\tkappa = \phi \circ \kappa_1$. Then $\tkappa$ is a symplectomorphism from a neighborhood of $0$ to a neighborhood of $\rho$:
\begin{equation}
\tkappa^* \w = \kappa_1^* \phi^* \w = \kappa_1^* \w_1 = \w.
\end{equation}
Moreover, using \eqref{eq-1j} and \eqref{eq-6k}, we have $\kappa(z_1,0) = e^{z_1 u_1}(\rho)$ and, with $z = (z_1,\tz) \in \R \times \R^3$:
\begin{equation}
\tkappa(z) = \tkappa(z_1,0) + d\tkappa(z_1,0) \cdot \tz + \OO(2) = \zeta + S(\zeta) \cdot \tz + \OO(2), \qquad \zeta \de e^{z_1 u_1}(\rho).
\end{equation}
We recall that $q(\zeta) = 0$ and $\nabla q(\zeta) = 0$. Therefore,
\begin{equations}\label{eq-5c}
q \circ \tkappa(z) = q\big(\zeta + S(\zeta) \cdot \tz + \OO(2)\big) = \lr{ S(\zeta) \tz, \nabla^2 q(\zeta) \cdot S(\zeta) \tz } + \OO(3)
\\
 = q(\zeta; S(\zeta) \tz) + \OO(3) 
 = \lambda \circ \tkappa(z_1,0)^2 \cdot \big( z_2^2+z_3^2 + z_4^2 \big) + \OO(3).
\end{equations}

5. Noting that $\lambda$ does not vanish, we can write \eqref{eq-5c} as 
\begin{equation}
q \circ \tkappa(z) = \lambda^2 \circ \tkappa(z_1,0) \cdot q_1(z),
\end{equation} 
where $q_1(z) = z_2^2+z_3^2+z_4^2+\OO(3)$. This is nearly \eqref{eq-2g}, apart from the remainder $\OO(3)$ instead of $\OO(4)$. To conclude, we will need the following lemma:

\begin{lemma}\label{lem:1j}
 Let $q_1 : \R^4 \rightarrow \R$ be a smooth function such that  $q_1(z) = z_2^2+z_3^2+z_4^2 + \OO(3)$. 
Then there exists two smooth functions $f : \R^4 \rightarrow \R$ and $\nu : \R \rightarrow \R$ such that
\begin{equation}\label{eq-5i}
q_1 \circ e^{H_f} (z) = \big( 1 +  \nu(z_1) z_3\big)^2 \big( z_2^2+z_3^2+z_4^2 \big) + \OO(4)\qquad \text{and} \qquad e^{H_f}(z_1,0) = (z_1,0).
\end{equation}
\end{lemma}

We will prove this lemma in \S\ref{sec-3.2} below. Define now 
\begin{equation}\label{eq-11p}
\kappa(z) \de \tkappa \circ e^{H_f}, \qquad \mu(z_1)  \de \nu(z_1) \cdot \lambda \circ \tkappa(z_1,0),
\end{equation}
where $f$ and $\nu$ comes from Lemma \ref{lem:1j}. The map $\kappa$ is a  symplectomorphism, defined from a neighborhood of $0$ to a neighborhood of $\rho$. We extend it to a symplectomorphism of $\R^4$ thanks to \cite[Theorem 11.4]{Z12}.

 For $z$ near $0$, we have $e^{H_f} (z) = (z_1,0) + \OO(1)$; this yields \eqref{eq-2g}:
\begin{equations}
q \circ \kappa = q \circ \tkappa \circ e^{H_f}(z) = \lambda^2 \circ \tkappa(z_1,0) \big( 1+\nu(z_1) z_3\big)^2 \cdot \big( z_2^2+z_3^2+z_4^2 \big) + \OO(4)
\\
= \big(\lambda \circ \tkappa(z_1,0) +\mu(z_1) z_3\big)^2 \cdot \big( z_2^2+z_3^2+z_4^2 \big) + \OO(4).
\end{equations}
Finally, we note that $\p_x \pi \circ \kappa (0)$ is the top $2 \times 2$-bloc of the matrix $S(\rho)$, which is invertible by Theorem \ref{thm:1l}. This completes the proof. \end{proof}

\subsection{Refining the symplectomorphism}\label{sec-3.2} The proof of Theorem \ref{thm:1c} relied on Lemma \ref{lem:1j}, which we now prove. It relies on some (so far, unexploited) flexibility in the choice of $u_1$. Our proof is inspired by that of \cite[Lemma 2]{CV03}.

\begin{proof}[Proof of Lemma \ref{lem:1j}] We will use the following notations:
\begin{itemize}
\item $\PP_k$ is the space of homogeneous polynomials of degree $k$ in $(z_2,z_4)$, with coefficients depending smoothly on $z_1$ (in particular, elements of $\PP_k$ do not depend on $z_3$);
\item $q_0$ is the quadratic form defined by $q_0(z) = z_2^2+z_3^2+z_4^2$ and $r$ is the remainder $r = q_1-q_0 = \OO(3)$;
\item We write $z \in \R^4$ as $(z_1,\tz) \in \R \times \R^3$.
\end{itemize}
We look for the function $f$ in the form:
\begin{equation}\label{eq-5g}
f(z) = z_3^2 \cdot f_0(z_1) + \sum_{j=1}^3 z_3^{3-j} f_j(z), \qquad f_j \in \PP_j.
\end{equation}

1. We claim that for $f$ in the form \eqref{eq-5g} and $t \in [0,1]$, 
\begin{equation}\label{eq-5h}
e^{tH_f}(z) = z + 2t z_3 f_0'(z_1) e_1 + \OO(2).
\end{equation} 
This in particular implies that $e^{H_f}(z_1,0) = (z_1,0)$.
We observe that $H_f = 2z_3 f_0'(z_1) \p_1 + \OO(2)$; therefore, if $z(t)$ is an integral curve of $H_f$, then we we have
\begin{equation}
| \dot{z}(t) - 2z_3(t) f_0'\big(z_1(t)\big) e_1 | \leq C | \tz(t) |^2.
\end{equation}
We set $y(t) = z(t) - z(0) - 2t z_3(0) f_0'(z_1(0))$. In particular $\ty(t) = \tz(t)-\tz(0)$ and it follows that
\begin{equation}
-\dfrac{d}{dt} \dfrac{1}{| \ty(t) | + |\tz(0)|} = \dfrac{\lr{\dot{\ty}(t),y(t)}}{|y(t)|} \cdot \dfrac{1}{\big(| \ty(t)| + |\tz(0)|\big)^2} \leq \dfrac{|\dot{\tz}(t)|}{|\tz(t)|^2} \leq C.
\end{equation}
Integrating both sides, and isolating $\ty(t)$, we conclude that
\begin{equation}\label{eq-5l}
| \ty(t) | \leq \dfrac{Ct |\tz(0)|^2}{1-Ct |\tz(0)|^2}, \qquad |\tz(t)| \leq C |\tz(0)|.  
\end{equation}

We now work on $y_1$. Then there exists $C > 0$ such that
\begin{align}
| \dot{y_1}(t) | = | \dot{z_1}(t) - 2 z_3(0) f'_0(z_1(0)) | & \leq | 2 z_3(t) f'_0(z_1(t)) - 2 z_3(0) f'_0(z_1(0)) | + C |\tz(t)|^2
\\
\label{eq-5m} & \leq C |y_3(t)| + C|z_3(0)| |z_1(t)-z_1(0)| + C |\tz(0)|^2 
\\
& \leq  C \big( |\tz(0)| |y_1(t)| + |\tz(0)|^2 \big).
\end{align}
For the purpose of proving \eqref{eq-5h}, we can assume without loss of generalities that $2C|\tz(0)| \leq 1$. We claim that
\begin{equation}\label{eq-5k}
T \de \inf \big\{ t \geq 0 : |y_1(t)| \geq 2C |\tz(0)|^2\big\} \geq 1.
\end{equation}
Indeed, for $t \in [0,T]$, we have thanks to \eqref{eq-5m} and $2C|\tz(0)| \leq 1$:
\begin{equation}
| \dot{y_1}(t) | \leq C|\tz(0)|^2 \big( 2C |\tz(0)| + 1 \big) \leq 2C|\tz(0)|^2.
\end{equation}
Integrating on $[0,T]$ and using $y_1(0) = 0$, $y_1(T) = 2C |\tz(0)|^2$ yields
\begin{equation}
2C |\tz(0)|^2 \leq 2C |\tz(0)|^2 T \qquad \Rightarrow \qquad 1 \leq T.
\end{equation}
This proves the claim \eqref{eq-5k}. Together with \eqref{eq-5l}, we conclude that when $\tz(0)$ is sufficiently small, $|y(t)| \leq 2C |\tz(0)|^2$. The claim \eqref{eq-5h} follows from the definition of $y(t)$. 

2. We now use \eqref{eq-5h} to derive an expansion for $q_1 \circ e^{H_f}$ modulo $\OO(4)$. We have
\begin{equation}
q_1 \circ e^{H_f} - q_1 = \int_0^1 \dfrac{d}{dt} q_1 \circ e^{tH_f} dt = \int_0^1 \{f,q_1\} \circ e^{tH_f} dt.
\end{equation}
Since $\{z_3,q_0\} = 0$, we have
\begin{equation}\label{eq-5j}
\{f,q_0\} = -2 z_3^3 f_0'(z_1) + \sum_{j=1}^3 z_3^{3-j} \{f_j,q_0\} = \OO(3).
\end{equation}
In particular, thanks to \eqref{eq-5h}, we obtain $\{f,q_0\} \circ e^{tH_f} = \{f,q_0\} + \OO(4)$. Likewise, we have
\begin{equation}
\{f,r\} = z_3^2 \{f_0,r\} + \sum_{j=1}^3 z_3^{3-j} \{f_j,r\} - 2z_3 f_0 \p_1 r - \sum_{j=1}^3 (3-j)z_3^{2-j} f_j \p_1 r.
\end{equation}
We note that $\{f_0,r\} = \OO(2)$, $\{f_j,r\} = \OO(j+1)$, and $\p_1 r = \OO(3)$. Therefore $\{f,r\} = \OO(4)$. Composing with $e^{tH_f}(z) = z + \OO(1)$ produces $\{f,r\} \circ e^{tH_f} = \OO(4)$. We conclude that 
\begin{equation}
q_1 \circ e^{H_f} = q_1 + \{f,q_0\} + \OO(4) = q_0 + r + \{f,q_0\} + \OO(4). 
\end{equation}
In particular, the equation \eqref{eq-5i} is equivalent to
\begin{equation}\label{eq-5d}
r+\{f,q_0\} = 2\nu(z_1) z_3 q_0 + \OO(4).
\end{equation}

4. Define $\LL g = \{g,z_2^2+z_4^2\}$. Going back to \eqref{eq-5j}, we have
\begin{equation}
\{ f,q_0 \} = -2 z_3^3 f_0'(z_1) + \sum_{j=1}^3 z_3^{3-j} \LL  f_j - \sum_{j=1}^3 z_3^{4-j} \p_1 f_j.
\end{equation}
The last sum is $\OO(4)$ because $f_j \in \PP_j$, hence $\p_1 f_j \in \PP_j$. Therefore, we can simply ignore it. We now expand $r$ as 
\begin{equation}\label{eq-10q}
r(z) = \sum_{j=0}^3 z_3^{3-j} r_j(z) + \OO(4), \qquad r_j \in \PP_j.
\end{equation}
Identifying terms in $z_3^j$ on each side of \eqref{eq-5d}, we conclude that this equation is equivalent to the system
\begin{equation}\label{eq-5e}
r_3 + \LL f_3 = 0, \qquad 
r_2 + \LL f_2 = 2\nu (z_2^2+z_4^2), \qquad 
r_1 + \LL f_1 = 0, \qquad
r_0 - 2f_0' = \nu.
\end{equation}

3. We now solve \eqref{eq-5e}. We observe that the operator $\LL$ acts on $\PP_j$; below we specifically focus on $j \in \{1,2,3\}$. In polar coordinates $z_2= \rho \cos \te$ and $z_4 = \rho \sin \te$, we simply have $\LL  = \p_\te$. It follows that $\LL$ is a bijection of $\PP_3$ (and more generally, of $\PP_j$ for $j$ odd); therefore we can find $f_3$ with $r_3 + \LL f_3 = 0$; likewise we can find $f_1$ such that $r_1 + \LL f_1 = 0$. 

The range of $\LL$ on $\PP_2$ is however not all of $\PP_2$: it consists of elements of $\PP_2$ which average to $0$ on the orbit $\{ z_2^2+z_4^2 = 1 \}$. Therefore, if we set 
\begin{equation}\label{eq-10p}
\nu \de \dfrac{1}{4\pi} \int_{z_2^2+z_4^2 = 1} r_2,
\end{equation}
then we can find $f_2 \in \PP_2$ such that $r_2 + \LL f_2 = \nu (z_2^2+z_4^2)$. We finally construct $f_0$ by integrating $f_0'$ in the identity $r_0 - 2f_0' = \nu$. This completes the proof. 
\end{proof}

\begin{remark}\label{rem:1} From the formula \eqref{eq-10p}, and the relation \eqref{eq-10q} between $r$ and $r_2$, we see that if $r$ has no term in $\xi_1 x_2^2$ or $\xi_1 \xi_2^2$, then $\nu$ vanishes -- hence so does the coefficient $\mu$ defined in \eqref{eq-11p}. This will serve when looking at the example of \S\ref{sec:10}.
\end{remark}

\subsection{Generating function} We complete Theorem \ref{thm:1c} with a result on the existence of a generating function of $\kappa$ near $0$ with useful global properties. This justifies that one can quantize $\kappa$ into a Fourier integral operator of the form \eqref{eq-10c}. This will serve in \S\ref{chap:4}.

\begin{lemma}\label{lem:1q} Let $\kappa$ be the symplectomorphism produced by Theorem \ref{thm:1c}. There exists a smooth function $\phi : \R^4 \rightarrow \R$ such that: 
\begin{itemize}
\item[(i)] $\phi$ is quadratic outside a neighborhood of $0$;
\item[(ii)] The map $(x,\xi) \mapsto \big( \p_\xi \phi(x,\xi), \xi \big)$ is a smooth diffeomorphism of $\R^4$;
\item[(iii)] For $(x,\xi)$ in a neighborhood of $0$,
\begin{equation}
\kappa \big( \p_\xi \phi(x,\xi), \xi \big)  = \big( x, \p_x \phi(x,\xi) \big).
\end{equation}
\end{itemize}
\end{lemma}

\begin{proof} 1. Since $\p_x \pi \circ \kappa(0)$ is invertible, there exists a smooth function $\tphi$ defined on a neighborbood of $0$ such that
\begin{equation}
\kappa \big( \p_\xi \tphi(x,\xi), \xi \big)  = \big( x, \p_x \tphi(x,\xi) \big).
\end{equation}
We now Taylor-expand $\tphi$ near $(x,0)$ with respect to $\xi$; then we expand the remainder $r$ near $(\xi,0)$ with respect to $x$: 
\begin{equation}
\tphi(x,\xi) = \tphi(x,0) + \xi \cdot r(x,\xi) = \tphi_1(x) + \tphi_2(\xi) + \xi \cdot \tPhi(x,\xi) x,
\end{equation}
for some smooth functions $\tphi_1, \tphi_2, \tPhi$ defined on a neighborhood of $0$. We extend these functions into smooth functions bounded together with all their derivatives on $\R^2$. Also, we remark from the invertibility of $\p_x \pi \circ \kappa(0)$ that the matrix $\p_{x\xi}^2 \tphi(0) = \tPhi(0)$ is invertible. 

2. Fix $\chi : \R \rightarrow \R$ a smooth function equal to $1$ on $[-1,1]$ and $0$ outside $(-2,2)$. For $r \in (0,1]$, we introduce, with $z=(x,\xi)$:
\begin{equations}\label{eq-11r}
\tPhi(z) \de \chi\left( \dfrac{|z|}{r}\right) \tPhi(z) + \left(1-\chi\left( \dfrac{|z|}{r}\right) \right) \tPhi(0),
\\
\phi(z) \de \chi\left( r|z| \right) \tphi_1(x) + \chi\left( r|z| \right) \tphi_2(\xi) + x \cdot \Phi(z) \xi.
\end{equations}
Then $\Phi = \tPhi$ in the ball of radius $r$, and hence $\phi = \tphi$ in this ball (as long as $r$ is small enough); this implies that (iii) holds. Likewise $\Phi = \Phi(0)$ outside the ball of radius $2r$, therefore $\phi$ is quadratic outside this ball; this implies that (i) holds.

3. We now check that $\p_{x\xi}^2 \phi$ never vanishes as long as $r$ is sufficiently small. Taking the derivatives with respect to $(x,\xi)$ of the first two terms in the definition \eqref{eq-11r} produces a term $O(r)$. For the third term, we write
\begin{equation}
\Phi(z) = \tPhi(0) + \chi\left( \dfrac{|z|}{r}\right) z \cdot \Psi(z).
\end{equation}
for some smooth function $\Psi$. Taking one derivative with respect to $x,\xi$ produces a term $O(1)$; two derivatives produces $O(r^{-1})$. Therefore, we end up with
\begin{equation}
\p_{x\xi}^2 \phi(x,\xi) = \tPhi(0) + O(r). 
\end{equation}
Picking $r$ sufficiently small ensures that $\p_{x\xi}^2 \phi$ is invertible at every point. Moreover, outside a compact set we have $\p_\xi \phi(x,\xi) = \Phi(0)x$. It follows that the map $(x,\xi) \mapsto \big( \p_\xi \phi(x,\xi), \xi \big)$ is proper. By Hadamard's global inversion theorem, see e.g. \cite[Theorem 6.2.4]{KP02} , this map is a diffeomorphism. 
\end{proof}

\subsection{Construction of a $SU(2)$-gauge}\label{sec:2} To complete the proof of Theorem \ref{thm:1m}, it suffices to construct a suitable gauge transformation.

\begin{proof}[Proof of Theorem \ref{thm:1m}] 1. Let $\kappa, \lambda, \mu$ be the symplectomorphism and functions constructed in Theorem \ref{thm:1c}. We note that $q$ vanishes on $\Gamma_\di$ while $q \circ \kappa$ vanishes on a neighborhood of $0$ in $\R \times \{0\}^3$. It follows that $p_j \circ \kappa(z_1,0)$ vanishes on a neighborhood of $0$ in  $\R \times \{0\}^3$. Therefore, for $z_1$ sufficiently small, we have
\begin{equation}
\dfrac{p_j \circ \kappa(z)}{\lambda(z_1) (1+\mu(z_1)z_3)} =  \lr{ \alpha_j(z_1), \tz} + \OO(2), \qquad \alpha_j(z_1) = \dfrac{\nabla \big( p_j \circ \kappa \big)(z_1,0)}{\lambda(z_1)}.
\end{equation}
In particular,
\begin{equation}\label{eq-1r}
\dfrac{\di \circ \kappa(z)}{\lambda(z_1) (1+\mu(z_1)z_3)} = \sum_{j=1}^3 \lr{\alpha_j(z_1),\tz} \sigma_j + \OO(2).
\end{equation}

2. We take the determinant on both sides of \eqref{eq-1r} and obtain, thanks to \eqref{eq-2g}:
\begin{equation}
-\det\left[\dfrac{\di \circ \kappa(z)}{\lambda(z_1) (1+\mu(z_1)z_3)}\right] = |\tz|^2 + \OO(4) = -\det\left[\sum_{j=1}^3 \lr{\alpha_j(z_1),\tz} \sigma_j \right]+ \OO(3).
\end{equation}
Identifying terms that are quadratic in $\tz$, we deduce that
\begin{equation}
|\tz|^2 = -\det\left[\sum_{j=1}^3 \lr{\alpha_j(z_1),\tz} \sigma_j \right] = \sum_{j=1}^3 \lr{\alpha_j(z_1),\tz}^2.
\end{equation}
It follows from Lemma \ref{lem:1b} that there exists a smooth map $U$ defined on a neighborhood of $0$ to $SU(2)$ such that
\begin{equation}\label{eq-1y}
U(z_1)^* \left( \sum_{j=1}^3 \lr{\alpha_j(z_1),\tz} \sigma_j\right) U(z_1) = z_2 \sigma_1 + z_3 \sigma_2 + z_4 \sigma_1 \de \di_0(z).
\end{equation}
Plugging \eqref{eq-1y} into \eqref{eq-1r} yields 
\begin{equation}\label{eq-5w}
U(z_1)^* \di \circ \kappa(z) U(z_1) = \lambda(z_1) \big(1+\mu(z_1)z_3\big) \big( \di_0(z) + \OO(2) \big).
\end{equation}

3. Thanks to \eqref{eq-5w}, to prove the theorem it suffices to work with a symbol $\di$ such that
\begin{equation}
q(z) = q_0(z) + \OO(4), \qquad \Tr \di(z) = 0, \qquad \di(z) = \di_0(z) + \OO(2),
\end{equation}
and construct $U : \R^4 \rightarrow SU(2)$ such that $U(z)^* \di(z) U(z) = \di_0(z) + \OO(3)$. 

Below we will look for $U(z) = e^{iH(z)}$, where $H = \sum_{k=1}^3 h_{k+1} \sigma_k$, $h_2, h_3, h_4 \in \PP_1$. In particular, $U(z) = \Id + i H(z) + \OO(2)$, hence $U(z)^* \di(z) U(z) = \di_0(z) + \OO(3)$ is equivalent to
\begin{equation}\label{eq-5q}
\di(z) \big( \Id + i H(z) \big) = H(z) \di_0(z). 
\end{equation} 
Writing $\di(z) = \di_0(z) + R(z) + \OO(3)$ where $R = 2\sum_{j=1}^3 r_{j+1} \sigma_j$, $r_j \in \PP_2$. 
Then \eqref{eq-5q} reduces to
\begin{equation}
R(z) = i [H(z),\di_0(z)]. 
\end{equation}
Using the decomposition of $R$ and $H$ in terms of Pauli matrices, we end up with 
\begin{equation}
2\sum_{j=1}^3 r_j \sigma_j = \sum_{k, \ell = 1}^3 h_{k+1} z_{\ell+1} \cdot i [\sigma_k,\sigma_\ell] = \sum_{j, k, \ell = 1}^3 h_{k+1} z_{\ell+1} \cdot 2 \epsi_{jk\ell} \sigma_j.
\end{equation}
This reduces to a $3 \times 3$ system:
\begin{equation}\label{eq-5n}
\systeme{ z_3 h_4 - z_4 h_3 = r_2 \\ z_4 h_2 - z_2 h_4 = r_3 \\ z_2 h_3 - z_3 h_2 = r_4 }.
\end{equation}
This is a cross product-type system: it takes the form $z \times h = r$ with $z = [z_1,z_2,z_3]^\top$, $h = [h_2,h_3,h_4]^\top$ (the unknown) and $r = [r_2,r_3,r_4]^\top$ (the data). In particular, it is degenerate: it has a solution if and only if $z \cdot r = 0$. We verify that $r$ satisfies this condition: from $q(z) = q_0(z) + \OO(4)$, we have
\begin{equation}
(z_2+r_2)^2 + (z_3+r_3)^2 +(z_4+r_4)^2 = z_2^2 + z_3^2 + z_4^2 + \OO(4).
\end{equation}
Identifying the homogeneous terms of degree $3$ on both sides, we deduce that $z \cdot r = 0$, i.e.
\begin{equation}\label{eq-5p}
z_2 r_2 + z_3 r_3 + z_4 r_4 = 0.
\end{equation}

Thanks to the relation $z \cdot r = 0$, the last equation of \eqref{eq-5n} automatically holds when the first two do. Therefore, the system \eqref{eq-5n} reduces from $3$ to $2$ equations:
\begin{equation}\label{eq-5o}
\systeme{ z_3 h_4 - z_4 h_3 = r_2 \\ z_4 h_2 - z_2 h_4 = r_3}.
\end{equation}

4. We now focus on solving \eqref{eq-5o}. Thanks to the constraint \eqref{eq-5o}, we observe that $r_2$ has no terms in $z_2^2$. In particular, we can write it as $r_2 = z_3 h_4-z_4 h_3$ for some $h_3, h_4 \in \PP_1$. To solve \eqref{eq-5o}, hence \eqref{eq-5n}, it remains to find $h_2$.  

Again, thanks to  \eqref{eq-5o}, $r_3$ has no terms in $z_3^2$ and we can write it as $r_3 = z_4 \ell_2 - z_2 \ell_4$ with $\ell_2, \ell_4 \in \PP_1$. Plugging this and $r_2 = z_3 h_4-z_4 h_3$ in \eqref{eq-5o}, we obtain
\begin{equations}
0 = z_2 (z_3 h_4-z_4 h_3) + z_3 (z_4 \ell_2 - z_2 \ell_4) + z_4 r_4.
\\
\Rightarrow \qquad 0 = z_2 z_3 (h_4-\ell_4) + z_4 (r_4 + z_3 \ell_2 -z_2 h_3) = 0.
\end{equations}
In particular, $h_4-\ell_4 = \alpha z_4$ for some (smooth) function $\alpha$ of $z_1$. We then have
\begin{equation}
r_3 = z_4 \ell_2 - z_2 \ell_4 = z_4 \ell_2 - z_2 (h_4-\alpha z_4) = z_4 (\ell_2-\alpha z_2) - z_2 h_4. 
\end{equation}
It suffices then to set $h_2 = \ell_2 - \alpha z_2$ to solve \eqref{eq-5o}. This in turns produces $H$, hence a map $U$ defined on a neighborhood $\Omega$ of $0$, with values in $SU(2)$, such that \eqref{eq-2h} holds. 

5. It remains to extend the domain $\Omega$ of $U$ to $\R^4$. By considering the map $U(0)^* \cdot U$, we may assume that $U(0) = 0$. The exponential map from the Lie algebra $su(2)$ ($i$ times traceless $2 \times 2$ Hermitian matrices) to the Lie group $SU(2)$ is a diffeomorphism on a neighborhood of $0$ to a neighborhood of $\Id$. Therefore, after potentially shrinking $\Omega$, there exists a smooth map $A : \Omega \rightarrow su(2)$ such that $U(x,\xi) = e^{A(x,\xi)}$. It suffices then to extend $A$ to an arbitrary smooth function $\R^4 \rightarrow su(2)$, and $U$ accordingly by taking the exponential map. \end{proof}

\subsection{Expressions of $V_\di$ in terms of $\lambda_\di$ and $\kappa$:}
Having proved Theorem \ref{thm:1m}, we now compute the vector field $V_\di$ defined in \eqref{eq-10v}:

\begin{lemma}\label{lem:1m} Let $\kappa$ and $\lambda$ such that \eqref{eq-2h} holds. Then along $\Gamma_\di$, we have the relations:
\begin{equation}\label{eq-10n}
\lambda_\di \circ \kappa = \lambda, \qquad
V_\di = d\kappa \big( \lambda(x_1) \p_{x_1} \big).
\end{equation}
\end{lemma}

\begin{proof} 1. Thanks to \eqref{eq-2h} and \eqref{eq-12i}, we have, along $\Gamma_\di$, the relations
\begin{equation}
\lambda_\di \circ \kappa = \lambda_\dii, \qquad M_\dii = U^* \cdot M_\di \circ \kappa \cdot U, \qquad d\kappa(H_\dii) = U^* \cdot H_\di \cdot U.
\end{equation}
We deduce that along $\Gamma_\di$,
\begin{equation}
d\kappa(V_\dii) = -d\kappa \left(\dfrac{\Tr (M_\dii H_\dii)}{2\lambda^2_\dii}\right) = - \dfrac{\Tr (M_\di H_\di)}{2\lambda^2_\di} = V_\di. 
\end{equation}
The conclusion follows from $\lambda_\dii = \lambda$ and $V_\dii = \lambda(x_1)\p_1$ on $\R \times \{0\}^3$. 
\end{proof}

\section{Quantization}\label{sec-2.4}  In this section we prove Theorem \ref{thm:1d}. As a preliminary, we review some standard facts about pseudodifferential operators. Given a (potentially matrix-valued) symbol $A$, the Weyl quantization of $A$ is the operator
\begin{equation}
A^w u(x) = \int_{\R^4} e^{i\xi(x-y)/h} A\left( \dfrac{x+y}{2},\xi \right) u(y) \cdot \dfrac{dy d\xi}{(2\pi h)^2} ;
\end{equation}
we refer to $A^w$ as a pseudodifferential operator. We then have the following facts:
\begin{itemize}
\item[(a)] The adjoint of $A^w$ is pseudodifferential, with symbol $A^* + O(h)$; moreover $A^w$ is selfadjoint if and only if $A$ is Hermitian-valued.
\item[(b)] The composition of two pseudodifferential operators $A^w$ and $B^w$ is also pseudodifferential. Its symbol, commonly denoted $A \# B$, satisfies
\begin{equation}
A \# B = A B + \dfrac{h}{2i} \{A,B\} + O(h^2).
\end{equation}
\item[(c)] If $\det A$ is bounded below, then for $h$ sufficiently small the operator $A^w$ is invertible. It inverse is a pseudodifferential operator with symbol $A^{-1} + O(h)$.
\item[(d)] If $C= \Id+ O(h)$, then there exists a symbol $B = \Id + O(h)$ such that $(B^w)^* B^w = C^w$.
\item[(e)] If $U$ is such that $U^* U = \Id$, then there exists a symbol $\UU$ such that $\UU = U+O(h)$ and $(\UU^w)^* \UU = \Id$.
\end{itemize}

We now review a few facts about Fourier integral operators. Given a function $\phi$ satisfying the conclusions of Lemma \ref{lem:1q}, we define an operator $F$ by
\begin{equation}\label{eq-11s}
Fu(x) = \int_{\R^2} e^{\frac{i}{h}(\phi(x,\xi) - y\xi)} u(y) \cdot \big| \det \p_{x\xi} \phi(x,\xi) \big|^{1/2} \dfrac{dy d\xi}{(2\pi h)^2}.
\end{equation}
We then have the three following facts:
\begin{itemize}
\item[(f)] The operator $F$ is nearly unitary: $F^* F$ is a pseudodifferential operator with symbol $C = \Id + O(h)$;
\item[(g)] Egorov's theorem: if $A$ is a symbol bounded together with all its derivatives, then $F^* A^w F$  is a pseudodifferential operator and its symbol is $A \circ \kappa + O(h)$.
\item[(h)] Composing $F$ with a pseudodifferential operator produces an operator of the form \eqref{eq-10c}.
\end{itemize}

Standard references for these facts are \cite[\S4, \S11]{Z12} and  \cite[\S10]{GS94} (where $h=1$, but the proofs apply here without change).

\begin{proof}[Proof of Theorem \ref{thm:1d}] 1. Let $\kappa, \lambda, \mu, U$ be the quantities produced by Theorem \ref{thm:1m}. Let $\phi$ constructed according to Lemma \ref{lem:1q} and let $F$ be the Fourier integral operator \eqref{eq-11s}. By (g), $F^* \Di F$ is a pseudodifferential operator with symbol $\di \circ \kappa + O(h)$.

Moreover, by (f), $F^*F$ is a pseudodifferential operator with symbol $\Id+O(h)$. By (d), there exists $B = \Id + O(h)$ such that $F^* F = (B^w)^* B^w$. The operator $B$ is invertible, and its inverse has symbol $\Id+O(h)$ by (c). We deduce that $F_2 = F (B^w)^{-1}$ is unitary. Therefore, $F_2^* \Di F_2$ is a selfadjoint pseudodifferential operator, with symbol $\di \circ \kappa + O(h)$. 

Because of (e), there exists a symbol $\UU = U+O(h)$ such that $\UU^w$ is unitary. The operator $F_3 = \UU^w F_2$ is then unitary; hence $\FF_3^* \Di \FF_3$ is a Hermitian pseudodifferential operator with symbol
\begin{equation}\label{eq-11t}
U \cdot \di \circ \kappa \cdot U^* + O(h).  
\end{equation}
Expanding transversely to $(x_1, 0)$ the $O(h)$-remainder above, we may write \eqref{eq-11t} as 
\begin{equation}\label{eq-11v}
\big(U \cdot \di \circ \kappa \cdot U^*\big)(x,\xi) + h R(x_1) + \OO(1) O(h) + \OO(h^2). 
\end{equation}

3. We note that $R(x_1)$ is Hermitian-valued because $\FF_3$ is unitary and $\Di$, $U \cdot \di \circ \kappa \cdot U^*$ are Hermitian (as operators and symbols, respectively), see (a). Therefore, we may decompose $R$ as
\begin{equation}\label{eq-11w}
R(x_1) = s(x_1) + \lambda(x_1)\tR(x_1), \qquad \tR(x_1) = \sum_{j=1}^3 r_j(x_1) \sigma_j.
\end{equation}
Let $\dii_0 = x_2 \sigma_1 + \xi_2 \sigma_2 + \xi_1 \sigma_3$. We claim that there exists a scalar, never-vanishing symbol $c$ such that on a neighborhood of $0$,
\begin{equation}\label{eq-11u}
i\{ \dii_0,c \} + \tR c = \OO(1).
\end{equation}
Indeed, with the decomposition $\tR = \sum_{j=1}^3 r_j \sigma_j$, \eqref{eq-11u} reduces to three equations:
\begin{equation}\label{eq-11z}
\systeme{ -i\p_{\xi_2} c + r_1 c = \OO(1) \\ i \p_{x_2} c + r_2 c = \OO(1) \\  i \p_{x_1} c + r_3 c = \OO(1)}.
\end{equation}
We then verify that the symbol $c$ defined on a neighborhood of $0$ by 
\begin{equation}\label{eq-11y}
c(x,\xi) = \exp\left( i \int_0^{x_1} r_3 \right) \big( 1- i \xi_2 r_1(x_1) + i x_2 r_2(x_1) \big),
\end{equation}
and extended to a never-vanishing symbol on $\R^4$, satisfies \eqref{eq-11z} hence \eqref{eq-11u}.

4. Below we write $A \equiv B$ to denote two symbols $A$ and $B$ such that 
\begin{equation}
A-B = \OO(3)+ O(h^2) + \OO(1) O(h)
\end{equation}
on a neighborhood of $0$. We define the operator $\FF = c^w \FF_3$. By (h), this is an operator of the form \eqref{eq-10c}. By Step 2, $\FF \Di \FF^{-1}$ is a pseudodifferential operator. To compute its symbol (in a neighborhood of $0$), we use \eqref{eq-11v}; Theorem \ref{thm:1m}; the composition rule (b) and inversion rule (c); and \eqref{eq-11w}. This gives:
\begin{equations}\label{eq-11x}
c \# \big( U \cdot \di \circ \kappa \cdot U^* + h   R \big) \# \big( c^{-1} + O(h) \big) = c \# \big( \dii  + h R + \OO(3) \big) \# \big( c^{-1} + O(h) \big)
\\
\equiv \dii + \dfrac{h c }{2i} \{ \dii, c^{-1} \} +                                          \dfrac{ h c^{-1}}{2i} \{ c,\dii \} + h  R
\equiv
 \dii + h s + h c^{-1} \left(i \{ \dii,c \} + \lambda \tR c \right).
\end{equations}

From \eqref{eq-11y}, we note that $\p_{\xi_1} c = 0$. Since $\dii = (\lambda+\mu \xi_1) \dii_0$, we obtain
\begin{equation}
\{ \dii,c \} = \lambda \{ \dii_0,c \} + \mu \p_{x_1} c \cdot \dii_0 = \lambda \{ \dii_0,c \} + \OO(1).
\end{equation}
Plugging this in \eqref{eq-11x}, and using \eqref{eq-11u} for $c$, we conclude that the symbol of $\FF \Di \FF^{-1}$ is:
\begin{equation}
 \dii + h s + h c^{-1} \lambda \left(i \{ \dii_0,c \} +  \tR c \right) \equiv \dii + h  s.
\end{equation}
This completes the proof of Theorem \ref{thm:1d}.\end{proof}

\begin{remark}\label{rem:5} We comment briefly on the case $\pi(\Gamma_\di)$ and $\Gamma_\di$ homeomorphic to $\R$. In this case, the proof of Theorems \ref{thm:1m} and \ref{thm:1c} already produces $\kappa$ on neighborhood of arbitrary bounded subsets of $\Gamma_\di$; this uses essentially that vector bundles over $\Gamma_\di$ are trivial as $\Gamma_\di$ is simply connected \cite[\S1.3]{M01}. We can moreover ensure the existence of a generating function $\phi$ such that the conclusions of Lemma \ref{lem:1q} hold on neighborhoods of $(-L,L) \times \{0\}^3$, for arbitrary $L > 0$ (instead of simply a neighborhood of $0$). Therefore, the equation \eqref{eq-2u} holds with a remainder $R$ whose symbol is $\OO(3) + \OO(1) O(h) + \OO(h^2)$ on neighborhoods of $(-L,L) \times \{0\}^3$. 
\end{remark} 

\renewcommand{\thechapter}{C}

\chapter{Analysis of the model operator}\label{chap:3}

\section{Setting and main result} 

In this chapter we assume given two smooth functions $\lambda : \R \rightarrow \R^+$ and $\mu : \R \rightarrow \R$ and we assume that $\Dii$ is the operator that emerges in \S\ref{sec-2.4}:
\begin{align}
\Dii  = & \left( h \lambda(x_1) D_1 + \dfrac{h \lambda'(x_1)}{2i} + h^2 D_1 \mu(x_1) D_1  - h^2 \dfrac{\mu''(x_1)}{4}  \right) \sigma_3 
\\
\label{eq-7w}
& + \left( \lambda(x_1) + h \mu(x_1) D_1 + \dfrac{h \mu'(x_1)}{2i}\right) \matrice{0 & x_2 -h \p_2 \\ x_2 + h \p_2 & 0} 
 + h s(x_1).
\end{align}
We perform a semiclassical analysis for initial data originally concentrated (in phase space) at the point $0 \in \R \times \{0\}^3$.  Specifically, we show that solutions to
\begin{equation}\label{eq-7r}
\left( h D_t + \Dii \right) F = 0, \qquad F(0,x) = \dfrac{1}{\sqrt{h}} \cdot A\left( \dfrac{x}{\sqrt{h}} \right), \qquad A \in \SSS(\R^2)
\end{equation}
split, up to a small remainder in $L^2$, as:
\begin{itemize}
\item a semiclassical wavepacket propagating rightwards along $\R \times \{0\}^3$, at speed $\lambda(x_1)$;
\item an orthogonal wave that disperses very rapidly. 
\end{itemize}
In particular, coherent propagation only happens rightwards (here $\lambda > 0$; $\lambda< 0$ would yield leftward propagation). Forcing propagation in the opposite direction -- for instance by adequately preparing the initial data -- immediately destroys the coherent state structure.

In the framework of \S\ref{chap:2}, the operator $\Dii$ is a semiclassical operator, whose principal symbol 
\begin{equation}
\dii(x,\xi) = \big( \lambda(x_1) + \mu(x_1) \xi_1 \big) \matrice{\xi_1 & x_2-i\xi_2 \\ x_2 + i \xi_2 & -\xi_1}
\end{equation} 
satisfies $\bm{(\AAA_1})$-$\bm{(\AAA_2})$ from \S\ref{sec-2.1}. It forms a family of $2 \times 2$ Hermitian traceless matrices; its eigenvalues are repeated along $\dii^{-1}(0)$, which contains $\R \times \{0\}^3$. Moreover, along $\R \times \{0\}^3$ the vector field $V_\dii$ and scalar $\lambda_\dii$ from \S\ref{sec-2.1} are given by:
\begin{equation}
V_\dii(x_1,0) = \lambda(x_1) \p_1, \qquad
\lambda_\dii(x_1,0) = \lambda(x_1).
\end{equation}

To state quantitatively our main result, we first let $x_t^+ \in \R$ 
be the solution of the ODE
\begin{equation}\label{eq-7ua}
\dot{x^+_t} = \lambda(x_t^+), \qquad x_0^+ = 0.
\end{equation}
In particular, $(x_t^+,0)$ is the integral curve of $V_\dii$ starting at $0$. 
We then define
\begin{equation}\label{eq-7ub}
\rho_t \de \dfrac{\lambda(x_t)}{\lambda(0)}; \qquad \nu_t \de \int_0^t \dfrac{\mu(y_\tau)}{\rho_\tau^2} d\tau; \qquad S_t \de \int_0^t s(x_\tau^+) d\tau.
\end{equation}

We finally set: 
\begin{equation}\label{eq-7v}
a_0(x_1) \de \dfrac{1}{\sqrt{\pi}} \int_\R A(x_1,x_2) \cdot e^{-\frac{x_2^2}{2}} \matrice{1\\0} dx_2, \qquad 
a_t \de \dfrac{e^{-i\pi/4}}{|4\pi\nu_t|^{1/2}} \cdot e^{i \frac{x_1^2}{4 \nu_t}} \star a_0,
\end{equation}
where $\star$ denotes convolution with respect to the $x_1$-variable; and the Gaussian in the definition of $a_t$ is understood as a Dirac mass at $0$ when $\nu_t = 0$ (in  particular at $t=0$). 

\begin{theorem}\label{thm:1i} There exists $T > 0$ such that the solution $F$ of \eqref{eq-7r} satisfies uniformly, for $(t,x) \in (0,T] \times \R^2$ and $h \in (0,1]$:
\begin{equations}\label{eq-7s}
F(t,x) = \dfrac{e^{iS_t}}{\sqrt{\rho_th}} \cdot a_t \left( \dfrac{x_1-x_t^+}{ \rho_t\sqrt{h}} \right) e^{-\frac{x_2^2}{2h}} \matrice{1 \\ 0}  + O_{L^\infty}\big(h^{-1/4}\big) + O_{L^2}\big(h^{1/4}\big).
\end{equations}
\end{theorem}

The expansion \eqref{eq-7s} means that the dominant part of $F$ is a wavepacket propagating along $\R \times \{0\}$ at physical speed $\lambda(x_1)$ plus a term in $L^\infty$. While the latter generally has a significant $L^2$ mass, its amplitude is much smaller than that of the wavepacket. Visually, it simply collapses and hence loses its initial coherence. 

\begin{remark} The coefficient $\rho_t$ is uniformly bounded from above and below; hence the dilation by $\rho_t$ in \eqref{eq-7s} is uniformly controlled. On the other hand, $\nu_t$ may grow linearly with $t$, see \eqref{eq-7ub}. The implies that the envelop \eqref{eq-7v} may decay to $0$ as $t \rightarrow \infty$ (for instance, for magnetic fields \cite[\S6]{BBD22} or viscous models \cite{B19,SD+19}). This is a much slower dispersion than that associated with the $O_{L^\infty}\big(h^{-1/4}\big)$ remainder, which takes place as soon as time is turned on.  
\end{remark}

There are two cases where we can remove the limitation on $t$. The first scenario does not require additional assumptions on $\Dii$, but it demands a special initial data: 

\begin{corollary}\label{cor:1b} If $a(x) = a_0(x_1) e^{-x_2^2/2} [1,0]^\top$ then \eqref{eq-7s} holds for any fixed $T$, without the $O_{L^\infty}\big(t^{-1/2}\big)$-remainder.
\end{corollary}

In particular, Corollary \ref{cor:1b} predicts the existence of a wavepacket that propagates along $(x_t^+,0)$. The second scenario assumes that $\lambda_\di$ is constant:

\begin{corollary}\label{cor:2} If $\lambda$ is constant, then \eqref{eq-7s} holds for any fixed $t$.
\end{corollary}



To prove Theorem \ref{thm:1i}, we will perform an orthogonal decomposition of the initial data into Gaussian and higher-order Hermite modes in the variable $x_2$, see \S\ref{sec:3.1}. A coherent state expansion quickly produces the leading order term in \eqref{eq-7t}, see Theorem \ref{thm:1k}. Combining this result with Theorem \ref{thm:1d} allows us to recover the wavepacket constructions of \cite{BB+21,BBD22,HXZ22}; see \S\ref{sec:D} for applications. 

Higher-order Hermite modes are more challenging to analyze. We derive an oscillatory integral formula for their time-dynamics, with a more adequate semiclassical parameter roughly equal to $\sqrt{h}$. The oscillatory phase solves the eikonal equation
\begin{equation}\label{eq-7t}
\p_t \varphi + p(x_1, \p_1 \varphi) = 0, \qquad p(x_1,\xi_1) \de \lambda(x_1)\sqrt{1+(\p_{x_1} \varphi)^2}. 
\end{equation}
It admits short-time solutions, which extend to longer times if the bicharacteristic flow of the Hamiltonian vector field of $p$ is reversible; that is, if the map
\begin{equation}
x_1 \mapsto \pi \circ e^{tH_p}(x_1,\xi_1), \qquad \pi(x_1,\xi_1) = x_1
\end{equation}
is invertible for all $(t,\xi_1)$. Superposing these solutions produce a Lax-type parametrix. Its analysis through stationary and non-stationary bounds in oscillatory integrals eventually yields the $t^{-1/2}$-decay of \eqref{eq-7t}.

\section{Reduction to $2 \times 2$ systems on $\R$ and time-reversal considerations}\label{sec:3.1} 

\subsection{Semiclassical Hermite functions}\label{sec-H} Let $\aaa^* = x_2-\p_2$ and $\aaa = x_2+\p_2$ be the (semiclassical) creation and anhilation operators, and introduce the semiclassical Hermite functions
\begin{equation}
 g_0(x_2) \de \dfrac{1}{\pi^{1/4}} e^{-\frac{x_2^2}{2}}; \qquad g_n(x_2) \de \dfrac{(\aaa^*)^n}{2^{n/2}\sqrt{n!}} g_0(x_2). 
\end{equation}
They satisfy the relations $\aaa^* g_{n-1} =  \sqrt{2n} g_n$ and $\aaa g_n = \sqrt{2n} g_{n-1}$, where by convention $g_{-1} = 0$. For $f \in C^\infty(\R,\C^2)$, we define
\begin{equation}\label{eq-7x}
G_n(x_2) \de \matrice{g_n(x_2) \\ g_{n-1}(x_2)}, \qquad (f\otimes G_n) (x) \de \matrice{f_1(x_1) g_n(x_2) \\ f_2(x_1) g_{n-1}(x_2)}.
\end{equation}
We then set
\begin{equation}
g_{n,h}(x_2) \de \dfrac{1}{h^{1/4}} g_n\left(\dfrac{x_2}{\sqrt{h}}\right), \qquad G_{n,h}(x_2) \de \dfrac{1}{h^{1/4}} G_n\left(\dfrac{x_2}{\sqrt{h}}\right)
\end{equation}

\subsection{Block-diagonal decomposition}\label{sec-2.2} The relations between the Hermite functions and the creations/anhilation operators allow us to write $\Dii$ in the block-diagonal form
 \begin{equations}\label{eq-6n}
\Dii \left(\sum_{n=0}^\infty f_n\otimes G_{n,h}\right) =  L f_0 \otimes G_{0,h} +   \sum_{n=1}^\infty \sqrt{2n h} \big(\Dii_{n,h_n} f_n\big)\otimes G_{n,h}, \qquad h_n = \sqrt{\dfrac{h}{2n}}
\end{equations}
where the blocks $L$, $\Dii_{n,\epsilon}$ are differential operators in the variable $x_1$ only:
\begin{align}
L \de \ & h\lambda(x_1)  D_1 + \dfrac{h\lambda'(x_1)}{2i} + h^2 D_1 \mu(x_1) D_1 
- h^2 \dfrac{\mu''(x_1)}{4} + h s(x_1),
\\
\label{eq-7z}
\Dii_{n,\epsilon} \de \ & \left( \lambda(x_1) \epsilon D_1 + \dfrac{ \epsilon \lambda'(x_1)}{2i} + 2n \epsilon^3 D_1 \mu(x_1) D_1 - \dfrac{n \epsilon^3 \mu''(x_1)}{2} \right) \sigma_3 
\\
& + \left( \lambda(x_1) + 2n \epsilon^2 \mu(x_1) D_1  + \dfrac{n\epsilon^2 \mu'(x_1)}{i} \right) \sigma_1 + \epsilon s(x_1).
\end{align}

\subsection{Time-reversal considerations}\label{sec-2.3} The decomposition \eqref{eq-6n} reduces the system $(h D_t + \Dii) F = 0$ on $\R^2$ to decoupled one-dimensional systems on $\R$, of the two following types:
\begin{equations}\label{eq-7ya}
 (h D_t + L) f = 0; 
\end{equations}
\begin{equations}\label{eq-7yb}
(\epsilon D_t + \Dii_{n,\epsilon}) f = 0.
\end{equations}

A system with two internal degrees of freedom is time-reversal invariant if there exits an antiunitary operator $\TT$ such that $\TT^2 = -\Id$, and $\TT$ anticommutes with the Hamiltonian \cite{GP12,ASV13}. In particular, \eqref{eq-7yb} is time-reversal invariant: the antiunitary operator $\TT$ defined by
\begin{equation}
\TT \psi = \matrice{0 & -1 \\ 1 & 0} \ove{\psi} = i\sigma_2 \ove{\psi}.
\end{equation}
satisfies $\TT^2 = -\Id$ and $\Dii_{n,\epsilon} \TT = -\TT \Dii_{n,\epsilon}$. This implies that $\psi$ and $\TT \psi$ propagate in different directions of time. 

\begin{remark} The operator $\TT$ maps a wavepacket $\psi$ with orientation $u$ to one with complex-orthogonal orientation $\ove{u}^\perp$. If $\psi$ is concentrated at $(x_0,\xi_0) \notin \dii^{-1}(0)$ and $u$ is a positive-energy eigenvector of $\dii(x_0,\xi_0)$, then $\ove{u}^\perp$ is a negative-energy eigenvector of $\dii(x_0,\xi_0)$. From \cite{T75}, the dynamics of $\psi$ and $\TT\psi$ respectively follow bicharacteristics of the positive and negative eigenvalue of $\dii$, hence propagate in opposite directions of time. 
\end{remark}

On the other hand, the bulk-edge correspondence suggests that \eqref{eq-7r} is not time reversal-invariant. Since it is a direct sum of systems of the form \eqref{eq-7ya}-\eqref{eq-7yb}, we deduce that \eqref{eq-7ya} is not time-reversal invariant. Therefore, all the information about asymmetric/topological transport for $\Dii$ comes from \eqref{eq-7ya}. This justifies that the wavepackets constructed in \cite{BB+21,BBD22,HXZ22} are the dynamical analogues to topological edge states in curved settings. 

There are moreover analytic differences between \eqref{eq-7ya} and \eqref{eq-7yb}. Specifically, \eqref{eq-7ya} comes with a family of traveling waves, corresponding to the dominant term in \eqref{eq-7s}; and \eqref{eq-7yb} disperses rapidly and generates the $L^\infty$-remainder in \eqref{eq-7s}. We investigate these two effects in \S\ref{sec:3.2} and \S\ref{sec-4}-\ref{sec-5}.

\section{Propagating wavepackets}\label{sec:3.2} In this section, we construct an infinite-dimensional family of solutions to $(h D_t + \Dii) F = 0$ that at leading order travel along $\R \times \{0\}^3$, at speed $\lambda$. We realize the  corresponding initial data in the form of wavepackets with  profile transverse to $\R \times \{0\}$ given by $G_{0,h}$:
\begin{equation}\label{eq-6m}
(h D_t + \Di) F = 0, \qquad F(0,x) = \dfrac{1}{\sqrt{h}} \cdot a_0\left( \dfrac{x_1}{\sqrt{h}} \right) e^{-\frac{x_2^2}{2h}} \matrice{1 \\ 0}, \qquad a_0 \in \SSS(\R).
\end{equation}
We recall that the quantities $x_t^+, \rho_t, \nu_t, S_t, a_t$ are given by \eqref{eq-7ua}--\eqref{eq-7v}.

\begin{theorem}\label{thm:1k} For $a_0 \in \SSS(\R)$ and $t$ in compact sets, the solution to \eqref{eq-6m} satisfies:
\begin{equations}\label{eq-1k}
F(t,x) = \dfrac{e^{iS_t}}{\sqrt{\rho_th}} \cdot  a_t\left(  \dfrac{x_1-x_t^+}{ \rho_t\sqrt{h}} \right) e^{-\frac{x_2^2}{2h}} \matrice{1 \\ 0} + O_{L^2}\big(h^{1/2}\big).
\end{equations}
\end{theorem}

\begin{proof} 1. We observe that the initial data in \eqref{eq-6m} has profile transverse to $\R \times \{0\}$ given by $G_{0,h}$. According to \S\ref{sec:3.1}, the solution $F$ to \eqref{eq-6m} takes the form
\begin{equation}
F(t,x) = f(t,x_1) G_{0,h}(x_2),
\end{equation}
where $f$ solves $(h D_t + L) f = 0$. We will look for $f$ in the form
\begin{equation}
f(t,x_1) = \dfrac{1}{h^{1/4}} a \left(t,\dfrac{x_1-x_t^+}{\sqrt{h}} \right) \de \Ww[a](t,x_1), \qquad a(0,\cdot) = a_0.
\end{equation}

2. We note that
\begin{align}
(h D_t + L) \Ww[a]  & =  \sqrt{h}  W\big[\big(\lambda(x_t^+)-\dot{x_t^+}\big) D_1 a\big] 
\\
\label{eq-6z} & + h \Ww\left[ \left( D_t + \lambda'(x_t^+) x_1D_1 + \dfrac{\lambda'(x_t^+)}{2i} + \mu(x_t^+) D_1^2 + s(x_t^+)\right) a \right] + O_{L^2}(h^{3/2}).
\end{align} 
With $\dot{x_t^+} = \lambda(x_t^+)$ the first term vanishes. When the second term also vanishes, we end up with 
\begin{equation}\label{eq-13q}
(h D_t + L) \Ww[a] = O_{L^2}(h^{3/2})
\end{equation}
A standard Duhamel argument -- see e.g. \cite[Lemma 3.5]{BB+21} -- implies that $f = \Ww[a] + O_{L^2}(h^{1/2})$ for bounded times, hence 
\begin{equation}\label{eq-6t}
F(t,x) = \Ww[a](x_1) G_{0,h}(x_2) + O_{L^2}(h^{1/2}).
\end{equation}
Therefore, it remains to solve the equation
\begin{equation}\label{eq-6e}
\left(D_t + \lambda'(x_t^+) x_1D_1 + \dfrac{\lambda'(x_t^+)}{2i} + \mu(x_t^+) D_1^2 + s(x_t^+)\right) a = 0, \qquad a(0,\cdot)=a_0.
\end{equation}

3. The (spatial) Fourier transform of the differential operator in \eqref{eq-6e} is
\begin{equation}
D_t - \lambda'(x_t^+) \dfrac{\xi_1 D_1 + D_1 \xi_1}{2} + \mu(x_t^+) \xi_1^2 + s(x_t^+).
\end{equation} 
The operator $D_1 \xi_1 + \xi_1 D_1$ generates the semigroup of dilations $U_r f(\xi_1) = e^{r/2} f(e^r \xi_1)$. 
Therefore, if $r_t$ is defined by $\dot{r_t} = \lambda'(x_t^+)$ with $r_0 = 0$, then
\begin{equation}\label{eq-6c}
U_{r_t}^{-1} \left( D_t - \lambda'(x_t^+) \dfrac{\xi_1 D_1 + D_1 \xi_1}{2}  \right) U_{r_t} = \big(\dot{r_t} - \lambda'(x_t^+) \big) \dfrac{\xi_1 D_1 + D_1 \xi_1}{2} + D_t = D_t.
\end{equation}
We moreover note that $U_r^{-1} \xi_1^2 U_r = e^{-2r} \xi_1^2$, and that $e^{r_t} = \rho_t$. We therefore obtain the relation
\begin{equation}\label{eq-6h}
U_{r_t}^{-1} \left( D_t - \lambda'(x_t^+) \dfrac{\xi_1 D_1 + D_1 \xi_1}{2} + \mu(x_t^+) \xi_1^2 + s(x_t^+)\right) U_{r_t} =  D_t + \dfrac{\mu(x_t^+)}{\rho_t^2}\xi_1^2 + s(x_t^+).
\end{equation}
We deduce from \eqref{eq-6h} that the Fourier transform of the solution to \eqref{eq-6e} is
\begin{equation}\label{eq-6g}
\sqrt{\rho_t} e^{iS_t-i \nu_t \rho_t^2 \xi_1^2} \widehat{a_0}\left( \rho_t \xi_1 \right), \qquad \nu_t \de \int_0^t \dfrac{\mu(x_\tau)}{\rho_\tau^2} d\tau, \qquad S_t \de \int_0^t s(x_\tau^+) d\tau.
\end{equation}

The Fourier transform of \eqref{eq-6g} produces the solution to \eqref{eq-6e}:
\begin{equation}\label{eq-6i}
a(t,x_1) = \dfrac{e^{iS_t}}{\sqrt{\rho_t}} 
\left(\dfrac{e^{-i\pi/4}}{|4\pi\nu_t|^{1/2}}e^{i \frac{x_1^2}{4 \nu_t}} \star a_0 \right) \left( \dfrac{x_1}{\rho_t} \right).
\end{equation}
If $\nu_t = 0$ the above formula must be interpreted as a limit as $\nu_t \rightarrow 0$, i.e. the Gaussian in \eqref{eq-6i} is replaced by a Dirac mass. This completes the proof thanks to \eqref{eq-6t}. \end{proof}

\begin{remark}\label{rem:3} For the purpose of regularizing topological invariants in the bulk, \cite{B19} considered a Dirac operator with viscosity:
\begin{equation}\label{eq-0l}
\Dii  = \matrice{  D_1 + \mu D_1^2 & x_2 - \p_2 \\ x_2 +  \p_2 & -  D_1 - \mu D_1^2}, \qquad \mu \neq 0.
\end{equation}
Viscosity also originates naturally in topological models of active  fluids and plasmas, see  \cite{SD+19}. The traveling mode for \eqref{eq-0l} produced by Theorem \ref{thm:1k} is actually exact:
\begin{equation}
F(t,x) = \dfrac{e^{-i\pi/4}}{|4\pi\mu t|^{1/2}} \cdot e^{-\frac{x_2^2}{2}} \left( e^{i \frac{x_1^2}{4 \mu t}} \star a_0 \right)(x_1-t) \matrice{1 \\ 0}.
\end{equation}
Young's convolution inequality reveals then that $F(t,x) = O_{L^\infty}(t^{-1/2}e^{-x_2^2/2})$: the edge state slowly disperses (along the interface). Hence, while added viscosity does not affect the coarse topological properties of the system, it does affect the long-time dynamics.
\end{remark}

\section{WKB solutions}\label{sec-4}

In \S\ref{sec:3.2}, we constructed wavepacket solutions to $(h D_t + \Di) F = 0$, initially concentrated (in phase-space) at $(0,0)$ and propagating rightward along $\R \times \{0\}^3$, with speed $\lambda(x_1)$. Up to semiclassical rescaling, the $x_2$-profile is a Gaussian, while the $x_1$-profile is given by \eqref{eq-6i}.

We focus now on the second type of equation that arose in \S\ref{sec-2.2}:
\begin{equation}\label{eq-8a}
\left(\epsilon D_t + \Dii_{n,\epsilon} \right) f = 0, 
\end{equation} 
where $\Dii_{n,\epsilon}$ is the Dirac operator \eqref{eq-7z}:
\begin{align}
\Dii_{n,\epsilon} \de & \left( \lambda(x) \epsilon D_x + \dfrac{ \epsilon \lambda'(x)}{2i} + 2n \epsilon^3 D_x \mu(x) D_x - \dfrac{n \epsilon^3 \mu''(x)}{2} \right) \sigma_3 
\\
& \label{eq-8q} + \left( \lambda(x) + 2n \epsilon^2 \mu(x) D_x  + \dfrac{n\epsilon^2 \mu'(x)}{i} \right) \sigma_1.
\end{align}
In \eqref{eq-8q}, we intentionally dropped the subscript in $x_1, D_1$ and wrote $x, D_x$ instead; we will use this convention for all of \S\ref{sec-4} and most of \S\ref{sec-5}. 

We construct here WKB waves for \eqref{eq-8a}, i.e. solutions modulo $O(\epsilon^2)$ of the form $b e^{i\varphi/\epsilon}$, depending on $(t,x) \in (0,T) \in \R$ as well as on a frequency parameter $\xi \in \R$. The time $T$ essentially quantifies how long the bicharacteristic flow associated to \eqref{eq-8q} is reversible for. To state our result, we introduce 
\begin{equation}
\Omega_0 = [0,T) \times \R, \qquad 
\Omega = \Omega_0 \times \R = [0,T) \times \R^2
\end{equation}
and the concept of admissible function:

\begin{definition}\label{def:2} A function $q \in C^\infty(\Omega)$ depending implicitly on $\epsilon > 0$ and $n \in \N$ is admissible of order $m$ if there exists $m$ such that for all $\alpha \in \N^3$, uniformly in $n \in \N, \epsilon \in (0,1]$:
\begin{equation}\label{eq-7i}
\sup_{(t,x,\xi) \in \Omega}  \lr{\xi^2,n}^{-m-|\alpha|} \big| \p^\alpha q(t,x,\xi) \big| < \infty.
\end{equation}
If \eqref{eq-7i} holds for all $m$, we say that $q$ has order $-\infty$.
\end{definition}

\begin{theorem}\label{thm:1g} There exist $T > 0$ and $b, r, \varphi \in C^\infty_b(\Omega)$, such that:
\begin{itemize}
\item $\varphi$ is real-valued and solves the eikonal equation 
\begin{equation}\label{eq-11m}
\p_t \varphi + \lambda(x) \sqrt{1+(\p_x\varphi)^2} = 0, \qquad \varphi(0,x,\xi) = x\xi. 
\end{equation}
\item $b, r$ are $\C^2$-valued and admissible of respective orders $1$ and $4$,  with $|b(0,x,\xi)| = 1$ for $(x,\xi) \in \R^2$; and
\begin{equation}
\left( \epsilon D_t + \Dii_{n,\epsilon} \right) b e^{i\varphi/\epsilon} = \epsilon^2 r e^{i\varphi/\epsilon}.
\end{equation}
\end{itemize}
\end{theorem}

The function $\varphi$ will satisfy some important stationary and non-stationary estimates, see Theorem \ref{thm:1f} below. In \S\ref{sec-5}, we will show that superposing WKB waves produces a parametrix for \eqref{eq-8a}. Various non-stationary estimates on $\varphi$, proved in \S\ref{sec-4.3}, allow us to derive various dispersive bounds on the parametrix. In \S\ref{sec-4.5}, we use the time-reversal invariance of $\Dii_{n,\epsilon}$ (see \S\ref{sec-2.3}) to construct a WKB wave that travels in the opposite direction.

\subsection{Eikonal and transport equations}\label{sec-1.3} Theorem \ref{thm:1g}  consists of constructing solutions to $(\epsilon D_t + \Dii_{n,\epsilon}) f = O(\epsilon^2)$ in WKB form:
\begin{equation}\label{eq-2f}
f(t,x,\xi) = \exp \left( i \frac{\vp(t,x,\xi)}{\epsilon} \right) b(t,x,\xi), \qquad b = b_0+\epsilon b_1,  
\end{equation}
where $\varphi, b_0, b_1 \in C^\infty(\Omega)$ are independent of $\epsilon$ and satisfy the initial conditions 
\begin{equation}\label{eq-6x}
\vp(0,x,\xi) = x\xi, \qquad \big| b_0(0,x,\xi) \big| = 1, \qquad b_1(0,x,\xi) = 0.
\end{equation}

We now derive equations for $\varphi$, $b_0$ and $b_1$. Given $\varphi$, we have the operator-valued expansion
\begin{equations}\label{eq-6q}
\exp \left( -i \frac{\vp(t,x,\xi)}{\epsilon} \right) \left(\epsilon D_t + \Dii_{n,\epsilon}\right) \exp \left( i \frac{\vp(t,x,\xi)}{\epsilon} \right) = L_0+\epsilon L_1 + \epsilon^2 L_2,
\end{equations}
where the differential operators $L_0, L_1$ and $L_2$ are given by:
\begin{align}
L_0 & = \p_t \varphi + \lambda(x) \big( \sigma_1 + \sigma_3 \p_x\varphi \big),
\\
\label{eq-7o} L_1 & = D_t  + \left(\lambda(x) D_x + \dfrac{\lambda'(x)}{2i}\right) \sigma_1 + 2n \mu(x) \p_x\varphi  \big( \sigma_1 + \sigma_3 \p_x\varphi \big) + s(x_1),
\\
L_2 & = 2n \big(  \p_x \varphi \mu(x) D_x +   D_x \p_x \varphi \mu(x) + \epsilon D_x \mu(x) D_x \big) \sigma_3 +  n\big( 2 \mu(x) D_x  - i \epsilon^2 \mu'(x) \big) \sigma_1.
\end{align}
Hence, the equation $(\epsilon D_t + \Dii_{n,\epsilon}) f = O(\epsilon^2)$ has a solution of the form \eqref{eq-2f} if $L_0 b_0 = 0$ and $L_1 b_0 + L_0 b_1 = 0$. For the leading equation $L_0 b_0 = 0$ to have a solution $b_0 \neq 0$, $L_0$ must have a non-trivial kernel: $\det L_0 = 0$. This produces the eikonal equation:
\begin{equation}\label{eq-2s}
(\p_t \varphi)^2 = \lambda(x)^2\big( (\p_x \varphi)^2 + 1 \big), \qquad \varphi(0,x,\xi) = x\xi.
\end{equation}
Theorem \ref{thm:1g} focuses on one branch of solutions to \eqref{eq-2s}:
\begin{equation}\label{eq-8r}
\p_t \varphi + \lambda(x) \sqrt{ (\p_x \varphi)^2 + 1 }, \qquad \varphi(0,x,\xi) = x\xi.
\end{equation}
This is a Hamilton--Jacobi equation. We refer to \cite[\S10.2]{Z12} for the standard approach to construct WKB solutions. It consists first in studying the bicharacteristic flow of the underlying symbol, $p(x,\xi) = \lambda(x) \sqrt{1+\xi^2}$: the flow $x_0 \mapsto x_t$ generated by the ODE
\begin{equation}
\systeme{\dot{x_t} = \lambda(x_t) \dfrac{\xi_t}{\sqrt{1+\xi_t^2}} 
\\
\dot{\xi_t} = -\lambda'(x_t) \sqrt{1+\xi_t^2}}.
\end{equation}
In \S\ref{sec-4.2}, we prove that this flow is reversible for short time independent of the initial data $(x_0,\xi_0) \in \R^2$. This produces solutions to \eqref{eq-8r} expressed in terms of the reversed bicharacteristic flow. We will use this semi-explicit construction to prove non-stationary estimates for $\varphi$.

With $\varphi$ solution of \eqref{eq-8r}, the nullspace of $L_0$ is not empty, and one can solve the equation $L_0 b_0 = 0$. The solution $b_0$ takes the form $b_0 = \alpha u$, where $u$ is a normalized vector in the nullspace of $L_0$, and $\alpha$ is a scalar-valued function with $\alpha(0,x,\xi) = 1$. For the subleading equation $L_0 b_1 + L_1 b_0 = 0$ to have a solution $b_1$, the term $L_1 b_0$ must be orthogonal to the kernel of $L_0$. This gives the equation $L_1 (\alpha u_0) \perp \ker L_0$ for $\alpha$, which we solve in \S\ref{sec-4.4}

\subsection{Classical flow}\label{sec-4.2}  We study here the flow 
\begin{equation}\label{eq-3q}
\systeme{\dot{x_t} = \lambda(x_t) \dfrac{\xi_t}{\sqrt{1+\xi_t^2}} 
\\
\dot{\xi_t} = -\lambda'(x_t) \sqrt{1+\xi_t^2}}, \qquad (x_0,\xi_0) = (x,\xi) \in \R^2;
\end{equation}
see e.g. \cite[\S5]{GS94} for its relation with the Hamilton--Jacobi equation \eqref{eq-8r}. We will use the notation:
\begin{equation}
F(t,x,\xi) = x_t, \qquad G(t,x,\xi) = \xi_t, \qquad (t,x,\xi) \in \R^3
\end{equation}
We prove uniform estimates on $F$ and $G$ for $(x,\xi) \in \R^2$ and $t$ in compact sets. While $\R^2$ is not compact, the flow \eqref{eq-3q} (which is nearly homogeneous for large $\xi$) admits a natural compactification in $\xi$: with $\zeta_t = \arctan \xi_t$, $(x_t,\zeta_t)$ solves the ODE 
\begin{equation}\label{eq-3r}
\systeme{\dot{x_t} = \lambda(x_t) \sin \zeta_t 
\\
\dot{\zeta}_t = -\lambda'(x_t) \cos \zeta_t}, \qquad (x_0,\zeta_0) = (x,\zeta).
\end{equation}
Instead of looking at \eqref{eq-3r} for initial data $(x,\zeta) \in \R \times (-\pi/2,\pi/2)$, we will consider initial data $(x,\zeta) \in \R^2$. Formally speaking, \eqref{eq-3r} extends \eqref{eq-3q} beyond $\xi = \pm \infty$.

\begin{center}
\begin{figure}[!t]
\floatbox[{\capbeside\thisfloatsetup{capbesideposition ={right,center}}}]{figure}[ 10cm]
{\caption{Flow lines of \eqref{eq-3r}. The dark blue region corresponds to the physical space -- i.e. where \eqref{eq-3q} gets mapped to. The blue lines  correspond to $\xi = \pm \infty$ and separate the dark blue region from the unphysical light blue region. The flow does not connect the two regions.}\label{fig:2}}
{\begin{tikzpicture}
\node at (0,0) {\includegraphics[scale=1]{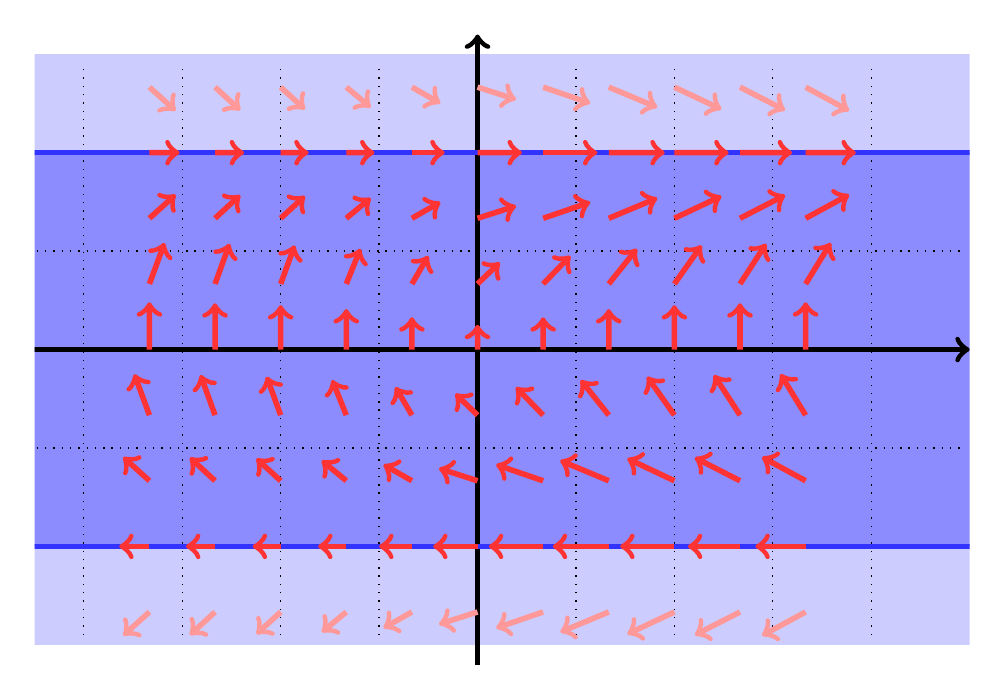}};
\end{tikzpicture}}
\end{figure}
\end{center}

We note that if $x_t, \zeta_t$ solve \eqref{eq-3r}, then $\dot{x_t}$ and $\dot{\zeta_t}$ are uniformly bounded; from Cauchy--Lipschitz theory, they are defined for all times. For $t \in \R$, we set
\begin{equation}
\tF(t,x,\zeta) = x_t, \quad \tG(t,x,\zeta) = \zeta_t, \quad (x,\zeta) \in \R^2,
\end{equation}
where $(x_t,\zeta_t)$ solves \eqref{eq-3r} with initial data $(x_0,\zeta_0) = (x,\zeta) \in \R^2$. While the ODE \eqref{eq-3r} is technically simpler than \eqref{eq-3q}, there is a simple relation relating the two flows:
\begin{equation}\label{eq-3y}
F(t,x,\xi) = \tF(t,x,\arctan \xi), \qquad G(t,x,\xi) = \tan \big( \tG(t,x, \arctan \xi) \big), \qquad (x,\xi) \in \R^2.
\end{equation} 

Later, we will need to consider the flow of \eqref{eq-3q} as $\xi$ approaches $\pm \infty$, which corresponds to $\zeta \rightarrow \pm \pi/2$. Since $\pi/2$ is a regular point of $\tF$, the following limits exist:
\begin{equation}\label{eq-8u}
F_\pm(t,x) \de \lim_{\xi \rightarrow \pm \infty} F(t,x,\xi) = \tF\big(t,x,\pm\pi/2\big), \qquad x_t^\pm = F_\pm(t,0).
\end{equation}
From \eqref{eq-3r}, the functions $F_\pm$ solve the autonomous first order ODE $\p_t F_\pm = \pm \lambda \circ F_\pm$; in particular, they admit an representation in terms of a antiderivative of $1/\lambda$:
\begin{equation}\label{eq-8t}
F_\pm(t,x) = \Lambda^{-1}\big(\Lambda(x) \pm t\big), \qquad x_t^\pm = \Lambda(\pm t), \qquad \Lambda(y) = \int_0^y \dfrac{dy}{\lambda(y)}.
\end{equation}


\begin{lemma}\label{lem:1h} All positive-order derivatives of $\tF, \tG$ are uniformly bounded for $t$ in compact sets and  $(x,\zeta) \in \R^2$. \end{lemma}

\begin{proof} 1. We write \eqref{eq-3r} in the form
\begin{equation}\label{eq-3u}
\dot{z_t} = \Phi(z_t), \qquad z_t = (x_t,\zeta_t), \qquad \Phi(x,\zeta) = \matrice{\lambda(x) \sin \zeta \\ -\lambda'(x) \cos \zeta}. 
\end{equation}
We note that $\Phi$ is bounded together with all its derivatives.  We now take $\alpha$ partial derivatives with respect to $(t,z)$ of \eqref{eq-3u}. A recursion yields
\begin{equation}\label{eq-3s}
\p^\alpha \dot{z_t} = p_\alpha  + \nabla \Phi(z_t) \cdot \p^\alpha z_t
\end{equation}
where $p_\alpha$ denotes a polynomial in the quantities $\p^\beta \Phi(z_t)$ and $\p^\gamma z_t$ for $|\beta| \leq |\alpha|$ and $|\gamma| \leq |\alpha|-1$. Note that $p_\alpha = 0$ when $|\alpha| = 1$.

2. We now show the lemma by strong recursion on $n=|\alpha|$. For $n=|\alpha| = 1$, we use \eqref{eq-3s}, $\nabla \Phi$ uniformly bounded, and $p_0 = 0$ to obtain
\begin{equation}\label{eq-3t}
\big| \p^\alpha \dot{z_t} \big| \leq C |\p^\alpha z_t|,  \qquad \big| \p^\alpha z_0 \big| = 1. 
\end{equation}
Set $r_t =  |\p^\alpha z_t|^2$. From \eqref{eq-3t} we deduce that $\dot{r_t} \leq C r_t$, $r_0 = 1$. Gronwall's lemma yields $r_t \leq e^{Ct}$. Therefore $\p^\alpha z_t$ is uniformly bounded for $t$ in compact sets and  $(x,\zeta) \in \R^2$.

We now assume that the property holds for all $\alpha$ of length smaller than some $n \geq 1$ and we prove it when $|\alpha| = n+1$. According to the recursion and $f$ bounded together with all its derivatives, the term $p_\alpha$ in \eqref{eq-3s} is uniformly bounded. We deduce that 
\begin{equation}
\big| \p^\alpha \dot{z_t} \big| \leq C + C |\p^\alpha z_t|,  \qquad \p^\alpha z_0 =  0. 
\end{equation}
Again setting $r_t =  |\p^\alpha z_t|^2$ we obtain $\dot{r_t} \leq C + C r_t$ with $r_0 = 0$. An application of Gronwall's lemma gives $r_t \leq e^{Ct}$. This completes the recursion. \end{proof}

\begin{lemma}\label{lem:1d} For every $(t,\zeta) \in \R^2$, the map $x \mapsto \tF(t,x,\zeta)$ has range $\R$. Moreover, there exists $T \in (0,1)$ such that for all $(t,x,\zeta) \in (0,T) \times \R^2$, $\big|\p_x\tF(t,x,\zeta) \big| \geq 1/2$.
\end{lemma}

\begin{proof} The range of $\tF(t,\cdot,\zeta)$ is an interval; therefore we have first to show that its supremum and infimum are $\pm \infty$. From \eqref{eq-3r},
\begin{equation}
\big| \tF(t,x,\zeta) -\tF(0,x,\zeta) \big| \leq |\lambda|_\infty t.
\end{equation}
Taking infimum and supremum over $x \in \R$ implies that $\tF(t,\cdot,\zeta)$ has range $\R$. Thanks to Lemma \ref{lem:1h}, $\p_t \p_x \tF \in C^\infty_b((0,1) \times \R^2)$. Moreover, $\p_x \tF(0,x,\xi) = \p_x x = 1$. Therefore,
\begin{equation}\label{eq-4s}
\big| \p_x \tF(t,x,\zeta) -1 \big| \leq t \cdot  \sup_{(0,1) \times \R^2} \big| \p_t \p_x \tF(t,x,\zeta) \big| \leq Ct.
\end{equation}
It suffices to restrict $t$ to sufficiently small values to deduce the lower bound on $\p_x \tF$.
\end{proof}

We recall that $\Omega_0 = (0,T) \times \R$ and $\Omega = (0,T) \times \R^2$.  Lemma \ref{lem:1d} implies that for any $(t,\zeta) \in \Omega_0$, the map $\tF(t,\cdot,\zeta)$ is a smooth diffeomorphism of $\R$. Its inverse $\tH(t,\cdot,\zeta)$ satisfies the relations:
\begin{equation}\label{eq-11a}
\p_t \tH = -\dfrac{\p_t \tF(t,\tH,\xi)}{\p_x \tF(t,\tH,\xi)}, \qquad \p_x \tH = -\dfrac{1}{\p_x \tF(t,\tH,\xi)}, \qquad \p_\zeta \tH = -\dfrac{\p_\zeta \tF(t,\tH,\xi)}{\p_x \tF(t,\tH,\xi)}.
\end{equation}
Thanks to Lemmas \ref{lem:1h} and \ref{lem:1d}, all derivatives of $\tH$ are in $C^\infty_b(\Omega)$.

We now set $H(t,x,\xi) = \tH(t,x,\arctan \xi)$, which satisfies $F(t,H(t,x,\xi),\xi) = x$. This implies that for $(t,\xi) \in (-T,T) \times \R$, $F(t,\cdot,\xi)$ is a smooth diffeomorphism of $\R$, with inverse $H(t,\cdot,\xi)$. We stress that since all derivatives of $\tH$ are bounded on $\Omega$, so are all derivatives of $H$.

\subsection{Eikonal equation}\label{sec-4.3} We recall that $\Omega= [0,T) \times \R^2$ and that $x_t^\pm$ are given by \eqref{eq-8t}, or equivalently \eqref{eq-8u}. We define
\begin{equation}
I_t \de \big( x_t^-, x_t^+ \big).
\end{equation}

\begin{theorem}\label{thm:1f} After potentially shrinking $T$:
\begin{itemize}
\item[(i)] The equation
\begin{equation}\label{eq-4a}
\p_t \varphi + p(x,\p_x \varphi) = 0, \qquad \varphi(0,x,\xi) = x\xi
\end{equation}
admits a solution $\varphi \in C^\infty(\Omega)$.

\item[(ii)] For all $\alpha \in \N^2, m \in \N$, there exists $C > 0$ with
\begin{equation}\label{eq-6r}
\left| \p_{x,t}^\alpha\p_\xi^m \big( \varphi(t,x,\xi) - x\xi\big) \right| \leq C \lr{\xi}^{1-m}, \qquad (t,x,\xi) \in \Omega. 
\end{equation}

\item[(iii)] There exists $c > 0$ such that:
\begin{equations}\label{eq-7p}
\big|\p_\xi \varphi(t,x,\xi)\big| \geq c \big( t \lr{\xi}^{-2} + \dist( x , I_t) \big), \qquad (t,x,\xi) \in \Omega, \qquad x \notin I_t;
\\
\p_\xi^2 \vp(t,x,\xi) \leq -c \cdot t \lr{\xi}^{-3}, \qquad (t,x,\xi) \in \Omega.
\end{equations} 
\end{itemize}
\end{theorem}

The geometric approach to solving \eqref{eq-4a} consists of showing that the manifold
\begin{equation}\label{eq-4b}
\big\{ (t,x_t,\tau,\xi_t) : \ (x_t,\xi_t) = e^{tH_p}(x,\xi), \ \tau + p(x_t,\xi_t) =0, \ (t,x,\xi) \in \Omega  \big\} \subset T^*\R^2,
\end{equation}
which is Lagrangian, projects diffeomorphically to $\Omega_0 \times \{0\}^2$; see e.g. \cite[\S5]{GS94}. This holds in our case because $F(t,\cdot,\xi)$ is invertible. The Lagrangian manifold \eqref{eq-4b} admits then a smooth generating function $\varphi \in C^\infty(\Omega)$; that is, there exists $\varphi \in C^\infty(\Omega)$ such that for $(t,x,\xi) \in \Omega$:
\begin{equation}\label{eq-4c}
(x_t,\xi_t) = e^{tH_p}(x,\xi) \quad \Rightarrow \quad \p_t \varphi(t,x_t,\xi) + p(x_t,\xi_t) = 0, \quad
\p_x \varphi(t,x_t,\xi) = \xi_t. 
\end{equation} 
In particular, $\p_x \varphi(0,x,\xi) = \xi$, therefore 
\begin{equation}
\varphi(0,x,\xi) = x\xi + \varphi(0,0,\xi). 
\end{equation}
We can assume without loss of generalities that $\varphi(0,0,\xi) = 0$, and therefore $\varphi$ is a solution to \eqref{eq-4a}. It admits a representation in terms of $(x_t,\xi_t)$: by computing the time-derivative of $\varphi(t,x_t,\xi)$ then integrating,
\begin{equation}\label{eq-4d}
\varphi(t,x_t,\xi) = x\xi - p(x,\xi) t + \int_0^t \xi_\tau \dd{p}{\xi}(x_\tau,\xi_\tau) d\tau = x\xi + \int_0^t \left(\xi_\tau \dd{p}{\xi}(x_\tau,\xi_\tau) - p(x_\tau,\xi_\tau)\right) d\tau.
\end{equation}
In the second equality we used that $p(x,\xi) = p(x_\tau,\xi_\tau)$ for any $\tau$. A direct computation shows that $\xi \p_\xi p(x,\xi) - p(x,\xi) = -\lambda(x) \lr{\xi}^{-1} = -\lambda(x)^2 p(x,\xi)^{-1}$. Therefore, \eqref{eq-4d} simplifies to
 \begin{equation}
\varphi(t,x_t,\xi) = x\xi - \dfrac{1}{p(x,\xi)}\int_0^t \lambda(x_\tau)^2 ds = x\xi - \dfrac{1}{\lambda(x) \sqrt{1+\xi^2}}\int_0^t \lambda^2 \circ F(\tau,x,\xi) d\tau.
\end{equation}
Using that $H(t,x_t,\xi) = x$, we obtain a semi-explicit formula for $\varphi$ on $\Omega$, proving in particular part (i) of Theorem \ref{thm:1f}:
\begin{equation}\label{eq-5b}
\varphi(t,x,\xi) = H(t,x,\xi) \xi - \dfrac{1}{\sqrt{1+\xi^2} \cdot \lambda \circ H(t,x,\xi) }\int_0^t \lambda^2 \circ F\big(s,H(t,x,\xi),\xi\big) ds.
\end{equation}

We now prove part (ii):

\begin{lemma}\label{lem:1i} The solution $\varphi$ to \eqref{eq-4a} 
satisfies the  bound \eqref{eq-6r}.
\end{lemma} 

\begin{proof} 1. We define $\tvarphi \in C^\infty\big( (0,T) \times \R \times (-\pi/2,\pi/2) \big)$ by
\begin{equation}\label{eq-4i}
\tvarphi(t,x,\zeta) = \tH(t,x,\zeta) \tan \zeta - \dfrac{\cos \zeta}{\lambda \circ H(t,x,\zeta)}\int_0^t \lambda^2 \circ \tF\big(s,\tH(t,x,\zeta),\zeta\big) ds.
\end{equation}
The heuristic behind this definition comes from the relations \eqref{eq-3y} between $F, \tF, H$ and $\tH$. They connect $\varphi$ and $\tvarphi$ through
\begin{equation}\label{eq-4e}
\varphi(t,x,\xi) = \tvarphi(t,x,\arctan \xi), \qquad (t,x,\xi) \in X.
\end{equation}

2. We focus on the first term in \eqref{eq-4i}. Using $\tH(t,x,\pi/2) = F_-(t,x)$, we obtain:
\begin{equations}\label{eq-4j}
\tH(t,x,\zeta) \tan \zeta = \tH\left(t,x,\dfrac{\pi}{2}\right) \tan \zeta + \left(\tH(t,x,\zeta) - \tH\left(t,x,\dfrac{\pi}{2}\right)\right) \tan \zeta 
= F_-(t,x) \tan \zeta + \psi(t,x,\zeta), \\
 \qquad \psi(t,x,\zeta)  \de
\left( \tH(t,x,\zeta) -  F_-(t,x)\right) \tan \zeta = \left( \zeta-\dfrac{\pi}{2} \right)\tan \zeta \cdot \int_0^1 \dd{\tH}{\zeta}\left(s,x,\te \zeta + (1-\te) \dfrac{\pi}{2}\right) d\te.
\end{equations}
The identities defining $\psi$ show that it is a function on $(-T,T) \times \R \times (-\pi/2,\pi/2)$ that extends smoothly to $(-T,T) \times \R \times (-\pi/2,3\pi/2)$. Moreover, since $\p_\zeta\tH$ is uniformly bounded on $\Omega = \Omega_0 \times \R$, $\psi$ and all its derivatives are uniformly bounded on e.g. $\Omega_0 \times (0,\pi)$. 

3. The second term in \eqref{eq-4i} defines a function in $C^\infty_b(X)$, because $\lambda \in C^\infty_b(\R)$ and all derivatives of $\tF, \tH$ are uniformly bounded on $X$. Thanks to Step 2 and the relation \eqref{eq-4e}, we deduce that we can we write
\begin{equation}\label{eq-4r}
\tvarphi(t,x,\zeta) = F_+(t,x) \tan \zeta + R(t,x,\zeta), \qquad \varphi(t,x,\xi) = F_-(t,x) \xi + R(t,x,\arctan \xi),
\end{equation}
where $R \in C^\infty_b\big( \Omega_0 \times (0,\pi) \big)$. 
 To prove the bound \eqref{eq-6r}, we write, thanks to \eqref{eq-4r}:
\begin{equation}\label{eq-6s}
\varphi(t,x,\xi) - x\xi = \xi\big( F_-(t,x) - x \big) + R(t,x,\arctan \xi) = \xi \int_0^t \dd{\tF(s,x,\pi/2)}{s} ds +  R(t,x,\arctan \xi).
\end{equation}
The above integral defines a function in $C^\infty_b(\Omega_0)$ because derivatives of $\tF$ are bounded, see Lemma \ref{lem:1h}. Moreover $R \in C^\infty_b\big( \Omega_0 \times (0,\pi) \big)$, therefore $R(t,x,\arctan \xi)$ defines an expression in $C^\infty_b\big(\Omega_0 \times \R^+\big)$. It follows that $\p_\xi (\varphi(t,x,\xi) - x\xi)$ satisfies the bound \eqref{eq-6r} for $(t,x,\xi) \in \Omega_0 \times \R^+$. A similar approach yields the result with $\R^-$ replacing $\R^+$.
\end{proof}

In Lemmas \ref{lem:1k}-\ref{lem:1g} below, we prove stationary and non-stationary estimates on $\vp$. These will serve to prove localization and spreading estimates for superposition of WKB waves (see \S\ref{sec-5}).

\begin{lemma}\label{lem:1k} After potentially shrinking $T$, there exists $c > 0$ such that for $(t,x,\xi) \in X$:
\begin{equation}\label{eq-4q}
\p_\xi^2 \vp(t,x,\xi) \leq - ct \lr{\xi}^{-3}.
\end{equation}
\end{lemma}

\begin{proof} 1. We will rely on the formula $\p_\xi \varphi = H$. This comes from the fact that $\p_\xi \varphi(t,x_t,\xi)$ does not depend on $t$:
\begin{equations}
\dd{}{t} \big( \p_\xi \varphi(t,x_t,\xi) \big) = \dd{^2\varphi(t,x_t,\xi)}{\xi \p t} + \dot{x_t}  \dd{^2\varphi(t,x_t,\xi)}{x\p\xi} = -\dd{\big( p(x,\p_x \varphi)\big) }{\xi} (t,x_t,\xi) + \dot{x_t}  \dd{^2\varphi(t,x_t,\xi)}{x\p\xi}
\\
= -\dd{ p }{\xi}(x_t,\xi_t) \dd{^2\varphi(t,x_t,\xi)}{x\p\xi} + \dot{x_t}  \dd{^2\varphi(t,x_t,\xi)}{x\p\xi} = 0,
\end{equations}
where we used the eikonal equation and $\dot{x_t} = \p_\xi p(x_t,\xi_t)$. We therefore deduce
\begin{equation}\label{eq-4k}
\p_\xi \varphi(t,x_t,\xi) = \p_\xi \varphi(0,x_0,\xi) = x_0 = H(t,x_t,\xi), 
\end{equation}
and therefore $\p_\xi \varphi = H$.

From this identity, we deduce that Lemma \ref{lem:1k} is equivalent to $\p_\xi H(t,x,\xi) \leq -c t \lr{\xi}^{-3}$ for $t > 0$ small enough and $(x,\xi) \in \R^2$. We observe moreover that
\begin{equation}
\p_\xi H(t,x,\xi) = \lr{\xi}^{-2} \p_\zeta\tH(t,x,\arctan \xi),
\end{equation}
therefore \eqref{eq-4q} is equivalent to $\p_\zeta\tH(t,x,\zeta) \leq -c t \cos \zeta$ for $t > 0$ small enough and $(x,\zeta) \in \R \times (-\pi/2,\pi/2)$.

2. Taking derivatives with respect to $\zeta$ then $t$ of the relation $F(t,\tH, \zeta) = \Id$, we obtain 
\begin{equations}\label{eq-4m}
\p_\zeta \tF(t,\tH,\zeta) + \p_\zeta \tH \cdot \p_x \tF(t,\tH,\zeta) = 0, \\
\p_t \p_\zeta\tF(t,\tH,\zeta) + \p_t \tH \cdot \p_x \p_\zeta \tF(t,\tH,\zeta) + \p_t \p_\zeta \tH \cdot \p_x \tF(t,\tH,\zeta) 
\\
+ \p_\zeta \tH \cdot \p_t \p_x \tF(t,\tH,\zeta) + \p_\zeta \tH \p_t \tH \cdot \p_x^2 \tF(t,\tH,\zeta) = 0.
\end{equations}
Because of $\tF(0,x,\zeta) = \tH(0,x,\zeta) = x$, we know that $\p_x \p_\zeta \tF$ and $\p_\zeta \tH$ vanish at $t=0$, and $\p_x \tF(0,x,\zeta) = 1$. Therefore, evaluating \eqref{eq-4m} at $t=0$ produces 
\begin{equations}\label{eq-8m}
\p_t \p_\zeta\tF(0,x,\zeta) + \p_t \p_\zeta \tH(0,x,\zeta) = 0.
\end{equations}

3. We now compute $\p_t \p_\zeta \tF(0,x,\zeta) = \p_\zeta \dot{x_t}$ with $\tF(t,x,\zeta) = \dot{x_t}$ and $\tG(t,x,\zeta) = \zeta_t$. Taking derivatives of \eqref{eq-3r} with respect to $\zeta$,
\begin{equation}\label{eq-4o}
\p_\zeta \dot{x_t} = \lambda'(x_t) \sin \zeta_t \cdot \p_\zeta x_t + \lambda(x_t) \cos \zeta_t \cdot \p_\zeta \zeta_t.
\end{equation}
Since $\p_\zeta \zeta_0 = 1$ and $\p_\zeta x_0 = 0$, we conclude that $\p_t \p_\zeta \tF(0,x,\zeta) = \lambda(x) \cos \zeta$. From \eqref{eq-8m}, we deduce that $\p_t \p_\zeta \tH(0,x,\zeta) = -\lambda(x) \cos \zeta$. Hence, thanks to Taylor's formula and $\p_\zeta\tH \in C^\infty_b(X)$,
\begin{equation}
\left| \p_\zeta \tH(t,x,\zeta) + t \lambda(x) \cos \zeta \right| \leq Ct^2.
\end{equation}
This inequality shows that \eqref{eq-4q} is valid as long as $\inf_\R \lambda \cdot \cos \zeta \geq 2Ct$.

4. We now study the situation for $\zeta$ near $\pi/2$. We write Taylor's inequality:
\begin{equation}\label{eq-4n}
\left| \p_\zeta \tH (t,x,\zeta) - \sum_{j=0}^2 \p_\zeta^{j+1} \tH \left(t,x,\dfrac{\pi}{2}\right) \left( \zeta-\dfrac{\pi}{2}\right)^j \right| \leq C \left|\zeta-\dfrac{\pi}{2}\right|^3.
\end{equation}

We note that $\tG(t,x,\pi/2) = \pi/2$, therefore with $x_t = \tF(t,x,\pi/2)$ we have from \eqref{eq-4o}:
\begin{equation}
\p_\zeta \dot{x_t} = \lambda'(x_t) \p_\zeta x_t, \qquad \p_\zeta x_0 = 0. 
\end{equation}
It follows that $\p_\zeta x_t = 0$ for every $t$; that is, $\p_\zeta \tF(t,x,\pi/2)= 0$. From the first identity in \eqref{eq-4m} and $\p_x \tF(t,\tH,\zeta) = 1 + O(t)$, we deduce that $\p_\zeta \tH(t,x,\pi/2) = 0$ for $t$ sufficiently small. Taking derivatives of $\p_t \p_\zeta \tH(0,x,\zeta) = \lambda(x) \cos \zeta$ with respect to $\zeta$, we obtain
\begin{equation}
\p_t \p_\zeta^2 \tH\left(0,x,\dfrac{\pi}{2}\right) = \lambda(x), \qquad \p_t \p_\zeta^3 \tH\left(0,x,\dfrac{\pi}{2}\right) = 0.
\end{equation}
Therefore, using that $\p_\zeta^j H(0,x,\zeta) = 0$ for $j \geq 1$, we deduce that
\begin{equation}
\left| \p_\zeta^2 \tH\left(t,x,\dfrac{\pi}{2}\right) - \lambda(x)t \right| \leq Ct^2, \qquad \left| \p_\zeta^3 \tH\left(t,x,\dfrac{\pi}{2}\right) \right| \leq Ct^2.
\end{equation}
Plugging these in \eqref{eq-4n}, we deduce that 
\begin{equation}
\left| \p_\zeta \tH (t,x,\zeta) - \lambda(x) t \left( \zeta-\dfrac{\pi}{2}\right) \right| \leq C \left|\zeta-\dfrac{\pi}{2}\right|^3 + Ct^2  \left|\zeta-\dfrac{\pi}{2}\right| + Ct^2 \left|\zeta-\dfrac{\pi}{2}\right|^2 \leq C \left|\zeta-\dfrac{\pi}{2}\right|^3 + Ct^2 \left|\zeta-\dfrac{\pi}{2}\right|.
\end{equation}
It follows that for $(t,x,\zeta) \in X_0 \times (0,\pi)$,
\begin{equation}
\left| \p_\zeta \tH (t,x,\zeta) + \lambda(x) t \cos \zeta \right| \leq C\cos \zeta\left( \cos^2 \zeta + t^2\right).
\end{equation}
This inequality shows that \eqref{eq-4q} is valid as long as $\inf_\R \lambda \cdot t \geq C\left( \cos^2 \zeta + t^2\right)$, that is $t$ sufficiently small and $t \geq C\cos^2 \zeta$. This region of validity combined with that identified in Step 3 cover a whole strip of times near $0$, see Figure \ref{fig:1}. This proves the lemma. \end{proof}

\begin{figure}[!t]
\floatbox[{\capbeside\thisfloatsetup{capbesideposition ={right,center}}}]{figure}[\FBwidth]
{\caption{Step 3 shows that \eqref{eq-4q} holds within  the blue region. It therefore misses positive times for $\zeta$ near $\pm \pi/2$. Step 4 shows that
\eqref{eq-4q} holds in the red-enclosed region, hence covers points $(t,\zeta)$ missed by Step 3.}\label{fig:1}}
{\begin{tikzpicture}
\node at (0,0) {\includegraphics[scale=0.8]{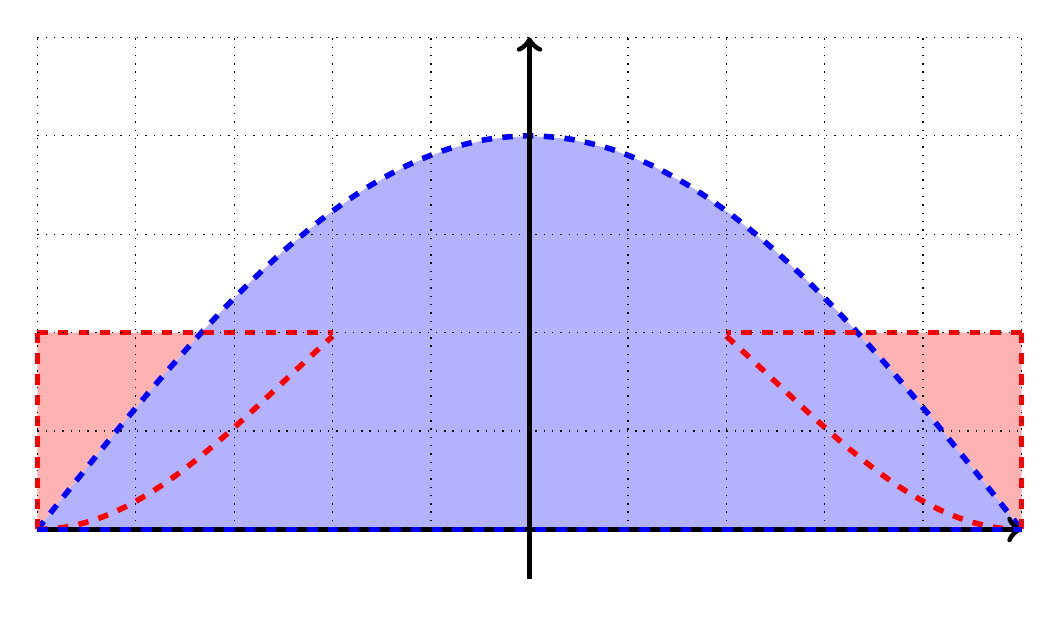}};
\end{tikzpicture}}
\end{figure}

Lemma \ref{lem:1i} provides the bound \eqref{eq-6r} and Lemma \ref{lem:1g} provides the second estimate in \eqref{eq-7p}. To complete the proof of Theorem \ref{thm:1f}, it remains to show the first estimate in \eqref{eq-7p}:

\begin{lemma}\label{lem:1g} There exists a constant $c > 0$ such that
\begin{equations}\label{eq-3f}
\big|\p_\xi \varphi(t,x,\xi)\big| \geq c \big( t \lr{\xi}^{-2} + \dist( x , I_t ) \big), \qquad (t,x,\xi) \in \Omega, \qquad x \notin I_t.
\end{equations}
\end{lemma}

\begin{proof}  We assume that $x \geq x_t^+$; the case $x \leq x_t^-$ is treated similarly. We first note that
\begin{equation}\label{eq-8n}
\lim_{\xi \rightarrow + \infty} \p_\xi \varphi(t,x_t^+,\xi) = \lim_{\xi \rightarrow +\infty} H(t,x_t^+,\xi) = \lim_{\zeta \rightarrow \pi/2} \tH(t,x_t^\pm,\zeta) = \tH\big(t,x_t^+,\pi/2\big) = 0.
\end{equation}
Thanks to \eqref{eq-8n}:
\begin{equations}\label{eq-8k}
\p_\xi \varphi(t,x,\xi) = \p_\xi \varphi(t,x,\xi) - \p_\xi \varphi(t,x_t^+,\xi) + \p_\xi \varphi(t,x_t^+,\xi) 
\\
= \int_{x_t^+}^x \p_{x\xi}^2  \varphi(t,s,\xi) ds - \int_\xi^{+\infty} \p_\xi^2 \varphi(t,x_t^+,s)ds.
\end{equations}
We note that $\p_{x\xi}^2  \varphi(0,s,\xi) = 1$ and $\p_{x\xi}^2  \varphi \in C^\infty_b(\Omega)$, hence $\p_{x\xi}^2  \varphi(t,s,\xi) = 1 + O(t)$ for $t$ sufficiently small, uniformly in $(x,\xi) \in \R^2$. For the term $\p_\xi^2  \varphi(t,x,s)$, we use Lemma \ref{lem:1k}. We deduce that for $t$ sufficiently small,
\begin{equations}
\p_\xi \varphi(t,x,\xi) \geq (x-x_t^+)\big( 1 + O(t) \big) + ct \int_\xi^\infty \lr{s}^{-3} ds \geq c \left( x-x_t + \dfrac{t}{\lr{\xi}^2} \right).
\end{equations}
This proves \eqref{eq-3f} when $x \geq x_t^+$. When $x \leq x_t^-$ instead, we use the identity
\begin{equations}
\p_\xi \varphi(t,x,\xi) = \int_{x_t^-}^x \p_{x\xi}^2  \varphi(t,s,\xi) ds + \int_{-\infty}^\xi \p_\xi^2 \varphi(t,x_t^+,s)ds,
\end{equations}
and a similar reasoning. This completes the proof. \end{proof}

The next result shows that $I_t$ is the smallest possible set such that \eqref{eq-3f} holds. Standard techniques imply that superpositions of WKB waves are essentially supported within the cone $\{(t,x) : x\in I_t\}$; see \S\ref{sec-5}.

\begin{corollary}\label{cor:5} For $t \in (0,T)$ and $x \in I_t = (x_t^-,x_t^+)$, the equation $\p_\xi \varphi(t,x,\xi) = 0$ has a unique solution denoted $\Xi(t,x)$. It depends smoothly on $(t,x) \in (0,T) \times \R$; moreover, there exists $c > 0$ such that
\begin{equation}\label{eq-8l}
\big|\Xi(t,x)\big| \geq c \left| \dfrac{t}{x - x_t^\pm} \right|^{1/2}, \qquad t \in (0,T), \quad x \in I_t.
\end{equation}
\end{corollary}

\begin{proof} Since $\p_\xi \varphi(t,x,\cdot)$ is decreasing (see Lemma \ref{lem:1k}), the equation $\p_\xi \varphi(t,x,\xi) = 0$ has a unique solution. Moreover, using that $\p_\xi \varphi = H$, the equation $\p_\xi \varphi(t,x,\xi) = 0$ is equivalent to $x = F(t,0, \xi)$. But for $x \in I_t$:
\begin{equation}
\lim_{\xi \rightarrow -\infty} F(t,0,\xi) = x_t^- < x < x_t^+ = \lim_{\xi \rightarrow +\infty} F(t,0,\xi).
\end{equation}
The intermediate value theorem produces then a solution $\Xi(t,x)$ to $x = F(t,0,\xi)$, which is unique. Smoothness with respect to $(t,x)$ follows from the implicit function theorem, noting that $\p_\xi^2 \varphi\big(t,x,\Xi(t,x)\big) \neq 0$.

It remains to prove \eqref{eq-8l}. This is a statement about the behavior of $\Xi(t,x)$ as $x$ approaches $x_t^\pm$. we use \eqref{eq-8k} specialized at $\xi = \Xi(t,x)$:
\begin{equations}
0 = \p_\xi \varphi\big(t,x,\Xi(t,x)\big) = \int_{x_t^+}^x \p^2_{x\xi} \varphi\big(t,s,\Xi(t,x)\big) ds -\int_{\Xi(t,x)}^{+\infty} \p_\xi^2 \varphi(t,x_t^+,s) ds
\\
\geq \int_{x_t^+}^x \p^2_{x\xi} \varphi\big(t,s,\Xi(t,x)\big) ds + ct \Xi(t,x)^{-2}.
\end{equations}
where we used Lemma \ref{lem:1k} for the lower bound.
In particular, using that $\p_{x\xi}^2 \varphi$ is uniformly bounded:
\begin{equation}
C (x-x_t^+) \geq \int_x^{x_t^+} \p^2_{x\xi}\varphi\big(t,s,\Xi(t,x)\big) ds  \geq ct \Xi(t,x)^{-2}; \qquad \text{equivalently} \qquad \Xi(t,x) \geq c \sqrt{\dfrac{t}{x_t^+-x}}.
\end{equation}
A similar estimate holds as $x \rightarrow x_t^-$.\end{proof}

Later, we will need bound on the function $\vartheta$ defined by:
\begin{equation}
\vartheta(t,x) = \varphi\big(t,x,\Xi(t,x)\big), \qquad t \in [0,T], \qquad x \in I_t.
\end{equation}

\begin{lemma}\label{lem:1w} There exists $C > 0$ such that for $j \in [1,3]$:
\begin{equation}
\big| \p_x^j \vartheta(t,x) \big| \leq C \dist(x,\R \setminus I_t)^{-3j}, \qquad t \in [0,T], \qquad x \in I_t.
\end{equation}
\end{lemma}

\begin{proof} 1. Since $\p_\xi \varphi(t,x,\Xi) = 0$, we have $\p_x \vartheta(t,x) = \p_x \varphi(t,x,\Xi)$. Taking then successive derivatives with respect to $x$ of this expression, noting that $|\p_x^j \varphi(t,x,\xi)| \leq C \lr{\xi}$ and $\p_\xi^k \p_x^\ell \varphi$ is uniformly bounded for $k \geq 1$ produces the inequalities:
\begin{equation}\label{eq-10z}
\big| \p_x \vartheta \big| \leq C \lr{\Xi}, \qquad 
\big| \p_x^2 \vartheta \big| \leq C \big(\lr{\Xi} + \lr{\p_x \Xi}\big), \qquad 
\big| \p_x^3 \vartheta \big| \leq C \big(\lr{\Xi} + \lr{\p_x \Xi}^2 + \lr{\p_x^2\Xi}\big).\end{equation}
It remains now to bound $\Xi$ and its derivatives. 

2. Taking derivatives of the expression $\p_\xi \varphi(t,x,\Xi) = 0$ with respect to $\xi$, we obtain -- with all identities evaluated at $\xi = \Xi(t,x)$:
\begin{equation}
\p_x \Xi \cdot \p_\xi^2 \varphi + \p_x \varphi = 0, \qquad
\p_x^2 \Xi  \cdot \p_\xi^2 \varphi + (\p_x \Xi)^2 \cdot \p_\xi^3 \varphi + \p_x^2 \varphi + \p_x \Xi \cdot \p^2_{x\xi} \varphi = 0.  
\end{equation}
Using the lower bound on $\p_\xi^2 \varphi$ (Lemma \ref{lem:1k}) and the upper bounds on $\p^\alpha \varphi$ (Lemma \ref{lem:1i}), we deduce that
\begin{equation}\label{eq-10y}
\big| \p_x \Xi \big| \leq C\lr{\Xi}^3, \qquad \big| \p_x^2 \Xi \big| \leq \lr{\Xi}^8.
\end{equation}
Plugging the bounds \eqref{eq-10y} in \eqref{eq-10z}, we deduce that
\begin{equation}\label{eq-11b}
\big| \p_x^j \vartheta \big| \leq C \lr{\Xi}^{3j}, \qquad j \in [1,3].
\end{equation}

3. We now prove an upper bound on $\Xi$. We recall that $\p_\xi \varphi = H$, hence $H(t,x,\Xi) = 0$. It follows that
\begin{equation}
H_+(t,x) = \lim_{\xi \rightarrow +\infty} H(t,x,\xi) - H(t,x,\Xi) = \int_{\Xi(t,x)}^{+\infty} \p_\xi H(t,x,s) ds. 
\end{equation}
Moreover, $\p_\xi H = \lr{\xi}^{-2} \p_\xi \tH$, and $\p_\xi \tH$ is uniformly bounded: $\tH$ is the reciprocal of $\tF$, whose derivative with respect to $x$ is bounded below, see Lemmas \ref{lem:1h}-\ref{lem:1d}. We deduce that when $\Xi > 0$,
\begin{equation}
\big| H_+(t,x) \big| \leq C\int_\Xi^{+\infty} s^{-2} ds = \dfrac{C}{\Xi},
\end{equation}
therefore $\Xi(t,x) \leq C \big| H_+(t,x) \big|^{-1}$. We likewise have $-\Xi(t,x) \geq C \big| H_-(t,x) \big|^{-1}$ when $\Xi(t,x) < 0$. Therefore an upper bound on $\Xi$ follows from a lower bound on $H_+$ and $H_-$. Thanks to Lemma \ref{lem:1h} and \eqref{eq-11a}, $\p_x \tH$ hence $\p_x H$ are bounded below on $\Omega$, hence
\begin{equation}
H_+(t,x) = H_+(t,x) - H_+(t,x_t^+) = \int^x_{x_t^+} \p_x H_+(t,y)dy \geq c (x_t^+-x). 
\end{equation}
A similar bound holds for $\p_x H_-$. We conclude that $|\Xi(t,x)| \leq C \dist(x,\R \setminus I_t)^{-1}$, which plugged in \eqref{eq-11b} completes the proof. 
\end{proof}

\subsection{Transport equation}\label{sec-4.4} We focus now on solving the transport equation that arises from the strategy outlined in \S\ref{sec-1.3}:
\begin{equation}\label{eq-4y}
\lr{L_1 \alpha u, u} = 0, \qquad L_1 =  D_t  + \left(\lambda(x) D_x + \dfrac{\lambda'(x)}{2i}\right) \sigma_1 + 2n \mu(x) \p_x\varphi  \big( \sigma_1 + \sigma_3 \p_x\varphi \big) + s(x).
\end{equation}
In \eqref{eq-4y}, $u$ is a real-valued normalized eigenvector $u$ of the matrix $\sigma_1 + \sigma_3 \p_x\varphi$, for the eigenvalue $\lr{\p_x\varphi}$; specifically:
\begin{equation}\label{eq-7k}
u = \matrice{\cos(\te/2) \\ \sin(\te/2)}, \qquad \te = \dfrac{\pi}{2} - \arctan \p_x \varphi.
\end{equation}

\begin{lemma}\label{lem:1f} The solution to \eqref{eq-4y} with initial data $\alpha(0,x,\xi) = 1$ is 
\begin{equation}\label{eq-7g}
\alpha = \big| \p_x H \big|^{1/2} \cdot \exp\left( -2n i
 \lr{\xi}^2\cdot \lambda^2 \circ H \cdot \int_{H}^x \dfrac{\mu}{\lambda^3} - \int_0^t s \circ F(\tau,H,\xi) d\tau \right),
\end{equation} 
with all functions in \eqref{eq-7g} evaluated at $(t,x,\xi)$. In particular, $\alpha$ is an admissible function of order $0$ (see Definition \ref{def:2}).
\end{lemma}

\begin{proof} 1. Since $|u|^2 = 1$, we have $\lr{\p_t u,u} = 0$. Moreover, we observe that $\p_x \lr{\sigma_3u,u} = 2 \lr{\sigma_3 \p_x u,u}$. Also, since $u$ is normalized eigenvector of $\sigma_1 + \sigma_3 \p_x\varphi$ for the eigenvalue $\lr{\p_x\varphi}$, 
\begin{equations}\label{eq-0r}
\lr{\big( \sigma_1 + \sigma_3 \p_x\varphi \big) u,u} = \lr{\p_x\varphi}, \\
\lr{\sigma_3 u, u} = \dfrac{\lr{\sigma_3 (\sigma_1 + \sigma_3 \p_x\varphi)u,u}}{\lr{\p_x\varphi}} = \dfrac{\lr{\sigma_3 (\sigma_1 + \sigma_3 \p_x\varphi)u,u}}{\lr{\p_x\varphi}} = \dfrac{\p_x\varphi}{\lr{\p_x \varphi}},
\end{equations}
where in the second identity we used that $\sigma_3^2 = \Id$ and $\sigma_3 \sigma_1$ is antisymmetric.



2. Thanks to the relations \eqref{eq-0r}, we deduce that, with all functions evaluated at $(t,x,\xi)$:
\begin{equations}\label{eq-6u}
i\lr{L_1 \alpha u, u} = \left( \dd{}{t} + \nu \dd{}{x} 
+ is + \dfrac{1}{2} \dd{\nu}{x} +  2ni \mu \p_x\varphi \lr{\p_x \varphi}\right) \alpha,  
\qquad \nu \de \dfrac{\lambda\p_x\varphi}{\lr{\p_x \varphi}}.
\end{equations}

3. We now solve the equation \eqref{eq-6u} along the bicharacteristic of $H_p$ (i.e. the projection along the $x$-axis of the Hamiltonian trajectories of $p$). We remark that
\begin{equation}
\nu(t,x_t,\xi) = \dfrac{\lambda(x_t)\p_x\varphi(t,x_t,\xi)}{\lr{\p_x \varphi(t,x_t,\xi)}} = \dfrac{\lambda(x_t) \xi_t}{\lr{\xi_t}} = \dot{x_t}. 
\end{equation}
Thanks to \eqref{eq-6u}, we deduce that $\alpha(t,x_t,\xi)$ satisfies the linear, homogeneous first order ODE
\begin{align}
\dfrac{d}{dt} \big(\alpha(t,x_t,\xi)\big) & =  \dd{\alpha(t,x_t,\xi)}{t} + \dot{x_t} \dd{\alpha(t,x_t,\xi)}{x}  =  \dd{\alpha(t,x_t,\xi)}{t} + \nu(t,x_t,\xi) \dd{\alpha(t,x_t,\xi)}{x}
\\
\label{eq-7d} &  = -\left(is(x_t) + \dfrac{1}{2} \dd{\nu}{x}(t,x_t,\xi)
+ 2ni \mu(x_t) \xi_t \lr{\xi_t}\right) \alpha(t,x_t,\xi). 
\end{align}

To solve this equation, we first compute an antiderivative of $\p_x \nu(t,x_t,\xi)$. Thanks to $\nu(t,x_t,\xi_t) = \dot{x_t} = \p_t F(t,x,\xi)$ and $H(t,x_t,\xi) = x$, we have, with all expressions evaluated at $(t,x,\xi)$:
\begin{equation}
\dd{\nu}{x}  = \left(\dd{H}{x} \cdot \dd{^2 F}{x \p t}\right)\big(t,H,\xi\big) = \left(\dfrac{1}{\p_x F} \cdot \dd{\p_x F}{t}\right)\big(t,H,\xi\big) = \dd{ \ln|\p_x F|}{t}(t,H,\xi),
\end{equation}
where we used that $\p_x H \cdot \p_x F(t,H,\xi) = 1$. It follows that
\begin{equation}\label{eq-6v}
\dd{\nu}{x}(t,x_t,\xi) = \dd{\ln|\p_x F|}{t}(t,x,\xi).
\end{equation}
We now focus on the term  $\mu(x_t) \xi_t \lr{\xi_t}$ from \eqref{eq-7d}. Thanks to $p(x_t,\xi_t) = p(x,\xi) = \lambda(x) \lr{\xi}$:
\begin{equation}\label{eq-6w}
\mu(x_t) \xi_t \lr{\xi_t} = \dfrac{\dot{x_t} \mu(x_t) \lr{\xi_t}^2}{\lambda(x_t)} = \dfrac{\dot{x_t} \mu(x_t)}{\lambda(x_t)^3}  p(x,\xi)^2 = p(x,\xi)^2 \dd{}{t} \left( \int_x^{x_t} \dfrac{\mu}{\lambda^3} \right).
\end{equation}

4. We now integrate the ODE \eqref{eq-7d}, using the formulas \eqref{eq-6v} and \eqref{eq-6w}:
\begin{equation}
\alpha(t,x_t,\xi) = \exp\left( -\dfrac{1}{2}\ln \big| \p_x F(t,x,\xi) \big| -2ni p(x,\xi)^2 \cdot \int_x^{x_t} \dfrac{\mu}{\lambda^3} - i\int_0^t s(x_\tau) d\tau\right).
\end{equation}
Using that $x = H(t,x_t,\xi)$, and using $\p_x H \cdot \p_x F (t,H,\xi) = 1$, we conclude that
\begin{equation}\label{eq-7h}
\alpha = \big| \p_x H \big|^{1/2} \cdot \exp\left( -2n i
 \lr{\xi}^2\cdot \lambda^2 \circ H \cdot \int_{H}^x \dfrac{\mu}{\lambda^3} - i\int_0^t s \circ F(\tau,H,\xi) d\tau \right),
\end{equation}
with all functions evaluated at $(t,x,\xi)$. This completes the proof of \eqref{eq-7g}.

5. It remains to show that $\alpha$ is an admissible function of order $0$. The term involving $s$ does not cause problem because derivatives of $F$ and $H$ are uniformly bounded.  By construction of $T$, the function $\p_x H$ is in $C^\infty_b(X)$ and is bounded below by a positive constant -- see Lemma \ref{lem:1d}. Therefore $\big| \p_x H(t,x,\xi) \big|^{1/2} \in C^\infty_b(X)$, hence it is admissible. 

We now look at the factor
\begin{equation}\label{eq-7j}
\exp\left( -2n i
 \lr{\xi}^2\cdot \gamma(t,x,\xi) \right), \qquad \gamma(t,x,\xi) \de \lambda^2 \circ H(t,x,\xi) \cdot \int_{H(t,x,\xi)}^x \dfrac{\mu}{\lambda^3}
\end{equation}
arising in \eqref{eq-7h}. We note that $H(t,x,\xi) - x$ is bounded together with all its derivatives on $X$; and so are $\lambda^2 \circ H$ and $\mu \lambda^{-3}$. Therefore $\gamma \in C^\infty_b(X)$. Taking $k$ derivatives of \eqref{eq-7j} produces then growth at most $\lr{n,\xi^2}^{k}$. Therefore \eqref{eq-7j}, hence $\alpha$, are admissible of order $0$.
\end{proof}

We then set $b_0 = \alpha u$, where $\alpha$ is given by \eqref{eq-7g}, so that $L_1 b_0$ is orthogonal to the kernel of $L_0$. The equation $L_0 b_1 + L_1 b_0 = 0$ has then a unique solution $b_1 \in (\ker L_0)^\perp$. We then set 
\begin{equation}\label{eq-5a}
b = b_0+\epsilon b_1
\end{equation}

\begin{lemma} The function $b$ is admissible of order $1$.
\end{lemma}

\begin{proof} Since $b = b_0 + \epsilon b_1 = \alpha u+\epsilon b_1$, and $\alpha$ is admissible of order $0$ (see Lemma \ref{lem:1f}), it suffices to show that $u$ and $b_1$ are both admissible of order (at most) $1$. Since $L_0$ is a $2 \times 2$ matrix with eigenvalues $0$ and $2\p_t \varphi = - 2\lambda \lr{\p_x\varphi}$, we have
\begin{equation}\label{eq-7l}
b_1 = \dfrac{L_1(\alpha u)}{2 \lambda \lr{\p_x\varphi}}.
\end{equation}
The estimate \eqref{eq-6r} implies that $\p_x \varphi$ is admissible of order $1/2$ (much better, in fact); hence from \eqref{eq-7k}, $u$ is admissible of order $0$. The coefficients in the non-differential parts of $L_1$ are admissible of order $3/2$, while the differential coefficients are just $C^\infty_b$, see \eqref{eq-4y}. Hence $L_1(\alpha u)$ is admissible of order $3/2$. Moreover, since $\p_x\varphi(0,x,\xi) = \xi$, we have from \eqref{eq-6r} that
\begin{equation}
\p_x\varphi(t,x,\xi) = \xi + \int_0^t \p_{tx}^2 \varphi(\tau,x,\xi) d\tau \geq \left( 1 - Ct\right) \xi.
\end{equation} 
From the expression \eqref{eq-7l} of $b_1$, we deduce that $b_1$ is admissible of order $1$. Therefore $b$ is admissible of order $1$.
 \end{proof}

\begin{lemma}\label{lem:1o} With $b$ and $\varphi$ defined as above, we have
\begin{equation}\label{eq-7n}
\big(\epsilon D_t + \Dii_{n,\epsilon}\big) b e^{i\varphi / \epsilon} = \epsilon^2 r e^{i\varphi / \epsilon},
\end{equation}
for an admissible function $r$ of order $4$. 
\end{lemma}

\begin{proof} We recall that 
\begin{equation}
e^{-i\varphi / \epsilon} \big(\epsilon D_t + \Dii_{n,\epsilon}\big) e^{i\varphi / \epsilon} = L_0 + \epsilon L_1 + \epsilon^2 L_2
\end{equation}
where $L_0, L_1, L_2$ are given in \eqref{eq-7o}. Using that $b = b_0+\epsilon b_1$, $L_0 b_0 = 0$ and $L_1 b_0 + L_1 b_0 = 0$, we deduce that 
\begin{equation}
\big(\epsilon D_t + \Dii_{n,\epsilon}\big) b e^{i\varphi / \epsilon} =  \epsilon^2 \big(L_2 b + L_1 b_1 \big) e^{i\varphi / \epsilon},
\end{equation}
and therefore the function $r$ emerging in \eqref{eq-7o} is given by $r = L_2 b + L_1 b_1$. We use the formulas \eqref{eq-7g} and the fact that $b, b_1$ are admissible of order $0$ to check that $r$ is admissible and compute its order. The worst term is $2n \p_x \varphi \mu(x) D_x b$, which is admissible of order $3/2+1+1=7/2 \leq 4$. This completes the proof. 
\end{proof}

\subsection{Conter-propagating WKB wave}\label{sec-4.5} We briefly explain how to construct a WKB wave starting from the other branch of \eqref{eq-2s}. One can check that the WKB solution constructed in \S\ref{sec-4.3} is also defined for short negative time, and satisfies the same estimates; in fact, one has
\begin{equation}
\vp(t,x,\xi) = -\vp(-t,x,-\xi).
\end{equation}
Consider then, for $t < 0$: 
\begin{equation}\label{eq-0s}
\TT (be^{i\varphi/\epsilon})(t,x,\xi) = i\sigma_2 \ove{b(t,x,\xi)} e^{-i\varphi(t,x,\xi)}.
\end{equation}
Because $\Dii_{n,\epsilon}$ anticommutes with $\TT$, see \S\ref{sec-2.3}, \eqref{eq-0s} is another WKB solution for $(\epsilon D_t + \Dii_{n,\epsilon})$. It has orientation orthogonal to $be^{i\varphi}$. We refer to it as the counter-propagating WKB wave to $be^{i\varphi}$.

\section{Parametrix}\label{sec-5}

In this section we complete the proof of Theorem \ref{thm:1i}. Our proof strategy goes as follows:
\begin{itemize}
\item[\textbf{1.}] We produce a parametrix $\fE$ for the operator $\epsilon D_t + \Dii_{n,\epsilon}$. It takes the form of an integral superposition of the WKB waves constructed in \S\ref{sec-4}. 
\item[\textbf{2.}] We prove an asymptotic expansion, $L^\infty$, and $L^2$ estimates for $\frak{E}_{n,\epsilon}$ as $\epsilon \rightarrow 0$.
\item[\textbf{3.}] We sum the parametrices, asymptotic expansions and estimates for each component in the block-diagonal decomposition of $\Dii$ (see \S\ref{sec-5.3}).
\end{itemize} 
The result is a parametrix for $h D_t + \Dii$, with associated $L^2 / L^\infty$ estimates that imply Theorem \ref{thm:1i}. 

\subsection{Preliminary estimates}\label{sec-5.1} We start by studying integral quantities of the form
\begin{equation}\label{eq-8f}
Q(t,x) = \int_\R e^{i\varphi(t,x,\xi) / \epsilon} \cdot q(t,x,\xi) \dfrac{d\xi}{\sqrt{\epsilon}},
\end{equation}
where $q$ is an admissible function of order $-\infty$, see Definition \ref{def:2}. Quantities of the form \eqref{eq-8f} are typically superposition of WKB waves constructed in \S\ref{sec-4}. We prove here that \eqref{eq-8f} expands in terms of the critical point $\Xi(t,x)$ of $\p_\xi \varphi(t,x,\cdot)$.

Before stating our next result, we comment that while technically speaking, the quantity $q\big(t,x,\Xi(t,x)\big)$ is only defined for $x \in I_t$, it extends smoothly for $x \in \R$ by:
\begin{equation}\label{eq-8o}
q \big(t,x,\Xi(t,x)\big) = 0, \qquad x \notin I_t.
\end{equation}
This is a consequence of the rapid decay of $\xi \mapsto q(t,x,\xi)$ ($q$ is an admissible function of order $-\infty$) and of the bound \eqref{eq-8l} for $\Xi(t,x)$ as $x$ approaches $x_t^\pm$.

\begin{theorem}\label{thm:1z} If $q$ is admissible of order $-\infty$, then:
\begin{itemize}
\item[(i)] Uniformly for $t$ in compact subsets of $(0,T)$, $\epsilon \in (0,1)$ and $n \in \N$:
\begin{equation}\label{eq-8v}
Q(t,x) = \left( e^{i \varphi / \epsilon}  \sqrt{\dfrac{2\pi i}{ \p_\xi^2 \varphi }} \cdot q \right)\big( t,x, \Xi(t,x) \big) + O_{L^2(\R)}\big(\epsi^{1/2} n^{-\infty}\big).
\end{equation}

\item[(ii)] Uniformly for $t \in (0,T)$, $\epsilon \in (0,1)$ and $n \geq 1$:
\begin{equation}\label{eq-8i}
Q(t,x) = O_{L^\infty(\R)} \big( t^{-1/2}n^{-\infty}  \big),
\qquad Q(t,x)   = O_{L^2(\R)} \big( t^{-1/2}n^{-\infty}  \big).
\end{equation}
\end{itemize}
\end{theorem}

\begin{proof}[Proof of \eqref{eq-8v}] 1. In this proof we will occasionally drop the variables $t, x$ to keep equations under reasonable length. Assume that $(t,x)$ are such that $\p_\xi \varphi (t,x,\cdot)$ does not vanish on the support of $q(t,x,\cdot)$. We integrate \eqref{eq-8f} by parts:
\begin{equation}
Q(t,x) = i \epsilon^{1/2} \int_\R e^{i\varphi / \epsilon} \dd{}{\xi}\dfrac{q}{\p_\xi \varphi}  d\xi.
\end{equation}
Using that $\p_\xi^2 \varphi$ is uniformly bounded above, we deduce there exists $C > 0$ such that for every $(t,x)$ with $\p_\xi \varphi (t,x,\cdot)$ not vanishing on the support of $q(t,x,\cdot)$: 
\begin{equation}\label{eq-8g}
\big| Q(t,x) \big| \leq C \epsilon^{1/2} \sum_{j=0}^1 \int_\R |\p_\xi \varphi|^{j-2} |\p_\xi^j q|  d\xi.
\end{equation}

2. If $x \notin I_t$ then according to Lemma \ref{lem:1g}:
\begin{equation}\label{eq-8h}
 \big|\p_\xi \varphi(\xi)\big| \geq c  \big( t\lr{\xi}^{-2} +  \dist(x,I_t) \big).
\end{equation}
We then apply \eqref{eq-8g}:
\begin{equation}\label{eq-8w}
\big| Q(t,x) \big| \leq C \epsilon^{1/2} \cdot \sum_{j=0}^1 \big( t +  \dist(x,I_t) \big)^{j-2} \int_\R \lr{\xi}^{4-2j}  |\p_\xi^j q| d\xi, \qquad x \notin I_t.
\end{equation} 
By assumption $q$ and its derivatives decay faster than any polynomial in $\xi$, uniformly for $(t,x) \in \Omega_0$. Therefore, we deduce that
\begin{equation}\label{eq-8j}
\int_{x \notin I_t} \big| Q(t,x) \big|^2 dx \leq  C \epsilon \sum_{j=0}^1 \int_{x \notin I_t} \big( t +  \dist(x,I_t) \big)^{2j-4} dx \leq C t^{-3} \epsilon.
\end{equation}

3. We now work with $x \in I_t$. Set $C_0 = ct (2+2|\p_\xi^3 \varphi|_\infty)^{-1}$, where $c$ is the constant involved in \eqref{eq-7p}, and assume that
\begin{equation}
\supp \ q(t,x,\cdot) \subset \big\{ \xi : \ |\xi - \Xi(t,x)| \geq C_0	 \lr{\Xi(t,x)}^{-3} \big\}. 
\end{equation}
Then on the support of $q(t,x,\cdot)$, we have (dropping the dependence in $t,x$):
\begin{equations}
\big| \p_\xi \varphi(\xi) \big| = \big| \p_\xi \varphi(\xi) - \p_\xi \varphi(\Xi) \big| = \left| \int_{\Xi}^\xi \p_\xi^2 \varphi(\theta) d\theta \right| \geq ct \left| \int_{\Xi}^\xi \lr{\theta}^{-3} d\theta \right|.
\end{equations}

Assume first that $(t,x)$ are such that $C_0 \lr{\Xi}^{-3} \geq 2$. Then the interval between $\xi$ and $\Xi$ contains either $[\xi-2,\xi]$ or $[\xi,\xi+2]$. In particular, using that $\lr{\xi-2}, \lr{\xi}$ and $\lr{\xi+2}$ are comparable, we deduce that 
\begin{equations}
\big| \p_\xi \varphi(\xi) \big| \geq ct \lr{\xi}^{-3}, \qquad (t,x,\xi) \in \supp(q).
\end{equations}
Applying \eqref{eq-8g} produces 
\begin{equation}
\big| Q(t,x) \big| \leq C \epsilon^{1/2} \sum_{j=0}^1 \int_\R   \lr{\xi}^{6-3j} |\p_\xi^j q|  d\xi.
\end{equation}

Assume now that $(t,x)$ are such that $C_0 \lr{\Xi}^{-3} \leq 2$. Then if $\xi$ is such that $|\xi - \Xi(t,x)| \geq C_0 \lr{\Xi(t,x)}^{-3}$, then $\lr{\xi}$ and $\lr{\Xi}$ are comparable. It follows from \eqref{eq-8i} that  
\begin{equation}
\big| \p_\xi \varphi(\xi) \big| \geq ct |\xi-\Xi| \lr{\xi}^{-3} \geq ct \lr{\Xi}^{-3} \lr{\xi}^{-3} \geq ct \lr{\xi}^{-6}.
\end{equation}
Again, applying \eqref{eq-8g} produces
\begin{equation}
\big| Q(t,x) \big| \leq C \epsilon^{1/2} \sum_{j=0}^1 \int_\R \lr{\xi}^{12-6j}  |\p_\xi^j q| d\xi.
\end{equation}

4. We now assume that $x \in I_t$, and that  $(t,x)$ are such that
\begin{equation}
\supp \ q(t,x,\cdot) \subset \{ \xi : \ |\xi - \Xi(t,x)| \leq 2C_0 \lr{\Xi(t,x)}^{-3}\}. 
\end{equation}
Following H\"ormander's approach to stationary phase estimates -- see \cite[\S7.7]{H90} -- we write (again dropping the dependence in $t,x$ to shorten expressions):
\begin{equations}
\varphi(\xi) = \varphi(\Xi) + \dfrac{\p_\xi^2\varphi(\Xi)}{2}(\xi-\Xi)^2 +  r(\xi)(\xi-\Xi)^3, \qquad r(\xi) \de \int_0^1 \dfrac{(1-\te)^2}{2} \p_\xi^3 \varphi \big(\te\xi+(1-\te) \Xi\big) d\theta.
\\
\varphi_\te(\xi) \de \varphi(\Xi) + \dfrac{\p_\xi^2\varphi(\Xi)}{2}(\xi-\Xi)^2 + \te r(\xi)(\xi-\Xi)^3,
\qquad
I(\te) \de \int_\R e^{i\varphi_\te / \epsilon} q \dfrac{d\xi}{\epsilon^{1/2}}.
\end{equations} 
We note that $\varphi_1 = \varphi$, in particular
\begin{equation}
Q(t,x) = I(1) = I(0) + \int_0^1 I'(\te) d\te, \qquad I'(\te) = \dfrac{1}{h^{3/2}} \int_\R (\xi-\Xi)^3 e^{i\varphi_\te / \epsilon} r q d\xi.
\end{equation}

We show that $I'(\te) = O(\epsilon^{1/2})$. For that, we note that 
\begin{equation}
\p_\xi \varphi_\te(\xi) = (\xi-\Xi) \tr_\te(\xi), \qquad \tr_\te(\xi) \de \p_\xi^2\varphi(\Xi) + \te(\xi-\Xi) \int_0^1 \p_\xi^3 \varphi\big(\te'\xi+(1-\te') \Xi\big) d\theta'. 
\end{equation}
In particular, 
\begin{equation}\label{eq-8s}
\big| \tr_\te(\xi) \big| \geq \big| \p_\xi^2\varphi(\Xi) \big| - \dfrac{1}{2}|\p_\xi^3 \varphi|_\infty |\xi-\Xi| \geq \dfrac{1}{2} ct \lr{\Xi}^{-3},
\end{equation}
where we used $|\xi-\Xi| \leq 2C_0 \lr{\Xi}^{-3}$ and the above value of $C_0$. We now integrate by parts twice the expression for $I'(s)$ using the expression
\begin{equation}
(\xi-\Xi) e^{i\varphi_\te / \epsilon} = \dfrac{\p_\xi \varphi_\te}{\tr_\te} e^{i\varphi_\te / \epsilon} = \dfrac{\epsilon}{i \tr_\te} \dd{e^{i\varphi_\te / \epsilon}}{\xi} . 
\end{equation}
This gives:
\begin{equations}
I'(\te) = \dfrac{1}{h^{1/2}} \int_\R  \dfrac{1}{i \tr_\te} \dd{e^{i\varphi_\te / \epsilon}}{\xi}  \cdot (\xi-\Xi)^2 r q d\xi 
= \dfrac{i}{h^{1/2}} \int_\R  e^{i\varphi_\te / \epsilon} \cdot \dd{}{\xi} \dfrac{(\xi-\Xi)^2 r q}{\tr_\te} d\xi
\\
= \dfrac{i}{h^{1/2}} \int_\R (\xi-\Xi) e^{i\varphi_\te / \epsilon} \cdot  \left( 2 + (\xi-\Xi) \dd{}{\xi}\right) \dfrac{r q}{\tr_\te} d\xi
= -\epsilon^{1/2} \int_\R e^{i\varphi_\te / \epsilon} \cdot  \dd{}{\xi} \dfrac{1}{\tr_\te}\left( 2 + (\xi-\Xi) \dd{}{\xi}\right) \dfrac{r q}{\tr_\te} d\xi.
\end{equations}
Thanks to the bound \eqref{eq-8s}, and that $\xi$ and $\Xi$ are comparable, we conclude that 
\begin{equation}
\big| I'(\te) \big| \leq C\epsilon^{1/2} \cdot \sum_{j=0}^2 \int_\R \lr{\xi}^{12-3j} \big|\p_\xi^j q \big| d\xi.
\end{equation}

It remains to evaluate $I(0)$. We follow the standard calculation for oscillatory integrals with quadratic phase. From Plancherel's formula:
\begin{equations}
I(0) = \int_\R e^{i\varphi_0 / \epsilon} q \dfrac{d\xi}{\epsilon^{1/2}} = e^{i\varphi(\Xi) / \epsilon} \sqrt{\dfrac{2\pi i}{\p_\xi^2 \varphi(\Xi)}} \cdot  \int_\R  e^{ - i \epsilon \frac{y^2}{2 \p_\xi^2 \varphi(\Xi)}} e^{i \Xi y}\widecheck{q}(y) dy
\\
= e^{i\varphi(\Xi) / \epsilon} \sqrt{\dfrac{2\pi i}{\p_\xi^2 \varphi(\Xi)}} \cdot \left( q(\Xi) + \int_\R \Big( e^{ - i \epsilon \frac{y^2}{2 \p_\xi^2 \varphi(\Xi)}} - 1\Big) e^{i \Xi y}\widecheck{q}(y) dy \right).
\end{equations}
We control the integral emerging in the right-hand side by:
\begin{equations}
\epsilon^{1/2} \int_\R \frac{\epsilon|y|}{\big| 2 \p_\xi^2 \varphi(\Xi)\big|^{1/2}} \left| \widecheck{q}(y) \right| dy \leq  \int_\R \frac{|y|}{\big| 2 \p_\xi^2 \varphi(\Xi)\big|^{1/2} (1+y^2)} \left| \big(q-\p_\xi^2 q\big)^{\widecheck{\ }}(y) \right| dy
\\
 \leq C \epsilon^{1/2} \left( \int_\R \lr{\xi}^3 \big| \big(1-\p_\xi^2\big) q \big|^2 \right)^{1/2},
\end{equations}
where we use the Cauchy--Schwartz inequality, Plancherel's formula and that $\lr{\xi}, \lr{\Xi}$ are comparable. 

5. To conclude, we first observe that when $x \notin I_t$, \eqref{eq-8j} provides $L^2$-decay. On the other hand, when $x \in I_t$, $\big| Q(t,x) \big|$ is uniformly bounded. We then split the integral by mean of a cutoff function, 
$\chi\big( |\xi-\Xi| \lr{\Xi}^{3} \big)$, where $\chi$ is even, smooth and has compact support, and $\chi(0)=1$. Taking at most two derivatives of $\chi$ produces a factor $\lr{\Xi}^6$, which is comparable to $\lr{\xi}^6$ on the support of $\chi'$. Since $q$ is of order $-\infty$, $L^1$-norms of $\lr{\xi}^k \p_\xi^j q$ decay faster than any power of $n$. This completes the proof of \eqref{eq-8v}.  \end{proof}

\begin{proof}[Proof of the $L^\infty$-bound in \eqref{eq-8i}] To limit the length of expressions, we occasionally drop the variables $(t,x)$ in the expressions below. We first assume that $q(t,x,\cdot)$ has support in for $\xi \in [j,j+2]$ for some $j \in \Z$. According to Lemma \ref{lem:1k}, $\varphi$ satisfies the non-stationnary estimate
\begin{equation}
\left| \dd{^2}{\xi^2} \dfrac{\varphi(\xi)}{ct \lr{j}^{-3}} \right| \geq 1, \qquad \xi \in [j,j+2]
\end{equation}
where $c$ is independent of $j$ and $t$. We can therefore apply \cite[Corollary page 334]{S93}. With the notations of this book, the large parameter is $\epsilon^{-1} ct \lr{j}^{-3}$ and the constant involved is $c_2 = 5 \cdot 2^{2-1}-2 = 8$, see the formula on \cite[page 333]{S93}. This implies that
\begin{equations}
\left| \dfrac{1}{\epsilon^{1/2}}\int_j^{j+2} e^{i\varphi(\xi)/h} q(\xi) d\xi \right| \leq \dfrac{8}{\epsilon^{1/2}} \left( \dfrac{c t \lr{j}^{-3}}{\epsilon} \right)^{-1/2} \int_j^{j+2} \big| \p_\xi q (\xi)  \big| d\xi 
\leq  \dfrac{C}{\sqrt{t}} \int_j^{j+2} \lr{\xi}^{3/2} \big| \p_\xi q  \big| d\xi,
\end{equations}
where we used that $\lr{j}^{3/2}$ and $\lr{\xi}^{3/2}$ are comparable for $\xi \in [j,j+2]$. The constant $C$ on the RHS is independent of $j, t, \epsilon$.

If now $q$ has full support, we decompose it into pieces supported in $\xi \in [j,j+2]$, $j \in \Z$, and we sum the resulting estimates over $j$. It produces the bound 
\begin{equations}
\left| \dfrac{1}{\epsilon^{1/2}}\int_\R e^{i\varphi(\xi)/h} q(\xi) d\xi \right| \leq  \dfrac{C}{\sqrt{t}} \int_\R \lr{\xi}^{3/2} \big| \p_\xi q  \big| d\xi.
\end{equations}
Since $q$ is admissible of order $-\infty$, the integral on the RHS decays faster than any power of $n$. This completes the proof. 
\end{proof}

\begin{proof}[Proof of the $L^2$-bound in \eqref{eq-8i}.] The bound \eqref{eq-8w} implies that
\begin{equation}
\big| Q(t,x) \big| \leq C \epsilon^{1/2} \cdot \sum_{j=0}^1 \dist(x,I_t)^{j-2} \int_\R \lr{\xi}^{4-2j}  |\p_\xi^j q| d\xi.
\end{equation}
We square then integrate this expression for $\dist(x,I_t) \geq 1$. It gives
\begin{equation}\label{eq-8y}
\int_{\dist(x,I_t) \geq 1} \big| Q(t,x) \big|^2 dx \leq C \epsilon  \left(  \sum_{j=0}^1 \int_\R \lr{\xi}^{4-2j}  |\p_\xi^j q| d\xi \right)^2.
\end{equation}
When $\dist(x,I_t) \leq 1$, we use the $L^\infty$ bound and get -- using that $|I_t| = O(t)$:
\begin{equation}\label{eq-8x}
\int_{\dist(x,I_t) \leq 1} \big| Q(t,x) \big|^2 dx \leq Ct^{-1} (t+ 1) \left(\int_\R \lr{\xi}^{3/2} \big| \p_\xi q  \big| d\xi\right)^2 \leq \dfrac{C}{t} \left(\int_\R \lr{\xi}^{3/2} \big| \p_\xi q  \big| d\xi\right)^2. 
\end{equation}
Summing \eqref{eq-8y} and \eqref{eq-8x} yields
\begin{equation}
\big|Q(t,\cdot)\big|^2_{L^2} \leq \dfrac{C}{t} \left(  \sum_{j=0}^1 \int_\R \lr{\xi}^{4-2j}  |\p_\xi^j q| d\xi \right)^2.
\end{equation} 
Since $q$ is admissible of order $-\infty$, the integral on the RHS decays faster than any power of $n$. This completes the proof of \eqref{eq-8i}.
\end{proof}

\subsection{Parametrix} We now construct a parametrix for $\epsilon D_t + \Dii_{n,\epsilon}$ as an integral operator of the form \eqref{eq-8f}. Let $a_n$ depending and $x$ and $n$ (but not on $\epsilon$) such that
\begin{equation}\label{eq-9d}
(t,x,\xi) \mapsto \ha_n(\xi) \de \int_\R e^{-ix\xi} a_n(x) dx
\end{equation}
is admissible of order $-\infty$. Let $b_n^+ = b$ given by \eqref{eq-5a}, $b_n^- =\ TT b =  i\sigma_2 \cdot \ove{b}$ (note that $b$ depends implicitly on $n$ and $\epsilon$). Define then 
\begin{equations}\label{eq-9a}
\fE \de \fE^+ + \fE^-, \qquad \text{where} \\
\fE^\pm a_n(t,x) \de \dfrac{1}{2\pi\epsilon^{1/2}} \int_\R e^{\pm i\varphi( t,x, \pm \xi) / \epsilon} \cdot b^\pm_{n,\epsilon}(t,x,\xi) b^\pm_{n,\epsilon}(0,0,\xi)^\top \ha_n(\xi) d\xi.
\end{equations}
Note that $\fE^-$ is asssociated to the counter-propagating WKB wave constructed in \S\ref{sec-4.5}. Moreover the integral \eqref{eq-9a} is of the form \eqref{eq-8f}, in particular if $a$ is an admissible function of order $-\infty$ then uniformly in $t \in (0,T)$, $n \geq 1$ and $\epsilon \in (0,1)$:
\begin{equation}\label{eq-13f}
 \fE a(t,\cdot) = O_{L^2(\R)}(\epsilon n^{-\infty} t^{-1/2}).
\end{equation}

Our next theorem shows that $\fE$ is a parametrix for $\epsilon D_t + \Di_{n,\epsilon}$: $\fE a_n$ approximately solves the equation
\begin{equation}\label{eq-7f}
\left( \epsilon D_t + \Dii_{n,\epsilon} \right) f = 0, \qquad f(0,x_1) = \dfrac{1}{\epsilon^{1/2}} \cdot a_n\left(  \dfrac{x}{\epsilon}\right).
\end{equation}

\begin{theorem}\label{thm:1h} The solution $f$ to \eqref{eq-7f} satisfies,  uniformly in $t \in (0,T)$, $n \geq 1$ and $\epsilon \in (0,1)$:
\begin{equations}\label{eq-3a}
f(t,x) = \fE a(t,x) + O_{L^2(\R)}\big( \epsilon n^{-\infty} \big).
\end{equations}
Moreover, uniformly for $t$ in compact subsets of $(0,T)$:
\begin{equations}\label{eq-3aa}
f(t,x) =  \sum_{\pm} \left( e^{\pm i \varphi / \epsilon}  \sqrt{\pm \dfrac{i}{ 2\pi\p_\xi^2 \varphi }} \cdot b^\pm_{n,\epsilon}  b^\pm_{n,\epsilon}(0,0,\cdot)^\top \ha_n\right)\big( t,x, \pm \Xi(t,x) \big) + O_{L^2(\R)}\big(\epsilon^{1/2}n^{-\infty}\big).
\end{equations}
\end{theorem}


The formula \eqref{eq-3aa} implies that up to a small remainder, $f$ is also essentially supported within the cone $\{(t,x) : x \in I_t\}$ and uniformly bounded as $\epsilon \rightarrow 0$.  This is a particularly strong form of dispersion: the $L^\infty$-norm passes from $\epsilon^{-1/2}$ at time $0$ to $O(1)$ as soon as $t > 0$. The initial wavepacket structure is immediately lost.

\begin{proof} 1. We first verify that $f(0,x) = \fE a(0,x)$. Using $\varphi(0,x,\xi) = x\xi$ and that $b_\pm(0,x,\xi)$ does not depend on $x$, see \eqref{eq-7p}, we have
\begin{equation}\label{eq-6y}
\fE a(0,x) = \dfrac{1}{2\pi\epsilon^{1/2}} \int_\R e^{ \frac{ix\xi}{\epsilon} } \left( \sum_\pm b_\pm(0,0,\xi) \cdot b_\pm(0,0,\xi)^\top \right) \ha(\xi) d\xi\end{equation}
Moreover, $b_\pm(0,0,\xi)$ are orthonormal eigenvectors of $\sigma_1+\xi \sigma_3$. It follows that the matrix in \eqref{eq-6y} is the $2 \times 2$ identity, therefore \begin{equation}
\fE a(0,x) = \dfrac{1}{2\pi\epsilon^{1/2}} \int_\R e^{ \frac{ix\xi}{\epsilon} } \ha(\xi) d\xi = \dfrac{1}{\epsilon^{1/2}} \cdot a\left(  \dfrac{x}{\epsilon}\right) = f(0,x).
\end{equation}

2. Thanks to Theorem \ref{thm:1g}, we have $\big( \epsilon D_t + \Dii_{n,\epsilon} \big) \fE  = \epsilon^2 R$, where
\begin{equation}
Ra(t,x) = \dfrac{1}{2\pi\epsilon^{1/2}} \sum_\pm \int_\R e^{\pm i\varphi(t,x,\pm \xi) / \epsilon} r_\pm(t,x,\xi) b_\pm(0,0,\xi) \ha(\xi) d\xi,
\end{equation}
for two admissible functions $r_\pm$ of order $4$. In particular, recalling that $\ha$ is admissible of order $-\infty$, $R a$ is an integral of the form \eqref{eq-8f}. From Theorem \ref{thm:1z}(ii), we deduce that uniformly for $t \in (0,T)$: 
\begin{equation}\label{eq-13n}
\left\| \big( \epsilon D_t + \Dii_{n,\epsilon} \big) \fE  a(t,\cdot) \right\|_{L^2(\R)} \leq C t^{-1/2} n^{-\infty} \epsilon^2.
\end{equation}
Moreover, the function $t \mapsto t^{-1/2}$ is integrable near $0$. Applying Duhamel's principle as in \cite[Lemma 3.5]{BB+21}, we conclude that
\begin{equation}
\big\| f(t,\cdot) - \fE a(t,\cdot) \big\|_{L^2} \leq \dfrac{C n^{-\infty}}{\epsilon} \int_0^t \epsilon^2 \tau^{-1/2} d\tau \leq Cn^{-\infty} t^{1/2} \epsilon.
\end{equation}
This proves \eqref{eq-3a}. To deduce \eqref{eq-3aa}, we observe that \eqref{eq-9a} precisely has the form of the integral studied in \S\ref{sec-5.1}; therefore it suffices to apply Part (i) in Theorem \ref{thm:1z}. This completes the proof. \end{proof}

\subsection{Summability.}\label{sec-5.3} We now produce a leading order expansion for the solution to: 
\begin{equation}\label{eq-9b}
(h D_t + \Dii) F = 0, \qquad F(0,x) = \dfrac{1}{\sqrt{h}} \cdot A\left( \dfrac{x}{\sqrt{h}} \right), \qquad A \in \SSS(\R^2).
\end{equation}


Given the block-decomposition of $\Dii$ exhibited in \S\ref{sec-2.2}, and Theorem \ref{thm:1k} about traveling waves, a natural guess for the solution $F$ to \eqref{eq-9b} is:
\begin{equations}\label{eq-12g}
\frak{E} A(t,x) = \sum_{n=0}^\infty (\frak{E}_n a_n)(t,x_1) G_{n,h}(x_2), \qquad \text{where:}
\\
(\frak{E}_0 a_0)(t,x_1) \de \dfrac{e^{iS_t}}{\sqrt{\rho_th}} \cdot \left(\dfrac{e^{-i\pi/4}}{|4\pi\nu_t|^{1/2}} \cdot e^{i \frac{x_1^2}{4 \nu_t}} \star a_0\right) \left( \dfrac{x_1-x_t^+}{ \rho_t\sqrt{h}} \right); \qquad \text{and}
\\
\frak{E}_n \de \frak{E}_{n,h_n}, \qquad h_n = \sqrt{\frac{h}{2n}}, \qquad  a_n(x_1) \de \dfrac{1}{(2n)^{1/4}} \int_\R A\left( \dfrac{x_1}{\sqrt{2n}}, x_2 \right) G_n(x_2) dx_2.
\end{equations}

\begin{theorem}\label{thm:1a} The solution $F$ to \eqref{eq-9b} satisfies, uniformly for $t \in (0,T)$:
\begin{equation}\label{eq-9e}
F(t,x) = \frak{E}A(t,x) + O_{L^2(\R^2)}\big(h^{1/2}\big).
\end{equation}
Moreover, uniformly for $t$ in compact sets of $(0,T)$:
\begin{align}\label{eq-9ee}
F(t,x)
&  = \dfrac{e^{iS_t}}{\sqrt{\rho_th}} \cdot \left(\dfrac{e^{-i\pi/4}}{|4\pi\nu_t|^{1/2}} \cdot e^{i \frac{x_1^2}{4 \nu_t}} \star a_0\right) \left( \dfrac{x_1-x_t^+}{ \rho_t\sqrt{h}} \right) \cdot e^{-\frac{x_2^2}{2h}} \matrice{1 \\ 0}
\\
& \quad + \sum_{\substack{n=1, \pm}}^\infty \left( e^{\pm i \varphi / h_n}  \sqrt{\dfrac{\pm i}{ 2\pi\p_\xi^2 \varphi }} \cdot b^\pm_{n,h_n}  b^\pm_{n,h_n}(0,0,\cdot)^\top \ha_n\right)\big( t,x, \pm \Xi(t,x) \big) \cdot G_{n,h}(x_2) + O_{L^2}\big(h^{1/4}\big).
\end{align}
\end{theorem}

\begin{center}
\begin{figure}[!t]
\floatbox[{\capbeside\thisfloatsetup{capbesideposition ={right,center}}}]{figure}
{\caption{Space-time diagram for propagation starting at $(0,0)$. The blue region represents all accessible positions. The red curve is the trajectory of a traveling mode. Away from this curve the propagation is dispersive. }\label{fig:5}}
{\begin{tikzpicture}
\node at (0,0) {\includegraphics[scale=1]{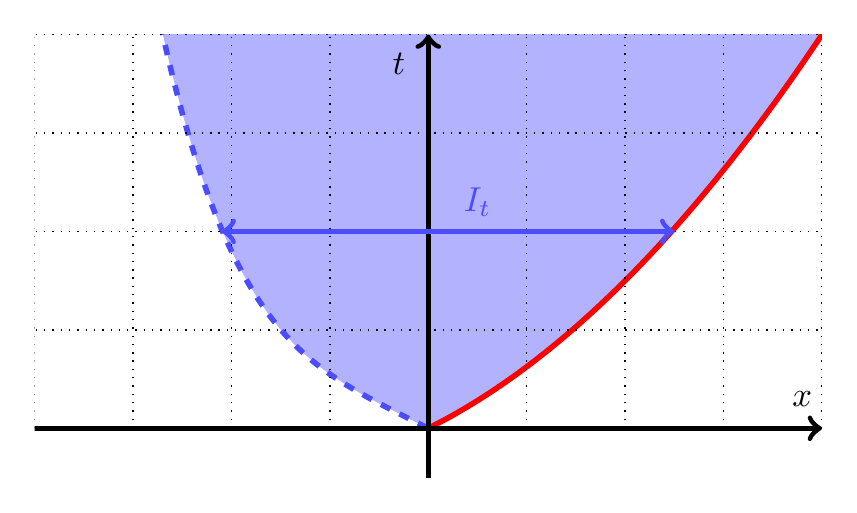}};
\end{tikzpicture}}
\end{figure}
\end{center}

\begin{proof}[Proof of Theorem \ref{thm:1a}] Thanks to Theorem \ref{thm:1k}, it suffices to prove the theorem when $a_0 = 0$. We decompose the solution $F$ to \eqref{eq-9b} as $F = \sum_{n=1}^\infty f_n \otimes G_{n,h}$, where each $f_n$ evolves according to $\left(h_n D_t + \Dii_{n,h_n} \right) f_n = 0$. In particular, we have
\begin{align}
f_n(0,x_1) & = \int_\R F(0,x_1,x_2) G_{n,h}(x_2) dx_2 = \int_\R \dfrac{1}{h^{3/4}}  A\left( \dfrac{x_1}{\sqrt{h}}, \dfrac{x_2}{\sqrt{h}} \right) G_n\left(\dfrac{x_2}{\sqrt{h}}\right) dx_2 
\\
& = \int_\R \dfrac{1}{h^{1/4}} A\left( \dfrac{x_1}{\sqrt{h}}, x_2 \right) G_n(x_2) dx_2 = \int_\R \dfrac{1}{(2n)^{1/4} \sqrt{h_n}} A\left( \dfrac{x_1}{\sqrt{2n} h_n}, x_2 \right) G_n(x_2) dx_2
\\
& = \dfrac{1}{\sqrt{h_n}} \cdot a_n \left( \dfrac{x_1}{h_n} \right),
\label{eq-8c}
\end{align}
where we defined
\begin{equation}\label{eq-9c}
a_n(x_1) \de \dfrac{1}{(2n)^{1/4}} \int_\R A\left( \dfrac{x_1}{\sqrt{2n}}, x_2 \right) G_n(x_2) dx_2.
\end{equation}

We claim that the functions $a_n$ have Schwartz seminorms decaying faster than any power of $n$. Indeed, thanks to \eqref{eq-9c}:
\begin{align}\label{eq-8b}
\big\| x_1^j  & \p_1^k a_n \big\|_{L^2}  = \left\| x_1^j  \p_1^k \int_\R A\left( \dfrac{x_1}{\sqrt{2n}},x_2 \right) G_n(x_2)dx_2 \right\|_{L^2}
\\
& \leq (3n)^{1+k} \left\|  \int_\R x_1^j  \p_1^k A(x) G_n(x_2)dx_2 \right\|_{L^2} = (2n)^{1+k} \left\|  \int_\R x_1^j  \p_1^k A(x) \dfrac{(-\p_2^2+x_2^2)^\ell G_n(x_2)}{(2n)^\ell} dx_2 \right\|_{L^2}
\\
& 
\leq (3n)^{1+k-\ell} \left\| \int_\R  (-\p_2^2+x_2^2)^\ell x_1^j  \p_1^k A(x) G_n(x_2) dx_2 \right\|_{L^2}  \leq (3n)^{1+k-\ell} \left\| \big(-\p_2^2+x_2^2\big)^\ell x_1^j  \p_1^k A \right\|_{L^2}.
\end{align}
In the second line we made the substitution $x_1 \mapsto \sqrt{2n} x_1$, then we used that the components of $G_n$ are eigenvectors of $-\p_2^2 + x_2^2$ with eigenvalues $2n+1, 2n-1$. In the third line we used that $-\p_2^2 + x_2^2$ is selfadjoint, then we used the Cauchy--Schwarz inequality, noting that $|G_n|_{L^2}=1$. By the Sobolev embedding theorem, \eqref{eq-8b} implies that $a_n \in \SSS(\R)$, with its Schwartz seminorms decaying faster than any power of $n$. 

In particular, $\ha_n$ is admissible of order $-\infty$, see \eqref{eq-9d}. 
Theorem \ref{thm:1h} produces then a leading-order description of $f_n$ in terms of $a_n$:
\begin{equation}\label{eq-10l}
f_n(t,x_1) - \frak{E}_{n,h_n} a_n(t,x_1) = O_{L^2(\R)} \big( h_n n^{-\infty}\big) =  O_{L^2(\R)} \big( h^{1/2} n^{-\infty}\big).
\end{equation}

To deduce \eqref{eq-9e}, it suffices to take tensor products of \eqref{eq-10l} with $G_{n,h}$ and to sum over $n$. Using that the functions $G_{n,h}$ form an orthonormal basis of $L^2(\R)$, we have:
\begin{equation}
F(t,x) = \sum_{n=1}^\infty f_n(t,x_1) G_{n,h}(x_2) = \sum_{n=1}^\infty \frak{E}_{n,h_n} a_n(t,x_1) G_{n,h}(x_2) + O_{L^2(\R)} \big( h^{1/2} \big) = \frak{E}A(t,x) + O_{L^2(\R)} \big( h^{1/2} \big).
\end{equation}

To deduce \eqref{eq-9ee}, we simply apply \eqref{eq-3aa}. This completes the proof. \end{proof}

\begin{proof}[Proof of Theorem \ref{thm:1i}] Theorem \ref{thm:1i} is now a quick consequence of \eqref{eq-9e}. We recall Cramer's inequality for Hermite functions: $|g_n|_{L^\infty} \leq \pi^{-1/4}$, see e.g. \cite{I61}; therefore $|G_{n,h}|_{L^\infty} \leq 2(\pi h)^{-1/4}$. We deduce that 
\begin{align}
\big| \frak{E} A(t,x) - \frak{E}_0 a_0(t,x_1) G_{0,h}(x_2) \big| & = \left| \sum_{n=1}^\infty \frak{E}_{n,h_n} a_n(t,x_1) G_{n,h}(x_2) \right| \\
& \leq 2(\pi h)^{-1/4} \sum_{n=1}^\infty \big| \frak{E}_{n,h_n} a_n(t,x_1) \big| \leq \dfrac{C}{h^{1/4} \sqrt{t}},
\end{align}
where we eventually used the uniform bound in \eqref{eq-8i} for $\frak{E}_{n,h_n}$. This completes the proof.
\end{proof}

\section{The case $\lambda$ constant}\label{sec:6}

When $\lambda$ is constant, then we can find explicitly the flow \eqref{eq-3q}:
\begin{equations}
x_t = x_0 + \lambda \dfrac{\xi }{\lr{\xi}}t, \qquad \xi_t = \xi_0.
\\
\Rightarrow \quad F(t,x,\xi) = x+ \lambda \dfrac{\xi }{\lr{\xi}}t, \qquad H(t,x,\xi) = x -  \lambda \dfrac{\xi }{\lr{\xi}}t.
\end{equations}
Hence the solution to the eikonal equation is explicit, given by 
\begin{equation}\label{eq-3b}
\varphi(t,x,\xi) = x\xi - \lambda \sqrt{1+\xi^2}.
\end{equation}
It follows that $\varphi$ is globally defined and satisfies the estimate \eqref{eq-3f} for all $t > 0$. Likewise, $\alpha, b_0, b_1, b$ are defined for all $t > 0$ (and admissible of order $-\infty$). It follows that the parametrices $\fE$ are defined for all $t > 0$ and satisfy 
\begin{equation}
\fE a(t,x) = 
\sum_{\pm} \left( e^{\pm i \varphi / \epsilon}  \sqrt{\pm \dfrac{i}{ 2\pi\p_\xi^2 \varphi }} \cdot b^\pm_{n,\epsilon}  b^\pm_{n,\epsilon}(0,0,\cdot)^\top \ha_n\right)\big( t,x, \pm \Xi(t,x) \big) + O_{L^2(\R)}\big(\epsilon^{1/2}n^{-\infty}\big).
\end{equation} 
Therefore, we deduce that Theorem \ref{thm:1a} holds for all fixed $t > 0$.

\begin{remark} This is valid for fixed time; we briefly discuss long-time alterations when $\mu$ is non-zero. One get sufficiently instruction by looking at the case $\lambda = 1, \mu = 1$, $s = 0$. A calculation shows that the function $\alpha$ from \eqref{eq-7g} depends on $\xi$:
\begin{equation}\label{eq-3g}
\alpha(t,x,\xi) = e^{ -2nit \lr{\xi} \xi }.
\end{equation}
When $t, \xi$ are bounded, we do not expect this difference to induce any qualitative change of behavior of the packet. As $t \rightarrow \infty$ however, \eqref{eq-3g} oscillates rapidly for fixed $\xi$; we therefore expect modified dispersive estimates in long time. We tracked down such effects when studying dynamical edge states in magnetic field, see \cite{BBD22}:  the wavepacket approximation there is valid for $T \sim h^{-1/8}$, smaller than $T \sim h^{-1/2}$ for the magnetic free case \cite{BB+21}. \end{remark}

\renewcommand{\thechapter}{D}

\chapter{Fourier reconstruction}\label{chap:4}

We now return to the equation $(hD_t + \Di) \Psi_t = 0$, where $\Di$ is no longer the model operator of Chapter \ref{chap:3} but instead is the Weyl quantization of a symbol $\di$ satisfying $\bm{(\AAA_1})-\bm{(\AAA_2)}$ (see \S\ref{sec-2.1}). We prove Theorem \ref{thm:1b}: semiclassically near $\Gamma_\di$, waves split in traveling and dispersive parts. 

We follow a standard strategy: we reduce the equation $(hD_t + \Di) \Psi_t = 0$ to the canonical model $(hD_t + \Dii) \Psi_t = 0$ studied in \S\ref{chap:3} using the Fourier integral operator $\FF$ constructed in Theorem \ref{thm:1d}. We then confirm that the traveling  / dispersive waves emerging in \S\ref{sec-5.3} keep their qualitative features under backward application of $\FF$.

\section{Preliminary estimates} Our program starts with some preliminary estimates that study how Fourier integral operators of the form \eqref{eq-} act on the quantities emerging in \eqref{eq-9ee}.

\subsection{Traveling wave}\label{sec:4.1} We study integrals of the form
\begin{equation}\label{eq-12n}
L(x) \de \int_{\R^4} e^{i\frac{\phi(x,\xi) - y\xi}{h}} c(x,\xi) \cdot \dfrac{1}{\sqrt{h}} q\left( \dfrac{y}{\sqrt{h}}\right)  \dfrac{dy d\xi}{(2\pi h)^2}, \qquad \text{where:}
\end{equation}
\begin{itemize}
\item $q$ is a Schwartz function on $\R^2$;
\item $c$ is a symbol bounded together with all its derivatives;
\item $\phi$ generates the generating function of the symplectomorphism $\kappa$ produced in Lemma \ref{lem:1q}. 
\end{itemize}
In particular, the equation $\p_\xi \phi(x,0) = 0$ has a unique solution, denoted $x_\star$ below. With $\xi_\star = \p_x \phi(x_\star, 0)$, the pair $(x_\star,\xi_\star)$ is uniquely characterized by the relation
\begin{equation}
\kappa(0,0) = (x_\star,\xi_\star).
\end{equation}
Our next lemma is (a version of) a standard result about action of Fourier integral operators on wavepackets. Essentially, \eqref{eq-12n} is a wavepacket concentrated at $(x_\star,\xi_\star)$ (see Definition \ref{def:1}).

\begin{lemma}\label{lem:1r} Under the above assumptions and notations, we have $L(x) = W(x)+O_{L^2}(\sqrt{h})$, where $W$ is a wavepacket concentrated at $(x_\star,\xi_\star)$. 
\end{lemma}

\begin{proof} 1. We first realize the integral in $y$ and obtain
\begin{equation}\label{eq-3h}
L(x) = \int_{\R^2} e^{i\frac{\phi(x,\xi)}{h}} c(x,\xi) \cdot \dfrac{1}{\sqrt{h}} \hq\left( \dfrac{\xi}{\sqrt{h}}\right)  \dfrac{d\xi}{(2\pi)^2h} = \dfrac{1}{\sqrt{h}} \int_{\R^2} e^{i\frac{\phi(x,\sqrt{h}\xi)}{h}} c(x,\sqrt{h}\xi) \cdot \hq(\xi) \dfrac{d\xi}{(2\pi)^2}.
\end{equation}
We then write a Taylor expansion of $c$ near $\xi=0$:
\begin{equations}\label{eq-3l}
c(x,\xi) = c(x,0) + \sum_{j=1}^2 \xi_j \cdot c_j(x,\xi), \qquad c_j \in C^\infty_b(\R^4).
\end{equations}
Plugging \eqref{eq-3l} in \eqref{eq-3h}, we end up with $L = L_1 + \sqrt{h} \cdot L_2$, where:
\begin{align}
L_1(x) & \de \dfrac{c(x,0)}{\sqrt{h}} \cdot \int_{\R^2} e^{i\frac{\phi(x,\sqrt{h}\xi)}{h}}  \cdot \hq(\xi) \dfrac{d\xi}{(2\pi)^2}.
\\
L_2(x) & \de \dfrac{1}{\sqrt{h}} \sum_{j=1}^2 \int_{\R^2} e^{i\frac{\phi(x,\sqrt{h}\xi)}{h}} c_j(x,\sqrt{h}\xi) \cdot \xi_j \hq(\xi) \dfrac{d\xi}{(2\pi)^2}.
\end{align}
 We observe that $L_2$ is the sum of two terms of the form \eqref{eq-3h}. Each of these term corresponds to a semiclassical Fourier integral operator applied to a function with fixed $L^2$-norm. We deduce that $L_2 = O_{L^2}(1)$, uniformly as $h \rightarrow 0$; hence $L = L_1 + O_{L^2}(\sqrt{h})$.
 
 2. We now write a Taylor expansion of $\phi$ near $\xi=0$:
\begin{equations}
\phi(x,\xi) = \phi(x,0) + \p_\xi \phi(x,0) \xi + \dfrac{1}{2} \lr{\xi, \p_\xi^2 \phi(x,0) \xi } +  \sum_{|\alpha|=3}\xi_\alpha \psi_\alpha(x,\xi).
\end{equations}
The functions $\psi_\alpha$ are smooth, compactly supported because $\phi$ is quadratic outside a compact set (see (i) from Lemma \ref{lem:1q}). As in \cite[\S7.7]{H90}, we define for $\te \in [0,1]$:
\begin{equations}
\phi_\te(x,\xi) = \phi(x,0) + \p_\xi \phi(x,0) \xi + \dfrac{1}{2} \lr{\xi, \p_\xi^2 \phi(x,0) \xi } + \te \cdot \sum_{|\alpha|=3}\xi_\alpha \psi_\alpha(x,\xi),
\\
L_\te(x) = \dfrac{c(x,0)}{\sqrt{h}} \cdot \int_{\R^2} e^{i\frac{\phi_\te(x,\sqrt{h}\xi)}{h}}  \cdot \hq(\xi) \dfrac{d\xi}{(2\pi)^2}.
\end{equations}
We then have $\phi_1 = \phi$ and therefore:
\begin{align}
L_1(x) & = L_0(x) + \int_0^1 \dfrac{dL_\te(x)}{d\te} d\te 
\\
& = L_0(x) + \sqrt{h} \cdot \dfrac{c(x,0)}{\sqrt{h}} \cdot \int_0^1 \sum_\alpha \left(\int_{\R^2} e^{i\frac{\phi_\te(x,\sqrt{h}\xi)}{h}} \psi_\alpha(x,\sqrt{h}\xi) \cdot \xi^\alpha\hq(\xi) \dfrac{d\xi}{(2\pi)^2}\right) d\te.
\end{align}
Again, the second term corresponds to a (integral over $\te$ of sums of) semiclassical Fourier integral operator applied to a normalized function. In particular, owing to the term $\sqrt{h}$ in front of it, we deduce that $L_1 = L_0 + O_{L^2}(\sqrt{h})$.

3. It remains to estimate the expression
\begin{equation}\label{eq-12s}
L_0(x) = \dfrac{c(x,0) e^{\frac{i}{h}\phi(x,0)}}{\sqrt{h}} \cdot \int_{\R^2} e^{i\xi \frac{\p_\xi \phi(x,0)}{\sqrt{h}}} \cdot e^{ \frac{i}{2} \lr{\xi, \p_\xi^2 \phi(x,0) \xi }}  \cdot \hq(\xi) \dfrac{d\xi}{(2\pi)^2}.
\end{equation}
We write the integral in \eqref{eq-12s} as
\begin{equations}\label{eq-12t}
\int_{\R^2} e^{i\xi \frac{\p_\xi \phi(x,0)}{\sqrt{h}}} \cdot \rho(x,\xi) \dfrac{d\xi}{(2\pi)^2} = \widecheck{\rho} \left(x, \frac{\p_\xi \phi(x,0)}{\sqrt{h}}\right), 
\\ \rho(x,\xi) \de e^{ \frac{i}{2} \lr{\xi, \p_\xi^2 \phi(x,0) \xi }}  \hq(\xi),
\qquad
\widecheck{\rho}(x,y) \de \int_{\R^2} e^{i\xi y} \cdot \rho(x,\xi) \dfrac{d\xi}{(2\pi)^2}.
\end{equations}
The function $\rho(x,\cdot)$ is Schwartz, with seminorms uniformly bounded for $x \in \R^2$: $\p_\xi^2 \phi$ is constant where $\phi$ is quadratic, hence outside a compact set. Hence so is its inverse Fourier transform. Recall that $x_\star$ is the solution to $\p_\xi \phi(x,0) = 0$ and Taylor-expand $\p_\xi \phi(x,0)$ near $x_\star$:
\begin{equation}
\p_\xi \phi(x,0) = \p_{x\xi}^2 \phi(x_\star,0)(x-x_\star) + \sum_{|\alpha| = 3} \Psi_\alpha(x) (x-x_\star)^\alpha, \qquad \Psi_\alpha \in C^\infty_0(\R^2).
\end{equation}
Using that $\p_{x\xi}^2 \phi(x_\star,0)$ is not singular, we can perform successive Taylor expansions of \eqref{eq-12t} near $x_\star$. We first obtain
\begin{align}
\widecheck{\rho} \left( x,\frac{\p_\xi \phi(x,0)}{\sqrt{h}}\right) & = \widecheck{\rho} \left( x,\p_{x\xi}^2 \phi(x_\star,0) \frac{x-x_\star}{\sqrt{h}}\right) + O_{L^2}(\sqrt{h})
\\
\label{eq-12u} & = \widecheck{\rho} \left( x_\star,\p_{x\xi}^2 \phi(x_\star,0) \frac{x-x_\star}{\sqrt{h}}\right) + O_{L^2}(\sqrt{h}),
\end{align}
We then go back to \eqref{eq-12s}. We recall that $\xi_\star = \p_x \phi(x_\star,0)$. By replacing $c(x,0)$ by $c(x_\star,0)$ and $\phi(x,0)$ by its Taylor expansion at order $2$, we deduce that
\begin{equation}\label{eq-12w}
L_0(x) = \dfrac{c(x_\star,0) e^{\frac{i}{h}\phi(x_\star,0) + \frac{i}{h} \xi_\star(x-x_\star)}}{\sqrt{h}}  \cdot e^{i \lr{z,\p_x^2 \phi(x_\star,0)z}} \widecheck{\rho} \left(x_\star, \p_{x\xi}^2 \phi(x_\star,0) z\right) \Bigg|_{z = \frac{x-x_\star}{\sqrt{h}}} + O_{L^2}(\sqrt{h}),
\end{equation}
 Hence, $L_0$ and $L$, are wavepackets concentrated at $(x_\star,\xi_\star)$ up to a remainder $O_{L^2}(\sqrt{h})$. \end{proof}

\subsection{Dispersive wave} Now we fix $n \geq 1$, $h > 0$ and we prove upper bounds on 
\begin{equation}\label{eq-9k}
B(x) \de \int_{\R^4} e^{i\frac{\phi(x,\xi) - y\xi}{h}} c(x,\xi) e^{i \vartheta_n(y_1) / \sqrt{h}}q(y_1) g_n\left( \dfrac{y_2}{\sqrt{h}}\right)  \dfrac{dy d\xi}{h^2}, \qquad \text{where:}
\end{equation}
\begin{itemize}
\item $q$ is admissible of order $-\infty$ (it particular it can depend implicitly on $n,h$) and has support in a bounded interval $I$;
\item $g_n$ is a Hermite function, see \S\ref{sec-H}; 
\item $\vartheta_n$ satisfies the bound
\begin{equation}\label{eq-11c}
\big| \p_{y_1}^j \vartheta_n(y_1) \big| \leq Cn^{1/2} \dist(y,\R \setminus I)^{-3j};
\end{equation}
\item $\phi$ generates the generating function of the symplectomorphism $\kappa$ produced in Lemma \ref{lem:1q}. 
\end{itemize}

\begin{lemma}\label{lem:1s} There exists $C$ independent of $n, h$ such that under the above assumptions,
\begin{equation}
\big|B(x)\big| \leq C n^{-\infty}, \qquad x \in \R^2.
\end{equation}
\end{lemma}

\begin{proof} 1. We first prepare the integral $B(x)$. Up to multiplying by a phase, we may assume that $\phi(x,0) = 0$. By translating $y_1$, we may assume that $\p_{y_1} \phi(x,0) = 0$. We now make the substitution $\xi \mapsto \sqrt{h} \xi$ and $y_2 \mapsto \sqrt{h} y_2$ so that 
\begin{equation}
B(x) = \int_{\R^4} e^{i\frac{\phi(x,\sqrt{h}\xi)}{h} + i\frac{\vartheta_n(y_1)-y_1\xi_1}{\sqrt{h}} - iy_2\xi_2} \cdot c(x,\sqrt{h} \xi) q(y_1) g_n(y_2)  \dfrac{dy d\xi}{\sqrt{h}}. 
\end{equation}
We now realize the integral over $y_2$ and obtain:
\begin{equation}\label{eq-9j}
B(x) = \int_{\R^3} e^{i\frac{\phi(x,\sqrt{h}\xi)}{h} + i\frac{\vartheta_n(y_1)-y_1\xi_1}  {\sqrt{h}} } c(x,\sqrt{h} \xi) q(y_1) \hg_n(\xi_2)  \dfrac{dy_1 d\xi}{\sqrt{h}}. 
\end{equation}
We comment that since the Fourier transform of the harmonic oscillator is itself, $\hg_n$ equals $g_n$ (up to a phase).

2. Fix $\ell \in \Z$; in Steps $2-4$ we assume that $q$ has support in the set $\big\{ |\vartheta_n'-\ell| \leq 1 \big\}$. Let $\chi \in C^\infty_0(\R)$, equal to $1$ on $[-3,3]$ and to $0$ outside $[-4,4]$. We write $B(x) = B_1(x) + B_2(x)$ with:
\begin{equations}\label{eq-11e}
B_1(x) \de \int_{\R^3} e^{i\frac{\phi(x,\sqrt{h}\xi)}{h} + i\frac{\vartheta_n(y_1)-y_1\xi_1}  {\sqrt{h}} } c(x,\sqrt{h} \xi) q(y_1) \chi(\xi_1-\ell) \hg_n(\xi_2)  \dfrac{dy_1 d\xi}{\sqrt{h}},
\\
B_2(x) \de \int_{\R^3} e^{i\frac{\phi(x,\sqrt{h}\xi)}{h} + i\frac{\vartheta_n(y_1)-y_1\xi_1}  {\sqrt{h}} } c(x,\sqrt{h} \xi) q(y_1) \big(1-\chi(\xi_1-\ell)\big) \hg_n(\xi_2)  \dfrac{dy_1 d\xi}{\sqrt{h}}.
\end{equations}

We focus on the integral in $y_1$ emerging in $B_2(x)$. For $(y_1,\xi_1)$ in the domain of integration, 
\begin{equation}
\big| \xi_1-\vartheta_n'(y_1) \big| \geq |\xi_1-\ell| - |\vartheta_n'(y_1)-\ell| \geq \dfrac{|\xi_1-\ell|}{2} + \dfrac{3}{2} - 1 \geq \dfrac{\lr{\xi_1-\ell}}{2}. 
\end{equation}
We integrate by parts twice according to the identity
\begin{equation}
\dfrac{-i\sqrt{h}}{\vartheta_n'(y_1)-\xi_1} \dd{}{y_1} e^{i\frac{  \vartheta_n(y_1)-\xi_1 y_1}{\sqrt{h}}}  = e^{i\frac{  \vartheta_n(y_1)-\xi_1 y_1}{\sqrt{h}}} .
\end{equation}
This yields
\begin{equations}\label{eq-9l}
\left| \int_\R e^{i\frac{\vartheta(y_1)-y_1\xi_1}  {\sqrt{h}} } q(y_1) \dfrac{dy_1}{\sqrt{h}} \right| = \sqrt{h} \left| \int_\R e^{i\frac{\vartheta(y_1)-y_1\xi_1}  {\sqrt{h}} } \left(\dd{}{y_1}\dfrac{1}{\vartheta'(y_1)-\xi_1}  \right)^2 q(y_1) dy_1 \right|
\\
\leq C \sqrt{nh} \sum_{j=0}^2 \lr{\xi_1-\ell}^{j-4} \int_\R \dist(y_1,I)^{-9} \big| \p_{y_1}^j q(y_1) \big| dy_1 \leq \dfrac{C\sqrt{nh}}{\lr{\xi_1-\ell}^{2}} \sum_{j=0}^2 \big\| \dist(y_1,I)^{-9} \p_{y_1}^j q \big\|_{L^1}.
\end{equations}
In the second line we used the bound \eqref{eq-11c} on derivatives of $\vartheta$. We note that since $q$ is smooth with support in $I$, so is $\dist(y_1,I)^{-9} \p_{y_1}^j q$. We then bound the quantity $B_2(x)$ defined in \eqref{eq-11e} using \eqref{eq-9l} and the integrability of $\lr{\xi_1}^{-2}$:
\begin{equations}\label{eq-9o}
\big|B_2(x)\big| \leq C \int_{\R^2} \big(1-\chi(\xi_1)\big) \big| \hg_n(\xi_2) \big| \cdot \left| \int_\R e^{i\frac{\vartheta(y_1)-y_1\xi_1}  {\sqrt{h}} } q(y_1) \dfrac{dy_1}{\sqrt{h}} \right| d\xi 
\\
\leq C\sqrt{nh} \big\| \hg_n \big\|_{L^1} \cdot \int_\R \dfrac{d\xi_1}{\lr{\xi_1-\ell}^{2}} \cdot \sum_{j=0}^2 \big\| \dist(y_1,I)^{-9} \p_{y_1}^j q \big\|_{L^1} \leq C \sqrt{n h} \big\| g_n \big\|_{L^1} \cdot \sum_{j=0}^2 \big\| \dist(y_1,I)^{-9} \p_{y_1}^j q \big\|_{L^1}.
\end{equations}

4. We now bound $B_1(x)$:
\begin{equation}\label{eq-9m}
\big|B_1(x)\big| \leq \| q \|_{L^\infty} \int_{\R^2}  \big| \hg_n(\xi_2)\big| \cdot\left| \int_\R e^{i\frac{\phi(x,\sqrt{h}\xi)}{h} - i\frac{y_1\xi_1}  {\sqrt{h}} } c(x,\sqrt{h} \xi) \chi(\xi_1-\ell)  d\xi_1 \right| \dfrac{dy_1 d\xi_2}{\sqrt{h}}.
\end{equation}
We focus on the integral over $\xi_1$. We Taylor-expand $\phi$ in $\xi_1$ near $0$: there exists a smooth function $\psi$, bounded together with all its derivatives, such that
\begin{equation}\label{eq-13c}
\phi(x,\xi) = \phi(x,0,\xi_2) + \p_{y_1} \phi(x,0,\xi_2) \xi_1 + \xi_1^2 \psi(x,\xi).
\end{equation}
We emphasize that $\psi$ is constant outside a compact set because $\phi$ is quadratic outside a compact set. From plugging \eqref{eq-13c} in \eqref{eq-9m}, we deduce
\begin{equations}\label{eq-9n}
\left| \int_\R e^{i\frac{\phi(x,\sqrt{h}\xi)}{h} - i\frac{y_1\xi_1}  {\sqrt{h}} }  c(x,\sqrt{h} \xi) \chi(\xi_1-\ell)  d\xi_1 \right| = 
\left| \int_\R e^{i\frac{\p_{y_1} \phi(0,\xi_2) - y_1}  {\sqrt{h}}\xi_1 }  \cdot   f_{h,x,\xi_2}(\xi_1) d\xi_1 \right| 
\\ = 
\left|\hf_{h,x,\xi_2}\left( \dfrac{y_1-\p_{y_1} \phi(0,\xi_2)}  {\sqrt{h}} \right)\right|,
\qquad
  f_{h,x,\xi_2}(\xi_1) \de c(x,\sqrt{h} \xi) e^{i\xi_1^2 \psi(\sqrt{h}\xi)}  \chi(\xi_1-\ell).
\end{equations}
Plugging \eqref{eq-9n} in \eqref{eq-9m}, we obtain
\begin{equation}
\big| B_1(x) \big| \leq \| q \|_{L^\infty} \int_{\R^2}  \big| \hg_n(\xi_2)\big| \cdot \left|\hf_{h,\xi_2}\left( \dfrac{y_1-\p_{y_1} \phi(0,\xi_2)}  {\sqrt{h}} \right)\right| \dfrac{dy_1 d\xi_2}{\sqrt{h}} \leq \| q \|_{L^\infty} \| \hg_n \|_{L^1} \cdot \sup_{\xi_2 \in \R} \big\| \hf_{h,\xi_2} \big\|_{L^1}.
\end{equation}
We control $\hf_{h,x,\xi_2}$ using Holder's $L^1-L^\infty$ bound and that the Fourier transform is bounded from $L^1$ to $L^\infty$:
\begin{equation}
\big\| \hf_{h,x,\xi_2} \big\|_{L^1} = \big\| \lr{y_1}^{-2} \lr{y_1}^2 \hf_{h,x,\xi_2} \big\|_{L^1} \leq \big\| (1+\p_{\xi_1}^2) f_{h,x,\xi_2} \big\|_{L^1} \leq C\lr{\ell}^4,
\end{equation}
where the constant $C$ is independent of $h, x, \xi_2$.
We recall that we assumed that the support of $q$ is contained in the set $\{ |\vartheta'-\ell| \leq 1\}$. Therefore, $\ell$ and $\vartheta'$ are comparable on the support of $q$; moreover $\lr{\vartheta'}$ is controlled by $\dist(y_1,I)^{-3}$, see \eqref{eq-11c}. We deduce that
\begin{equation}
\big| B_1(x) \big| \leq C \| q \|_{L^\infty} \| \hg_n \|_{L^1} \cdot \lr{\ell}^4 \leq C \big\| \dist(y_1,I)^{-12} q \big\|_{L^\infty} \| g_n \|_{L^1}.
\end{equation}
We conclude that whenever $\supp(q) \subset \{ |\vartheta'-\ell| \leq 1\}$,
\begin{equation}\label{eq-11d}
\big| B(x) \big| \leq C \sqrt{n} \| g_n \|_{L^1} \left( \big\| \dist(y_1,I)^{-12} q \big\|_{L^\infty} + \sum_{j=0}^2 \big\| \dist(y_1,I)^{-9} \p_{y_1}^j q  \big\|_{L^1} \right). 
\end{equation}

5. We now prove a bound on $\| g_n \|_{L^1}$. From Cauchy--Schwarz and integrability of $\lr{x_2}^{-2}$ over $\R$, we have:
\begin{equation}
\| g_n \|_{L^1} \leq \| \lr{x_2} g_n \|_{L^2} = \lr{(1+x_2^2) g_n,g_n} \leq \lr{1 - \p^2_2 + x_2^2) g_n,g_n} = 2n+2,
\end{equation}
where we used that $g_n$ is an eigenvector of $-\p_2^2+x_2^2$ for the eigenvalue $2n+1$. Hence, \eqref{eq-11d} becomes
\begin{equation}
\big| B(x) \big| \leq C n \left( \big\| \dist(y_1,I)^{-12} q \big\|_{L^\infty} + \sum_{j=0}^2 \big\| \dist(y_1,I)^{-9} \p_{y_1}^j q  \big\|_{L^1} \right). 
\end{equation}

6. To conclude, we must take off the restriction on the support of $q$. Let $\tchi \in C_0^\infty(\R)$ with support in $[-1,1]$ such that $\sum_{\ell \in \Z} \tchi(y_1-\ell) = 1$. We decompose $q$ as
\begin{equation}
q(y_1) = \sum_{\ell \in \Z} \tchi\big(\vartheta'(y_1)-\ell\big) q(y_1). 
\end{equation}
The terms in the sum satisfy the support condition of Steps 2-4. The additional price to pay is related to the growth of (at most two) derivatives of $\tchi\big(\vartheta'(y_1)-\ell\big)$, see \eqref{eq-11d}. These are controlled by $\vartheta''^2 + \vartheta'''$, hence by $n$ times $12$ additional moments of $\dist(y_1,I)^{-1}$. Summing these estimates produces 
\begin{equation}
\big| B(x) \big| \leq C n^2 \left( \big\| \dist(y_1,I)^{-12} q \big\|_{L^\infty} + \sum_{j=0}^2 \big\| \dist(y_1,I)^{-21} \p_{y_1}^j q  \big\|_{L^1} \right). 
\end{equation}
We can then bound these moments by sufficiently many derivatives of $q$, and use that $q$ is admissible of order $-\infty$ to conclude. 
 \end{proof}

\subsection{Pseudodifferential action} We recall that $g_{n,h}$ are semiclassical Hermite functions (see \S\ref{sec:2.2}) and $\frak{E}_n$ are the block parametrices defined in \eqref{eq-12g}.  

\begin{lemma}\label{lem:1t} Let $a$ be an admissible function of order $-\infty$; let $r = \OO(3) + \OO(1) O(h) + O(h^2)$. Then uniformly in $t \in (0,T)$, $h \in (0,1)$ and $n \in \N$:
\begin{equation}\label{eq-13a}
 r^w \big( \frak{E}_n a \otimes g_{n,h} \big) = O_{L^2} \big( h^{3/2} n^{-\infty} t^{-1/2} \big). 
 \end{equation}
\end{lemma}

\begin{proof} 1. First assume that $r = O(h^2)$. Then there is nothing to prove, because $r^w$ has $L^2$-norm controlled by $h^2$ and the $L^2$-norm of $\frak{E}_n a$ is controlled by $n^{-\infty} t^{-1/2}$, see \eqref{eq-13f}.

2. We now assume that $r = O(h) \OO(1)$. Then we can write
\begin{equations}\label{eq-13b}
r   = h (x_2 r_1 + \xi_1 r_2 + \xi_2 r_3) = h (x_2 \# r_1 + \xi_1 \#  r_2 + \xi_2 \# r_3) + O(h^2), 
\\
\qquad r^w   = h ( r_1^w \cdot h D_1  +  r_2^w \cdot x_2 +  r_3^w \cdot h D_2) + O(h^2)^w.
\end{equations}
Above, $\#$ denotes the operation induced at the level of symbols by the composition of two pseudodifferential operators. 
Again, for the purpose of proving estimates of the form \eqref{eq-13a}, we can ignore the term $O(h^2)^w$ in \eqref{eq-13b}. Moreover, from $2x_2 = \frak{a}+\frak{a^*}$ and $2 \p_2 = \frak{a}^*-\frak{a^*}$, we have the relations
\begin{equations}\label{eq-13i}
2x_2 g_{n,h} = h^{1/2} (\sqrt{2n+2} g_{n+1,h} + \sqrt{2n} g_{n-1,h}), 
\\ 
2h \p_2 g_{n,h} = h^{1/2} (\sqrt{2n} g_{n-1,h} +  \sqrt{2n+2} g_{n+1,h}),
\end{equations}
where we recall that $g_{-1} = 0$ by convention.
Therefore, using that $r_2^w$ is bounded on $L^2$, independently of $h$, and the bound \eqref{eq-13f} for the $L^2$-norm of $\frak{E}_n a$ when $n \geq 1$: 
\begin{equation}\label{eq-13l}
\big\| r_2^w \cdot x_2 \big( \frak{E}_n a \otimes  g_{n,h} \big) \big\|_{L^2} \leq C (nh)^{1/2} \sum_\pm \big\| \frak{E}_n a \otimes g_{n \pm 1} \big\|_{L^2} = O(h^{1/2} t^{-1/2} n^{-\infty}). 
\end{equation}
This is true also when $n=0$, as can be directly checked from the formula \eqref{eq-12g} for $\frak{E}_0$. Likewise, the same approach yields
\begin{equation}
\big\| r_3^w \cdot (h D_2) \big( \frak{E}_n a \otimes  g_{n,h} \big) \big\|_{L^2} \leq C (nh)^{1/2} \sum_\pm \big\| \frak{E}_n a \otimes g_{n \pm 1} \big\|_{L^2} = O(h^{1/2} t^{-1/2} n^{-\infty}). 
\end{equation}

We now focus on the term $r_1^w \cdot h D_1 \frak{E}_n a \otimes g_{n,h}$. If $n \geq 1$, using that $r_1^w$ is bounded on $L^2$:
\begin{equation}\label{eq-13h}
\big\| r_1^w \cdot \big( h D_1 \frak{E}_n a \otimes g_{n,h} \big) \big\|_{L^2} \leq C \big\| h D_1 \frak{E}_n a \big\|_{L^2}. 
\end{equation}
We then use the integral formula \eqref{eq-9a} for $\frak{E}_n a$. It gives:
\begin{align}
h D_1 \frak{E}_n a & = (2nh)^{1/2} \cdot \sum_\pm \dfrac{\pm 1}{2\pi h_n^{1/2}} \int_\R \p_{x_1} \varphi(t,x_1,\xi_1) e^{\pm i\varphi(t,x_1,\pm \xi_1) / h_n} \cdot b^\pm_{n,h_n}(t,x_1,\xi_1) b^\pm_{n,h_n}(0,0,\xi_1)^\top \ha(\xi_1) d\xi
\\
& + h \cdot \sum_\pm \dfrac{1}{2\pi h_n^{1/2}} \int_\R e^{\pm i\varphi(t,x_1,\pm \xi_1) / h_n} \cdot D_{x_1} b^\pm_{n,h_n}(t,x_1,\xi_1) b^\pm_{n,h_n}(0,0,\xi_1)^\top \ha(\xi_1) d\xi.
\end{align}
The above integrals are of the form \eqref{eq-8f}. Since they come with coefficients $h^{1/2}$ and $h$, we get 
\begin{equation}\label{eq-13g}
h D_1 \frak{E}_n a = O_{L^2} \big( h^{1/2} t^{-1/2} n^{-\infty}\big) .
\end{equation} 
This remains true when $n=0$, as can be directly verified from the formula \eqref{eq-12g}. Plugging \eqref{eq-13g} in \eqref{eq-13h}, and going back to \eqref{eq-13b}, we conclude that \eqref{eq-13a} holds whenever $r = O(h) \OO(1)$.

3. We now assume that $r = \OO(3)$. Then we can write
\begin{equations}
r = \sum_{\alpha_1+\alpha_2+\alpha_3 = 3} x_2^{\alpha_1} \xi_1^{\alpha_2} \xi_2^{\alpha_1} r_\alpha,
\\
r^w = \sum_{\alpha_1+\alpha_2+\alpha_3 = 3} r_\alpha^w \cdot x_2^{\alpha_1} (h D_1)^{\alpha_2} (h D_2)^{\alpha_3}  + \big( \OO(2) O(h) \big)^w + O(h^2)^w,
\end{equations}
where we used the expansion for the symbol of the composition of two pseudodifferential operators; see e.g. \cite[Theorem 4.12]{Z12}. Again, we can ignore the term $O(h^2)^w$. Symbols $\OO(2) O(h)$ are a fortiori $\OO(1) O(h)$, so therefore we can deal with them using Step 2. We now focus on the first term. Using the relations \eqref{eq-13i}, we can write
\begin{equation}
x_2^{\alpha_1} (h D_1)^{\alpha_2} (h D_2)^{\alpha_3} \big(\frak{E}_n a \otimes g_{n,h}\big) = \sum_{\alpha_2=0}^3  (h D_1)^{\alpha_2} \frak{E}_n a \cdot h^{\frac{3-\alpha_2}{2}} \sum_{j = \alpha_2-3}^{3-\alpha_2} c_{n,\alpha_2,j} g_{n \pm j,h},
\end{equation}
for some coefficients $c_{n,\alpha_2,j}$ that are $O(n^{3/2})$. Following the same path as that leading to \eqref{eq-13g}, we conclude that
\begin{equation}
\big\| r_\alpha^w \cdot x_2^{\alpha_1} (h D_1)^{\alpha_2} (h D_2)^{\alpha_3} \cdot \big( \frak{E}_n a \otimes g_{n,h} \big) \big\|_{L^2} \leq C n^{3/2} \sum_{\alpha_2=0}^3 h^{\frac{3-\alpha_2}{2}} \big\| (h D_1)^{\alpha_2} \frak{E}_n a \big\|_{L^2} \leq C h^{3/2} n^{-\infty} t^{-1/2}.
\end{equation}
This completes the proof. \end{proof}

\section{Proof of Theorem A.2}

\begin{proof}[Proof of Theorem \ref{thm:1b}] 1. Fix $(x_0,\xi_0) \in \Gamma$. According to Theorem \ref{thm:1d}, there exists:
\begin{itemize}
\item An operator $\FF$ of the form \eqref{eq-10c};
\item Three smooth functions $\lambda : \R \rightarrow \R^+$, $\mu, s : \R \rightarrow \R$;
\item A pseudodifferential operator $R$ with symbol $\OO(3) + \OO(1) O(h) + \OO(h^2)$ near $0$;
\end{itemize}
such that, with $\Dii$ of the form \eqref{eq-10b}:
\begin{equation}\label{eq-12v}
\Di \FF \  = \FF( \Dii + R ).
\end{equation}

2. Let $\psi_0$ be a wavepacket concentrated at $(x_0,\xi_0)$, see Definition \ref{def:1}. Up to a phase $e^{is_h}$,
\begin{equation}
\psi_0(x) = \dfrac{e^{\frac{i}{h} \xi_0(x-x_0)} }{\sqrt{h}} \cdot A\left( \dfrac{x-x_0}{\sqrt{h}} \right), \qquad A \in \SSS(\R^2).
\end{equation}
We define $\rho \in \SSS(\R^2)$ then $\tA \in \SSS(\R^2)$ through the implicit relations:
\begin{equations}
c(x_0,0) e^{\frac{i}{h} \phi(x_0,0)}  \cdot e^{i \lr{y,\p_x^2 \phi(x_\star,0)y}} \widecheck{\rho}_0 \left( \p_{x\xi}^2 \phi(x_\star,0) y\right) = A(y),
\\
\rho_0(\xi) = e^{\frac{i}{2} \lr{\xi, \p_\xi^2 \phi(x_0,0) \xi}} \widehat{\tA}(\xi), \qquad \tpsi_0(x) \de \dfrac{1}{\sqrt{h}} \tA\left( \dfrac{x}{\sqrt{h}} \right).
\end{equations}
Then, noting that $\FF \tpsi_0$ is an integral of the form \eqref{eq-12n}, and following Step 3 in the proof of Lemma \ref{lem:1r}, we obtain
\begin{equations}\label{eq-12a}
\FF \tpsi_0(x) = \dfrac{c(x_0,0) \cdot e^{\frac{i}{h}\phi(x_0,0) + \frac{i}{h} \xi_0 (x-x_0)}}{\sqrt{h}} \cdot  e^{\frac{i}{h} \lr{z,\p_x^2 \phi(x_0,0)z}} \widecheck{\rho}_0 \left( \p_{x\xi}^2 \phi(x_\star,0) z\right)  \Bigg|_{z = \frac{x-x_0}{\sqrt{h}}} + O_{L^2}(\sqrt{h})
\\
= \dfrac{e^{\frac{i}{h} \xi_0 (x-x_0)}}{\sqrt{h}} \cdot A\left(\dfrac{x-x_0}{\sqrt{h}}\right) + O_{L^2}(\sqrt{h}) = \psi_0(x) + O_{L^2}(\sqrt{h}).
\end{equations}

3. We now define $\tpsi_t(x) = (\frak{E}\tpsi_0)(t,x)$. Thanks to \eqref{eq-12v}, we have:
\begin{equation}
(h D_t + \Di) \big(\psi_t-\FF \tpsi_t\big) = - \FF (h D_t + \Dii + R) \tpsi_t. 
\end{equation}
We now apply Duhamel's formula as in \cite[Lemma 3.5]{BB+21}:
\begin{equation}\label{eq-12b}
\big\| \FF \tpsi_t-\psi_t\big\|_{L^2} \leq \big\| \FF \tpsi_0-\psi_0 \big\|_{L^2} + \dfrac{1}{h} \int_0^t \big\| \FF (h D_\tau + \Dii + R) \tpsi_\tau \big\|_{L^2} d\tau.
\end{equation}
To study the right-hand-side, we decompose $\tA = \sum_{n=0}^\infty \ta_n \otimes G_n$ according to \eqref{eq-12g}. Therefore, using that $\psi_\tau = \frak{E} \tA = \sum_{n=0}^\infty \frak{E}_n \ta_n \otimes G_{n, h}$ and the decomposition \eqref{eq-6n} of $\Dii$, we deduce that
\begin{equation}
(h D_\tau + \Dii) \tpsi_\tau = (h D_\tau + L) \frak{E}_0 \ta_0 \otimes G_{0,h} + \sum_{n=1}^\infty \sqrt{2nh} (h_n D_\tau + \Dii) \frak{E}_n \ta_n \otimes G_{n,h}.
\end{equation}
In particular, thanks to \eqref{eq-13q} and \eqref{eq-13n}:
\begin{align}
\big\| (h D_\tau + \Dii) \tpsi_\tau \big\|_{L^2} & \leq \big\| (h D_\tau + L) \frak{E}_0 \ta_0(\tau,\cdot) \big\|_{L^2} + \sum_{n=1}^\infty \sqrt{2nh} \big\| (h_n D_\tau + \Dii) \frak{E}_n \ta_n (\tau,\cdot) \big\|_{L^2} 
\\
& \label{eq-13p}
\leq  Ch^{3/2} + C h^{1/2} \tau^{-1/2} \sum_{n=1}^\infty  h_n^2 n^{1/2-\infty} = O(h^{3/2} \tau^{-1/2}).
\end{align}
Likewise, thanks to Lemma \ref{lem:1t}:
\begin{equation}\label{eq-13r}
\big\| R \tpsi_\tau \big\|_{L^2} \leq \sum_{n=0}^\infty \big\| R \big( \frak{E}_n \ta_n \otimes G_{n,h} \big) \big\|_{L^2} = O(\tau^{-1/2} h^{3/2}).
\end{equation}
As a semiclassical Fourier integral operator, $\FF$ is bounded (independently of $h$) on $L^2$. Therefore, plugging the bounds \eqref{eq-13p} and \eqref{eq-13r} in \eqref{eq-12b} and using that $\tau^{-1/2}$ is integrable near $0$, we conclude that uniformly in $t \in (0,T)$ and $h > 0$:
\begin{equation}\label{eq-13s}
\psi_t = \FF \tpsi_t + O_{L^2}(h^{1/2}).
\end{equation}

3. According to \eqref{eq-13s}, the solution $\psi_t$ to $(h D_t + \Di)\psi_t=0$ and the function $\FF \frak{E} \tpsi_0$ nearly coincide. To complete the proof of Theorem \ref{thm:1b}, it remains to understand $\FF \frak{E} \tpsi_0$. We have:
\begin{equation}\label{eq-13v}
\FF \frak{E} \tpsi_0 = \FF \frak{E}_0 \ta_0 \otimes G_{0,h} + \FF \sum_{n=1}^\infty \frak{E}_n \ta_n \otimes G_{n,h}.
\end{equation}

We recall that  $\frak{E}_0 \ta_0 \otimes G_{0,h}$ is a wavepacket concentrated at $(x_t^+,0)$, where $x_t^+ = e^{t\lambda(x_1) \p_{x_1}}(0)$, see \eqref{eq-7ua}. Therefore, according to Lemma \ref{lem:1r}, $\FF \frak{E}_0 \ta_0 \otimes G_{0,h}$ is a wavepacket concentrated at $(x_t,\xi_t) = \kappa(x_t^+,0)$. Because of Lemma \ref{lem:1m}, this point satisfies the ODE
\begin{equation}
\dfrac{d(x_t,\xi_t)}{dt} = d\kappa(x_t,\xi_t) \cdot (\dot{x_t^+},0) = d\kappa(x_t,\xi_t) \cdot \lambda(x_t^+) \p_{x_1} = V_\di(x_t,\xi_t).
\end{equation}
That is, $\frak{E}_0 \ta_0 \otimes G_{0,h}(t,\cdot)$ is concentrated along the integral curve of $V_\di$ starting at $(x_0,\xi_0)$.

We now focus on the second term in \eqref{eq-13v}. Fix $t \in (0,T)$ and recall that $h_n = \sqrt{h/2n}$. Thanks to \ref{eq-3aa}, we have, uniformly in $h \in (0,1), n \geq 1$:
\begin{align}
\frak{E}_n \ta_n \otimes G_{n,h}(t,y_1)
& =  \sum_{\pm} \left( e^{\pm i \varphi / h_n}  \sqrt{\dfrac{\pm i}{ 2\pi\p_\xi^2 \varphi }} \cdot b^\pm_{n,h_n}  b^\pm_{n,h_n}(0,0,\cdot)^\top \widehat{\ta}_n\right)\big( t,y_1, \pm \Xi(t,y_1) \big) \cdot G_{n,h} (y_2)  \\
 & \ \ \ \  + O_{L^2}\big(h^{1/4} n^{-\infty}\big) \label{eq-13d}
 \\
 &  =  \dfrac{1}{h^{1/4}}\sum_{\pm}  e^{ \pm i \vartheta_{n,\pm}(y_1)/\sqrt{h}} q_\pm(y_1) \cdot G_n \left( \dfrac{y_2}{\sqrt{h}} \right) + O_{L^2}\big(h^{1/4} n^{-\infty}\big),
\end{align} 
where we defined (with intentionally omitted time-dependence):
\begin{equations}
q_\pm(y_1) \de \left(\sqrt{\dfrac{\pm i}{ 2\pi\p_\xi^2 \varphi }} \cdot b^\pm_{n,h_n}  b^\pm_{n,h_n}(0,0,\cdot)^\top \widehat{\ta}_n\right)\big( t,y_1, \pm \Xi(t,y_1) \big),
\\
\vartheta_{n,\pm}(y_1) \de \sqrt{2n} \cdot \varphi\big( t,y_1, \pm \Xi(t,y_1) \big).
\end{equations}

Fourier integral operators are bounded on $L^2$, hence we can forget the remainder term in \eqref{eq-13d}. The function $q$ is smooth, admissible of order $-\infty$, supported in $I_t$ -- see the discussion after the proof of Corollary \ref{cor:5}. Moreover the function $\vartheta_{n,\pm}$ satisfies the estimates \eqref{eq-11c}, see Lemma \ref{lem:1w}. We can then apply Lemma \ref{lem:1s} and obtain that 
\begin{equations}
\sup_{x \in \R^2} \left| \FF \left( \sum_{\pm}  e^{ \pm i \vartheta_{n,\pm}(y_1)/\sqrt{h}} q_\pm(y_1) \cdot G_n \left( \dfrac{y_2}{\sqrt{h}} \right) \right)(x)\right| \leq C n^{-\infty}.  
\end{equations} 
Summing over $n$, we conclude from \eqref{eq-13d} that 
\begin{equation}
\sup_{x \in \R^2} \left| \FF \sum_{n=1}^\infty \frak{E}_n \ta_n \otimes G_{n,h}(t,x) \right| \leq C h^{-1/4}.
\end{equation}
This shows that the second term in \eqref{eq-13v} is $O_{L^\infty}(h^{-1/4}) + O_{L^2}(h^{1/4})$ for any fixed $t \in (0,T)$, hence completing the proof of \eqref{eq-13j}.

4. To complete the proof of Theorem \ref{thm:1b} it remains to discuss the case $\lambda_\di$ constant and $\Gamma_\di, \pi(\Gamma_\di)$ homeomorphic to $\R$. Fix $t > 0$ and set $L = \lambda_\di \cdot t$.

According to Remark \ref{rem:5}, we can construct a Fourier integral operator such that \eqref{eq-2u} holds with a remainder $R$ whose symbol is $\OO(3) + \OO(1) O(h) + \OO(h^2)$ on a neighborhood of $(-L -1 , L +1) \times \{0\}^3$.  Moreover, the estimate \eqref{eq-9ee} holds for all $t > 0$; in particular, $\tpsi_\tau$ js (up to $O_{L^2}(h^{1/4})$ remainder) concentrated in $I_t\times \{0\} = (-L, L) \times \{0\}$. As in Step 2, we deduce that $R \tpsi_\tau = O_{L^2}(\tau^{-1/2} h^{3/2})$. The Duhamel approach written above implies that $\psi_t = \FF \tpsi_t + O_{L^2}(h^{1/2})$. 
The rest of the argument applies with no change. This completes the proof. 
\end{proof}

\renewcommand{\thechapter}{E}

\chapter{Applications to topological insulators}\label{sec:D}

In this section we apply our microlocal approach to fundamental models in the field of topological insulators. The simplest one is the domain wall Dirac operator: we find the speed and orientation of propagating modes. The second example analyzes the effect of an external magnetic field. The third model is a natural extension of the latter in curved geometry. In the fourth application, we consider distortions of standard tight-binding models. By comparing our prediction with the second and third examples, we show that the propagation of modes in distorted lattice can be described in terms of pseudo metric and magnetic field.

\section{Domain wall Dirac operator}\label{sec:10}
We first consider the domain wall Dirac operator from \S\ref{sec:1}:
\begin{equation}\label{eq-13k}
\Di = hD_1 \sigma_1 + hD_2 \sigma_2 + m(x) \sigma_3,
\qquad
\di(x,\xi) = \xi_1 \sigma_1 + \xi_2 \sigma_2 + m(x) \sigma_3, 
\end{equation}
where $m \in C^\infty(\R^2,\R)$ is a smooth function that vanishes transversely along its zero set. Such operators emerge as effective Hamiltonians in graphene; see \cite{FLW16,LWZ19,D19,DW20} for rigorous derivations starting from continuous, translation-invariant systems. 
 
In this section, we justify that propagating modes for \eqref{eq-13k} travel along $m^{-1}(0)$ at speed one; and that the estimate \eqref{eq-13j} holds for all $t > 0$ when $m^{-1}(0)$ is homeomorphic to $\R$. This extends some results from \cite{BB+21}.

\subsection{Relevant quantities} We first compute the quantities $M_\di$, $\lambda_\di$, $H_\di$ and $V_\di$:
\begin{align}
M_\di(x,\xi) & = \dfrac{1}{2i} \{\di,\di\}(x,\xi) = \p_2 m (x) \sigma_1 - \p_1 m(x) \sigma_2,
\\
\lambda_\di(x,\xi) & = \left(\p_2 m(x)^2 + \p_1 m(x)^2 \right)^{1/4} = \big\| \nabla m(x) \big\|^{1/2},
\\
H_\di(x,\xi) & = \sigma_1 \dd{}{x_1} + \sigma_2 \dd{}{x_2} + \sigma_3 \left( \p_1 m(x) \dd{}{\xi_1} + \p_2 m(x) \dd{}{\xi_2} \right)
\\
V_\di(x,\xi) & = -\dfrac{\Tr(M_\di H_\di)}{2\lambda^2_\di} = \dfrac{1}{\| \nabla m(x) \|} \left( \p_1 m(x) \dd{}{x_2} - \p_2 m(x) \dd{}{x_1} \right) = \dfrac{\nabla m(x)^\perp}{\| \nabla m(x) \|}, \label{eq-0c}
\end{align}
where $u^\perp$ denotes the $\pi/2$ counterclockwise rotation of a vector $u \in \R^2$. 

We note that $V_\di$ is a unit vector tangent to $\Gamma_\di = m^{-1}(0) \times \{0\}^2$. In particular, Theorem \ref{thm:1b} predicts the existence of a coherent state traveling (in physical space) along $\Gamma_\di = m^{-1}(0) \times \{0\}^2$ at speed one, in the direction pointing to the left of the positive-topology region (see Figure \ref{fig:4}). Moreover, its orientation is given by the negative-energy eigenvector of $M_\di$: we find $[-1, e^{i\te(x)}]$ where $\te(x)$ is the angle between $\nabla m(x)$ and the $x_2$ axis. 

\subsection{Symplectomorphism} We construct explicitly the symplectomorphism of Theorem \ref{thm:1m} (and the associated Fourier integral operator), when $\di$ is the symbol \eqref{eq-13k}, $m^{-1}(0)$ is homeomorphic to $\R$ and $\|\nabla m(x)\| =1$ is constant along $m^{-1}(0)$.  We start with the quadratic form 
\begin{equation}\label{eq-10j}
q(x,\xi) = -\det \di(x,\xi) = \xi_1^2 + \xi_2^2 + m(x)^2.
\end{equation}
Introduce
\begin{equation}\label{eq-10h}
\tau(x) = \nabla m(x)^\perp, \qquad n(x) = \nabla m (x), \qquad \Omega(x) = \big[ \tau(x), n(x) \big].
\end{equation}
We note that $n$ and $\tau$ are normal and tangent to level sets of $m$, in particular to $m^{-1}(0)$; and that $\Omega(x)$ is an orthogonal matrix for every $x$. Fix $y_0 \in m^{-1}(0)$ and define $\Phi(x) = e^{x_2 n} \circ e^{x_1 \tau}(y_0)$. Since $m^{-1}(0)$ is homeomorphic to $\R$, $\Phi$ is a diffeomorphism from a tubular neighborhood of $\R \times \{0\}$ to a tubular neighborhood of $m^{-1}(0)$; we extended to a diffeomorhphism on $\R^2$;  We note that 
\begin{equation}
\dd{\Phi}{x_2}(x) = n \circ \Phi(x), \qquad 
\dd{\Phi}{x_1}(x_1,0) = \tau \circ \Phi(x_1,0), \qquad \dd{\Phi}{x_1}(x) = \tau \circ \Phi(x) + x_2 u(x),
\end{equation}
for some $u \in C^\infty(\R^2,\R^2)$. From \eqref{eq-10h} as well as $\tau \cdot \nabla m = 0$, $n \cdot \nabla m = 1$, we deduce that
\begin{equations}\label{eq-10k}
m \circ \phi(x) = m \circ \phi(x_1,0) + \int_0^{x_2} \dd{m \circ \Phi(x_1,y)}{x_2} dy =
\int_0^{x_1} \dd{m \circ \Phi(y,0)}{x_1} dy + \int_0^{x_2} \dd{m \circ \Phi(x_1,y)}{x_2} dy
\\ = \int_0^{x_1} (\tau \cdot \nabla m) \circ \Phi(x_1,0)  dy + \int_0^{x_2} (n \cdot \nabla m) \circ \Phi(x_1,y) dy 
= x_2 + O(x_2^2).
\end{equations}
Likewise, from \eqref{eq-10h} we obtain
\begin{equation}\label{eq-10i}
(\nabla \Phi^\top)^{-1} =  \left( \Omega \circ \Phi^\top + x_2 [u,0]^\top\right)^{-1} = \Omega \circ \Phi \left( \Id - x_2 [u,0]^\top \Omega \circ \Phi\right) + O(x_2^2).
\end{equation}

Define now 
\begin{equation}
\kappa(x) = \big(\Phi(x), (\nabla \Phi(x)^\top)^{-1} \xi \big). 
\end{equation}
This is a symplectomorphism. Using \eqref{eq-10j}, \eqref{eq-10i} and \eqref{eq-10k}, we note that
\begin{equations}
q \circ \kappa(x) = \big\| (\nabla \Phi(x)^\top)^{-1} \xi \big\|^2 +  m \circ \Phi(x)^2 
\\
=  \big\| \left( \Id - x_2 [u(x),0]^\top \Omega(x)\right) \xi  \big\|^2 +  x_2^2 + O(x_2^2 \xi^2, x_2 \xi_1^2, x_2 \xi_2^2, x_2^3)
\\
=  \xi_1^2 + \xi_2^2 + x_2^2  - 2 x_2 \xi_1 \lr{u(x), \Omega(x) \xi} + O(x_2^2 \xi^2, x_2 \xi_1^2, x_2 \xi_2^2, x_2^3).
\end{equations}
In particular, there are no cubic remainders in $\xi_1 \xi_2^2$ nor $\xi_1 x_2^2$. According to Remark \ref{rem:1}, $\mu = 0$. Therefore, the coefficient $\nu_t$ defined in \eqref{eq-7ub} vanishes. The dynamical edge state associated to \eqref{eq-13k} does not experience the weak dispersion discussed below Theorem \ref{thm:1i}. 

\subsection{Dispersive estimate} Another feature of the symplectomorphism $\kappa$ is that the induced Fourier integral operator is, at least locally, a simple change of coordinates:
\begin{equation}
\FF u(x) = \int_{\R^2} e^{\frac{i}{h}(\Phi(x) - y)\xi } u(y) \dfrac{dy}{(2\pi h)^2} = u \circ \Phi(x). 
\end{equation}
Hence, if we assume that $\| \nabla \kappa(x) \|$ is constant along $m^{-1}(0)$ then we are in the setup of \S\ref{sec:6}. The dispersive estimate \eqref{eq-7s} holds for all $t > 0$, see Corollary \ref{cor:2}. Since $L^\infty$ is invariant under change of coordinates, it transfers back to an estimate on $\Di$. This yields Theorem \ref{thm:1o}(iii).

\section{Magnetic Dirac operators}\label{sec:Da}  We couple here the operator \eqref{eq-13k} with a magnetic field $B = \p_1 A_1 - \p_2 A_1$, $A \in C^\infty_b(\R^2,\R^2)$:
\begin{align}
\Di & = \matrice{m(x) &  hD_1-A_1(x)-i (hD_2-A_2(x)) \\ hD_1-A_1(x) + i  (hD_2-A_2(x)) & - m(x)}  \\
& = \big(hD_1 - A_1(x)\big) \sigma_1 + (hD_2-A_2(x)\big)\sigma_2 +  m(x) \sigma_3, \\
\di(x,\xi) & = \big(\xi_1 - A_1(x)\big) \sigma_1 + (\xi_2-A_2(x)\big)\sigma_2 +  m(x) \sigma_3.\label{eq-10x}
\end{align}
As in \S\ref{sec:10}, we first compute the quantities $M_\di$, $\lambda_\di$, $H_\di$ and $V_\di$:
\begin{align}
 M_\di(x,\xi) &  = \p_2 m (x) \sigma_1 - \p_1 m(x) \sigma_2 - B(x) \sigma_3 = \matrice{B(x) & \p_2 m(x) + i \p_1 m(x) \\ \p_2 m(x) - i \p_1 m(x) & -B(x)},
\\
\lambda_\di(x,\xi) & = \left( \big\| \nabla m(x) \big\|^2 + B(x)^2 \right)^{1/4},
\\
H_\di(x,\xi) & = \sigma_1\dd{}{x_1}  + \sigma_2 \dd{}{x_2} + \sigma_3 \left( \p_1 m(x) \dd{}{\xi_1} + \p_2 m(x) \dd{}{\xi_2} \right) + O(\p_\xi),
\\
V_\di(x,\xi) & = \dfrac{\p_1 m(x) \p_{x_1} - \p_2 m(x) \p_{x_2}}{\sqrt{ \|\nabla m(x) \big\|2 + B(x)^2 }}  + O(\p_\xi)= \dfrac{\nabla m(x)^\perp}{\sqrt{ \|\nabla m(x) \big\|2 + B(x)^2 }} + O(\p_\xi),
\label{eq-0a}\end{align}
where $O(\p_\xi)$ denotes a vector field generated by $\p_{\xi_1}$ and $\p_{\xi_2}$. 

We deduce from \eqref{eq-0a} the value of the (physical) speed of the propagating mode:
\begin{equation}\label{eq-13m}
\dfrac{\| \nabla m(x) \big\|}{\sqrt{\| \nabla m(x)\|^2 + B(x)^2}}, \qquad x \in m^{-1}(0).  
\end{equation}
In particular, \eqref{eq-13m} is always smaller than $1$: magnetic fields introduced in \eqref{eq-13k} through minimal coupling may not speed up the propagation of energy.  We also obtain the orientation of the propagating mode: it is the negative-energy eigenvector of \eqref{eq-0a},
\begin{equation}
\matrice{-\cos(\vp(x)/2) \\ \sin(\vp(x)/2)e^{i\te(x)}}, \qquad \te(x) = \arg\big(\p_2 m(x) - i\p_1 m(x)\big), \qquad \vp(x) = \arg\big(B(x) + i \big \| \nabla m (x) \big\| \big).
\end{equation}
Obtaining these expressions by hand in \cite{BBD22} was a fairly technical procedure, which is now streamlined thanks to our microlocal framework.

\section{Dirac operators on surfaces}\label{sec:10.3}
In this section, we consider Dirac operators $\Di$ on a curved surface $M$. We find that edge states travel at unit speed (measured with respect to the Riemannian metric). Hence, the curvature does not affect the speed of propagation. We then investigate coupling effect with a magnetic field and derive a formula similar to \eqref{eq-13m}.

\subsection{Definitions} Let $M$ be a Riemannian surface. Topological considerations, see e.g. \cite{H56}, imply that $M$ admits a spin structure. In the present context, this means that there exists a complex vector bundle $\EE \rightarrow M$ of rank $2$, such that $T^*M \otimes \End(\EE) \rightarrow M$ admits a smooth section $\sigma$ with
\begin{equation}\label{eq-11f}
\sigma(x) \xi \cdot \sigma(x)\xi' + \sigma(x) \xi' \cdot \sigma(x)\xi = 2\lr{\xi,\xi'}_{g(x)} \cdot \Id_{\EE_x}, \qquad x \in M, \quad \xi, \xi' \in T_x^*M,  
\end{equation}
where $\sigma(x) \xi$ denotes the endomorphism $\sigma(x)(\xi,\cdot)$ of $\EE_x$.

Given $x \in M$ and $v \in T_x$, we let $v^{\perp_g} \in T_x M$ be the unique vector orthogonal to $v$, with same norm, such that $d\vol(v,v^{\perp_g}) > 0$, where $d\vol$ is the Riemannian volume form (in particular, for $f \in C^\infty(M,\R)$, $\nabla_g f^{\perp_g}$ is always tangent to the level sets of $f$). Using the isomorphism between $T_xM$ and $T_x^*M$, we can likewise define $\xi^{\perp_g} \in T_x^*M$ for $\xi \in T^*_xM$. We now set 
\begin{equation}\label{eq-11g}
\sigma_3(x) \de \dfrac{1}{i} \sigma(x) \xi \cdot \sigma(x) \xi^{\perp_g}, \qquad |\xi|_{g(x)}^2 = 1. 
\end{equation}
We note that the right-hand-side does not depend on $\xi$: indeed, if $\xi' \in T_xM$ has unit norm, then decomposing $\xi' = \cos \te \xi + \sin \te \xi^{\perp_g}$ we obtain 
\begin{equations}
\sigma(x) \xi' \cdot \sigma(x) \xi'^{\perp_g} = \sigma(x) \left( \cos \te \xi + \sin \te \xi^{\perp_g} \right) \cdot \sigma(x) \left( \cos \te \xi^{\perp_g} - \sin \te \xi \right)
\\
 = (\cos^2 \te +\sin^2 \te) \sigma(x) \xi \cdot \sigma(x) \xi^{\perp_g} + \sin \te \cos \te \left( (\sigma(x) \xi^{\perp_g})^2 - (\sigma (x) \xi)^2 \right)
 = \sigma(x) \xi \cdot \sigma(x) \xi^{\perp_g}. 
\end{equations}
We used that $\sigma(x) \xi$ and $\sigma(x) \xi^{\perp_g}$ anticommute, see \eqref{eq-11f}; and that $\xi, \xi^{\perp_g}$ have the same norm. 

Given a function $m \in C^\infty(M,\R)$, we define
\begin{equation}\label{eq-11i}
\p(x,\xi) = \sigma(x) \xi + m(x)\sigma_3(x),
\end{equation}
which is a smooth section of $\EE \rightarrow T^*M$. Let $\Di$ be a selfadjoint operator on acting on smooth sections of $\EE \rightarrow M$, with principal symbol $\di(x,\xi)$.

\subsection{Speed of propagating state} We compute the quantities $\{\di,\di\}$ and $H_\di$. Let $(x_1,x_2)$ be normal coordinates centered at a point of $m^{-1}(0)$; $(\xi_1,\xi_2)$ be induced local coordinates on the fibers of $T^*M$; and $(e_1,e_2)$ be a local frame on $\EE$. Noting that $\sigma(x) \xi$ is Hermitian and traceless, we may write (in these coordinates):
\begin{equations}
\sigma_1(x) \de \sigma(x) (1,0) = \sum_{k=1}^3 h_{1k}(x) \sigma_k, \qquad \sigma_2(x) \de \sigma(x) (0,1) = \sum_{k=1}^3 h_{2k}(x) \sigma_k, 
\end{equations}
for some smooth functions $h_{jk}$ defined near $0$. From \eqref{eq-11f} and using that we are working in normal coordinates,  
\begin{equations}\label{eq-11h}
\sum_{k=1}^3 h_{jk}(x) h_{\ell k}(x) = g^{j\ell}(x) = \delta_{j\ell} + O(x^2).
\end{equations}
Hence, $h_1(0)$ and $h_2(0)$, seen as vectors of $\R^3$, are orthonormal 
. Moreover, from \eqref{eq-11g} we have
\begin{equations}
 \sigma_3(x) = \dfrac{\sigma_1(x) \sigma_2(x)}{i |g^{11}(x) g^{22}(x)|^{1/2}} = \sum_{k,\ell=1}^3 h_{2k}(x) h_{1\ell}(x) \dfrac{\sigma_k \sigma_\ell}{i} + O(x^2) \\
 = \sum_{k,\ell,m=1}^3  \epsi_{k\ell m} h_{2k}(0) h_{1\ell}(0) \sigma_m + O(x^2) = \sum_{m=1}^3 \big( h_1(0) \times h_2(0)\big)_m \sigma_m + O(x^2).
 \end{equations}

Therefore, in the Pauli basis $(\sigma_1, \sigma_2, \sigma_3)$, the matrices $\sigma_1(0), \sigma_2(0), \sigma_3(0)$ have coordinates $h_1(0)$, $h_2(0)$ and $h_1(0) \times h_2(0)$; since $h_1(0)$ and $h_2(0)$ are orthonormal, $\sigma_1(0), \sigma_2(0), \sigma_3(0)$ form a direct orthonormal basis of traceless, Hermitian $2 \times 2$ matrices: their matrix in the Pauli basis is $[h_1,h_2, h_1 \times h_2]$. After conjugating $\di$ with an adequate matrix in $SU(2)$ (which does not affect the value of $\lambda_\di$ or $V_\di$), we can assume that $\sigma_j(0) = \sigma_j$. It follows that (in these coordinates, and up to conjugation):
\begin{align}
\di(x,\xi) & =  \xi_1 \sigma_1 + \sigma_2 \xi_2 + \big(\p_1 m(0) x_1 + \p_2 m(0) x_2\big) \sigma_3 + O(x,\xi)^2,
\\
M_\di(0) & = \p_2 m(0) \sigma_1 -\p_1 m(0) \sigma_2 ,
\\
\label{eq-0g} \lambda_\di(0)^4 & = \p_2 m(0)^2 + \p_1 m(0)^2 =  \| \nabla_g m(0) \|_{g(0)}^2,
\\
H_\di(0) & = \sigma_1 \dd{}{x_1} + \sigma_2 \dd{}{x_2} - \left(\p_1 m(0) \dd{}{\xi_1} + \p_2 m(0) \dd{}{\xi_2}\right) \sigma_3 ,
\\
V_\di(0) & = \dfrac{\p_1m(0)\p_{x_2} - \p_2m(0) \p_{x_1}}{\| \nabla_g m(0) \|_{g(0)}}  = \dfrac{\nabla_g m(0)^{\perp_g}}{\| \nabla_g m(0) \|_{g(0)}},
\end{align}
where we used that $g(0) = dx_1^2 +dx_2^2$ in normal coordinates centered at $0$. Going back to coordinate-free expression, and recalling that the origin of the normal coordinates system was an arbitrary point in $m^{-1}(0)$, we conclude that on $\Gamma = m^{-1}(0) \times \{0\}^2$:
\begin{equation}\label{eq-0z}
V_\di(x,0) = \dfrac{\nabla_g m(x)^{\perp_g}}{\| \nabla_g m (x)\|_g}, \qquad x \in m^{-1}(0).
\end{equation}

\begin{remark} The identity \eqref{eq-0z} implies that even in curved geometry, the speed of dynamical edge states for Dirac operators defined by semiclassical quantizations of \eqref{eq-11i} (measured with respect to the metric $g$) is independent of the choice of interface function $m$. 

We believe that a non-vanishing curvature may lead to a coefficient $\nu_t \neq 0$ in \eqref{eq-7ub}, in particular to dispersion of edge states as $t \rightarrow \infty$; we do not investigate this effect further.\end{remark}

\subsection{Coupling to magnetic fields} Turning on a magnetic $B \in C^\infty(M,\R)$ in \eqref{eq-11i} amounts to translating the momentum $\xi$ by a smooth one-form $A$ on $M$ such that $dA = B d\vol_g$; see e.g. \cite{BR20} for the Schr\"odinger operator. This produces the symbol
\begin{equation}\label{eq-0f}
\p(x,\xi) = \sigma(x) \big(\xi-A(x)\big) + m(x)\sigma_3(x).
\end{equation}
To compute the speed of a propagating state in curved geometry, with the magnetic field $B$, it suffices again to fix $x \in M$ and work in normal coordinates. Write (in this system of coordinates) $A(x) = A_1(x) dx_1 + A_2(x) dx_2$ and note that $d\vol_g(0) = dx_1 \wedge dx_2$; hence
\begin{equation}
B(0) = \p_1 A_2(0) - \p_2 A_1(0).
\end{equation}
In particular, \eqref{eq-0g} becomes:
\begin{align}
\di(x,\xi) & =  \sum_{j=1}^2 \big( \xi_j - A_1(0) - \p_{x_1} A_j(0) x_1 - \p_{x_2} A_j(0) x_2 \big)  \sigma_j + \big(\p_1 m(0) x_1 + \p_2 m(0) x_2\big) \sigma_3 + O(x,\xi)^2,
\\
M_\di(0) & = \p_2 m(0) \sigma_1 -\p_1 m(0) \sigma_2 + \big( \p_1 A_2(0)  - \p_2 A_1(0) \big)\sigma_3  = \p_2 m(0) \sigma_1 -\p_1 m(0) \sigma_2 + B(0) \sigma_3,
\\
\lambda_\di(0)^4 & = \p_2 m(0)^2 + \p_1 m(0)^2 + B(0)^2 =  \big\| \nabla_g m(0) \big\|_{g}^2 + B(0)^2,
\\
H_\di(0) & = \sigma_1 \dd{}{x_1} + \sigma_2 \dd{}{x_2}  + O(\p_\xi),
\\
V_\di(0) & = \dfrac{\p_1m(0)\p_{x_2} - \p_2m(0) \p_{x_1}}{\sqrt{\big\| \nabla_g m(0) \big\|_g^2  + B(0)^2}} + O(\p_\xi)  = \dfrac{\nabla_g m(0)^{\perp_g}}{\sqrt{\big\| \nabla_g m(0) \big\|_g^2  + B(0)^2}} + O(\p_\xi).
\end{align}
Varying the origin of the normal coordinate system, we obtain a formula analogous to \eqref{eq-13m}:
\begin{equation}\label{eq-0h}
d\pi(V_\di)(x) = \dfrac{\nabla_g m(x)^{\perp_g}}{\sqrt{\big\| \nabla_g m (x) \big\|_g^2 + B(x)^2}}, \qquad x \in m^{-1}(0).
\end{equation}

\section{Anisotropic Haldane models}

In this section, we study symbols of the form
\begin{equation}\label{eq-0b}
\di(x,\xi) = \matrice{m(x) & \overline{\alpha(x)} \cdot \big(\xi-\xi(x)\big) \\ \alpha(x)\cdot \big( \xi-\xi(x)\big) & -m(x)}   = \sum_{j=1}^2 \alpha_j(x) \cdot (\xi-\xi(x)) \sigma_j + m(x)\sigma_3
\end{equation}
where $\alpha = \alpha_1+i\alpha_2$, $\alpha_1, \alpha_2 \in C^\infty_b(\R^2,\R^2)$; $\xi \in C^\infty(\R^2,\R^2)$; and $m \in C^\infty_b(\R^2,\R)$. 

We first motivate the origin of \eqref{eq-0b} as an effective Hamiltonians for distorsions of standard models of condensed matter physics.  We then pull out a metric and a magnetic field such that \eqref{eq-0b} equals a symbol of the form \eqref{eq-0f}. By analogy with \S\ref{sec:10.3}, this allows us to produce a formula for the speed of propagating modes. At the physical level, our result implies that adiabatic variation in the position of atoms slows down electronic transport along interfaces between topological insulators.

\subsection{Tight-binding models} Our starting point is Wallace's Hamiltonian:
\begin{equation}\label{eq-13e}
(H_0\psi)_n = \matrice{\psi_n^B + \psi_{n-v_1}^B + \psi_{n-v_2}^B \\ \psi_n^A + \psi_{n+v_1}^A + \psi_{n+v_2}^A }, \qquad v_1 = -\dfrac{1}{2}\matrice{1 \\ \sqrt{3}}, \qquad v_2 = \dfrac{1}{2}\matrice{1 \\ -\sqrt{3}}, 
\end{equation}
with wavefunction $\psi \in \ell^2(\Z v_1 \oplus \Z v_2, \C^2)$. This Hamiltonian describes the motion of an electron on a honeycomb lattice, with hopping probabilities equally weighted on the nearest neighbors; see Figure \ref{fig:7}. Given the periodicity of $H$ with respect to the lattice $\Z v_1 \oplus \Z v_2$, it is natural to take its discrete Fourier transform:
\begin{equation}\label{eq-0n}
H_0(\xi) = \matrice{0 & \w(\xi) \\ \overline{\w(\xi)} & 0}, \qquad \w(\xi) = 1+e^{i\xi v_1} + e^{i\xi v_2}. 
\end{equation}
The eigenvalues of $H(\xi)$ are $\pm \big|\w(\xi)\big|$. In particular, they are degenerate precisely when $\xi = \pm \xi_\star \mod 2\pi$, where $\xi_\star = (2\pi/3,-2\pi/3)$. The crossing is conical: 
\begin{equation}
\w(\xi_\star + \xi) = \dfrac{\sqrt{3}}{4} (\xi_1 + i \xi_2) + O(\xi^2).
\end{equation}
In particular, the model \eqref{eq-13e} describes a semimetal. 

Small perturbations of semimetals generally transforms them in insulators. In a seminal work \cite{H88}, Haldane introduced a time-breaking Hermitian perturbation of Wallace's model that was the first model of topological insulator beyond the quantum Hall effect. Specifically:
\begin{equations}\label{eq-0m}
(H_m \psi)_n  = (H_0 \psi)_n + i m \matrice{\psi_{n+v_1}^A - \psi_{n-v_1}^A  - \psi_{n+v_2}^A + \psi_{n-v_2}^A + \psi_{n+v_2-v_1}^A - \psi_{n+v_1-v_2}^A \\ 
\psi_{n+v_1}^B - \psi_{n-v_1}^B  - \psi_{n+v_2}^B + \psi_{n-v_2}^B - \psi_{n+v_1-v_2}^B + \psi_{n-v_1+v_2}^B},
\\
H_m(\xi) = H_0(\xi) + 2im \beta(\xi) \sigma_3, \qquad \beta(\xi) = \sin (\xi v_1) - \sin(\xi v_2) - \sin\big(\xi(v_1-v_2)\big).
\end{equations}
In brief words, the states associated to the valence band (i.e. the negative energy eigenstates of $H_m(\xi)$, seen as functions on a $2$-torus), form a vector bundle with non-trivial topology: its Chern number is $\sgn(m)$. We also comment that $\beta(\xi_\star) = 3\sqrt{3}/2$.

\begin{center}
\begin{figure}[!t]
\floatbox[{\capbeside\thisfloatsetup{capbesideposition ={right,center}}}]{figure}[10cm]
{\hspace*{1cm}\caption{A honeycomb structure. An electron on a red site can hop to the nearest blue site (and conversely).}}
{\begin{tikzpicture}
\node at (0,0) {\includegraphics[scale=.7]{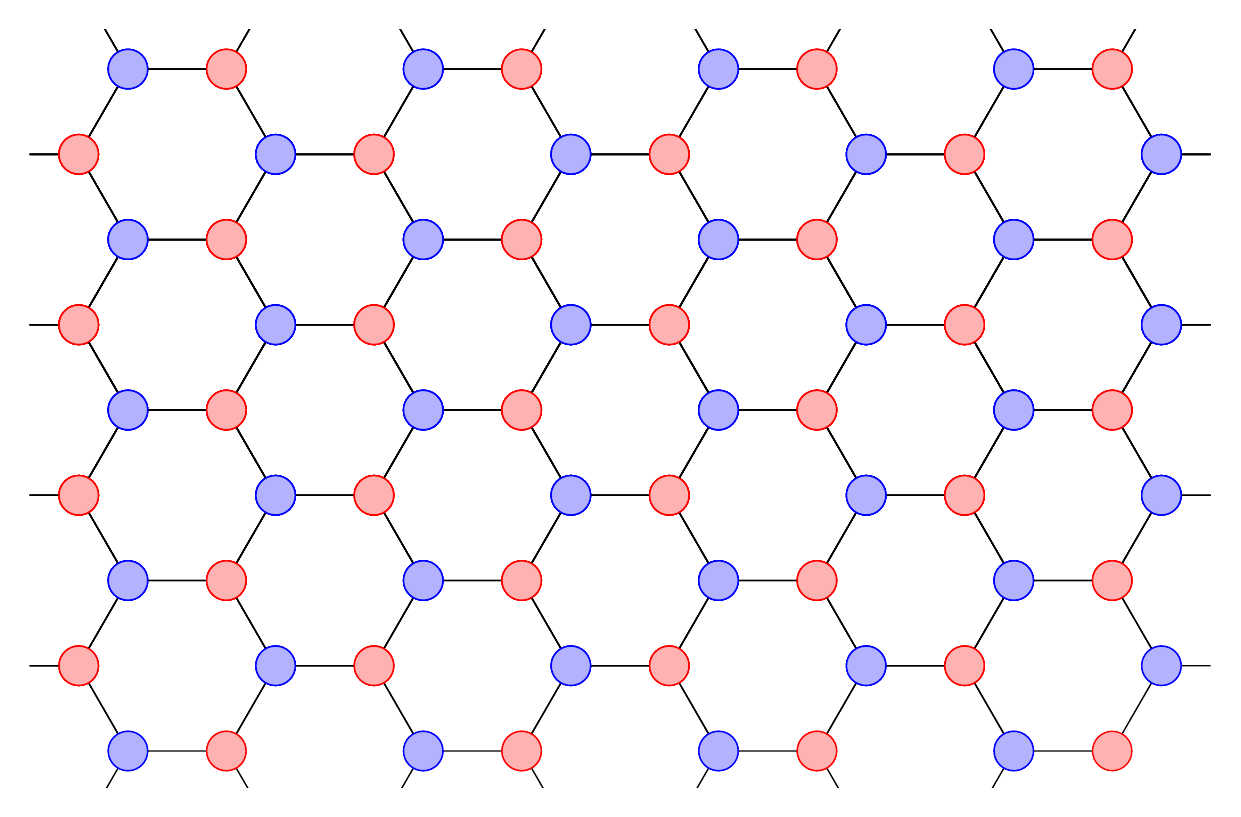}};
\end{tikzpicture}}\label{fig:7}
\end{figure}
\end{center}

The standard approach to study systems such as \eqref{eq-0m} near momentum $\xi_\star$ is to look for solutions of $(D_t + H_m) \psi = 0$ that take the form $\psi_n(t)= \Psi(ht,hn) e^{i\xi_\star n}$. This produces the effective equation
\begin{equation}
\left( hD_t + \dfrac{\sqrt{3}}{4}\matrice{m & hD_1-ihD_2 \\ hD_1+ihD_2 & -m} \right) \Psi = 0. 
\end{equation}
By slowly changing the sign of $m$ over space, we can generate structures whose topology is globally ill-defined but locally well defined when $m \neq 0$. The set $m^{-1}(0)$ acts as an interface between insulating phases of distinct Chern numbers. The resulting effective equation is: 
\begin{equation}\label{eq-0d}
\left( hD_t + \dfrac{\sqrt{3}}{4}\matrice{m(x) & hD_1-ihD_2 \\ hD_1+ihD_2 & -m(x)}\right) \Psi = 0,
\end{equation}
which is precisely that studied in \S\ref{sec:10}. According to Theorem \ref{thm:1b}, solutions to \eqref{eq-0d} split in two parts, one transporting in a direction prescribed by topology (specifically, the positive-topology region lies to the right of the direction of propagation); and one dispersing. 

While defects do not affect the insulating character of topological phases, they can influence the speed or profile of conduction along interfaces. To investigate these effects, we introduce anisotropy in the system. This corresponds to a Wallace model with the function
\begin{equation}
\w_a(\xi) = 1-a_1 - a_2 + a_1 e^{i\xi_1} + a_2 e^{i\xi_2}
\end{equation}
instead of $\w(\xi)$ in \eqref{eq-0n}; the parameter $a = (a_1,a_2) \in \R^2$ acounts for anistropic hopping probabilities $\frac{1-a_1-a+2}{3}$,$\frac{1+a_1}{3}$, $\frac{1+a_2}{3}$, see Figure \ref{fig:6}. For sufficiently small $a$, the equation $\w_a(\xi) = 0$ still has precisely two solutions $\pm (\xi_\star+\xi_a)$ modulo $2\pi$; this means that such perturbations preserve the semimetal character. Since we have
\begin{equation}
\w_a(\xi_\star + \xi) = \p_\xi \w_a(\xi_a) \cdot (\xi-\xi_a) + O(\xi-\xi_a)^2,
\end{equation}
an effective operator for the anisotropic Haldane model near momentum $\xi_\star-\xi_a$ is
\begin{equation}
\matrice{0 &  \ove{\p_\xi \w_a(\xi_a)} \cdot (D_x - \xi_a) \\ \w_a(\xi_a) \cdot (D_x - \xi_a) & 0}.
\end{equation}

\begin{center}
\begin{figure}[!t]
\floatbox[{\capbeside\thisfloatsetup{capbesideposition ={right,center}}}]{figure}[10cm]
{\hspace*{1cm}\caption{A distorted honeycomb structure. We drew the original hexagonal structure in the background (light gray).}}
{\begin{tikzpicture}
\node at (0,0) {\includegraphics[scale=.7]{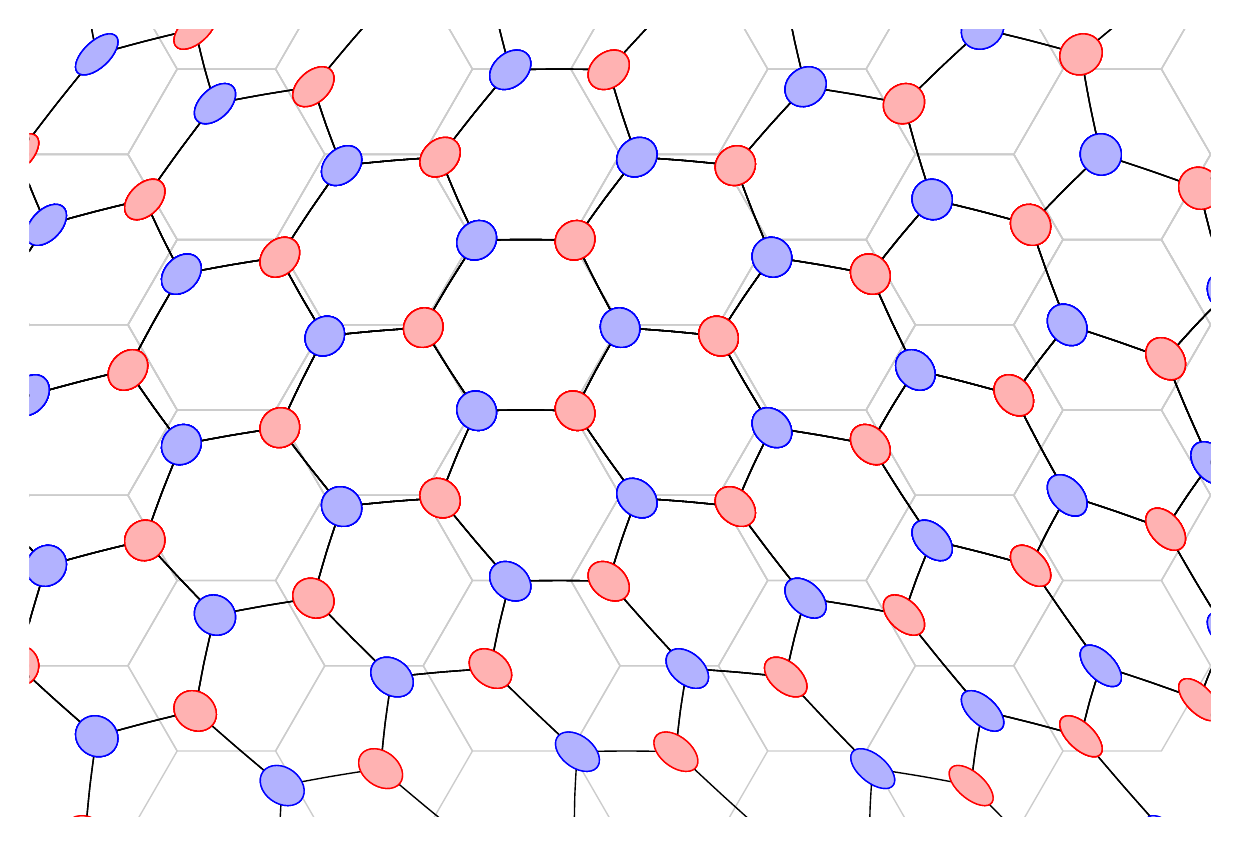}};
\end{tikzpicture}}\label{fig:6}
\end{figure}
\end{center}

Adding complex coupling as in \cite{H88}, and allowing the hopping probabilities and complex coupling to depend on $x$, we end up with an effective Hamiltonian $\Di = \di^w$ with symbol
\begin{equation}\label{eq-0e}
\di(x,\xi) = \matrice{m(x) & \alpha(x) \big(\xi-\xi(x)\big) \\ \alpha(x) \big( \xi-\xi(x)\big) & -m(x)}.  
\end{equation}
This Hamiltonian models a honeycomb lattice with microscopic unit cell, submitted to a macroscopic strain.

\subsection{Pseudo metric and magnetic field} Write 
\begin{equation}
\sigma(x) \xi = (\alpha_j(x) \cdot \xi) \sigma_j.
\end{equation} 
In analogy with \S\ref{sec:10.3}, we define a metric $g$ on $\R^2$ by (its action on the cotangent bundle $T^*\R^2$):
\begin{equation}
|\xi|^2_{g(x)}= \sum_{j=1}^2(\alpha_j(x) \cdot \xi)^2 = (\sigma(x) \xi)^2.
\end{equation}
Following \eqref{eq-11g}, we set
\begin{equation}
\sigma_3(x) = \dfrac{1}{i} \sigma(x) \xi \cdot \sigma(x) \xi^{\perp_g}, \qquad |\xi|_{g(x)} = 1, 
\end{equation}
which is independent of (adequately normalized) $\xi$. By plugging for $\xi$ a vector in a basis dual to $\alpha_1(x), \alpha_2(x)$, we end up with
\begin{equation}
\sigma_3(x) = \det\big[\alpha_1(x),\alpha_2(x)\big] \sigma_3.
\end{equation}
Hence, the symbol \eqref{eq-0e} has the form \eqref{eq-0f} for a domain wall $\tm = m / \det[\alpha_1, \alpha_2]$ and a magnetic potential $A(x) = \xi(x)$. From $dA = B d\vol_g$, we deduce that the effective magnetic field is
\begin{equation}
B(x) = \dfrac{\p_2 \xi_1(x) - \p_1 \xi_2(x)}{\det[\alpha_1(x),\alpha_2(x)]}.
\end{equation}

This implies that to a distorted Haldane model correspond an emerging metric and magnetic field. This relates to the numerical and experimental observation of pseudomagnetism in strained graphene \cite{MJ+13,GRW21}. Based on this analogy and on \eqref{eq-0h}, our microlocal framework produces a direct formula for the (physical) speed of the propagating state:
\begin{equations}
d\pi(V_\di)(x) = \dfrac{\nabla_g \tm(x)^{\perp_g}}{\sqrt{\big\| \nabla_g \tm (x)\big\|_g^2 + B(x)^2}}
= \dfrac{\nabla_g m(x)^{\perp_g}}{\sqrt{\big\| \nabla_g m (x)\big\|_g^2 + \| d\xi(x) \|_g^2}}, \qquad x \in m^{-1}(0),
\end{equations}
where we used the definition of $B$; and the relation between $m$ and $\tm$ together with $x \in m^{-1}(0)$. This implies that when appropriately measured, microscopic variations in atomic position can only slow down the propagation of edge states.

\end{document}